\RequirePackage{silence}

\documentclass[11pt]{article}
\usepackage[in]{fullpage}
\usepackage[normalem]{ulem}
\usepackage{parskip}
\usepackage{graphicx}

\usepackage{amsmath}
\usepackage{amsthm}
\usepackage{amssymb}
\usepackage{amsfonts}
\usepackage{bbm}
\usepackage{enumitem}
\usepackage[T1]{fontenc}
\usepackage{keytheorems}
\usepackage{mathtools}
\usepackage{mfirstuc}
\usepackage[table]{xcolor}
\usepackage{xspace}

\usepackage{alglib-tomer}
\makeatletter
\newcounter{mylabelcounter}
\newcommand{\mkname}[2]{%
\refstepcounter{mylabelcounter}%
\immediate\write\@auxout{%
  \string\newlabel{#2}{{1}{\thepage}{{\unexpanded{#1}}}{mylabelcounter.\number\value{mylabelcounter}}{}}%
}%
}
\makeatother

\newkeytheorem{LEMMA}      [parent=section, title=Lemma]
\newkeytheorem{theorem}    [sibling=LEMMA,  title=Theorem]
\newkeytheorem{observation}[sibling=LEMMA,  title=Observation]

\theoremstyle{definition}
\newkeytheorem{DEFINITION}[sibling=LEMMA, title=Definition]

\newcommand\reflemma[1]{\ref{lemma:#1}}
\newcommand\pagenumlemma[1]{\pageref{lemma:#1}}

\newcommand\refalg[1]{\ref{fig:alg:#1}}

\newcommand\defproc[2]{\mkname{\textsf{#2}}{procname:#1}}
\makeatletter
\newcommand\proc@starred[1]{\nameref*{procname:#1}\xspace}
\newcommand\proc@unstarred[1]{{\nameref{procname:#1}}\xspace}
\newcommand\proc{\@ifstar{\proc@starred}{\proc@unstarred}}
\makeatother

\newenvironment{proc-algo}[3][]{
\ifstrempty{#1}{\begin{algo}}{\begin{algo}[#1]}%
\procname{$\proc*{#2}(\ensuremath{#3})$}%
\label{fig:alg:#2}%
}{\end{algo}}
\newenvironment{lemma}[2][ ]{\begin{LEMMA}[note={#1}, store=lemma:#2, label=lemma:#2]}{\end{LEMMA}}
\newenvironment{definition}[2][ ]{\begin{DEFINITION}[note={#1}, store=def:#2, label=def:#2]}{\end{DEFINITION}}
\newcommand\repprovelemma[1]{\label{proofof:lemma:#1} \getkeytheorem{lemma:#1}}

\newcommand\algforprocshort[1]{The pseudocode for \proc{#1} is provided as Algorithm \ref{fig:alg:#1}.\xspace}
\newcolumntype{Y}{>{\centering\arraybackslash}X}

\newcommand\accept[0]{\textsc{accept}\xspace}
\newcommand\reject[0]{\textsc{reject}\xspace}
\newcommand\abs[1]{\left|{#1}\right|}
\newcommand\cond{\middle |}
\newcommand\floor[1]{{\left\lfloor{#1}\right\rfloor}}
\newcommand\ceil[1]{{\left\lceil{#1}\right\rceil}}
\newcommand\est{\mathrm{est}}

\newcommand\supp{\mathrm{supp}}
\newcommand\poly{\mathrm{poly}}

\newcommand{\Bin}{\mathrm{Bin}}
\newcommand{\Poi}{\mathrm{Poi}}

\newcommand\E{\mathop{{\rm E}\/}}
\newcommand\Var{\mathop{{\rm Var}\/}}

\newcommand\dtv{\ensuremath{d_\mathrm{TV}}}
\newcommand\kl{\ensuremath{\mathrm{KL}}}
\newcommand\dkl{\ensuremath{D_\kl}}
\DeclarePairedDelimiterX{\infdivx}[2]{(}{)}{%
  #1\;\delimsize\|\;#2%
}
\newcommand{\DKL}{\dkl\infdivx*}
\newcommand\eps{\varepsilon}
\usepackage{hyperref}
\hypersetup{
	colorlinks,
	citecolor=blue,
	filecolor=black,
	linkcolor=black,
	urlcolor=black,
	linktoc=all
}

\newcommand\algimplicitamp{\algcontract{Amplification}{implicit (at the cost of $O(\log \eta^{-1})$ penalty)}}
\newcommand\algmoments[1]{\algcontract{Moments}{#1}}
\newcommand\algunbiasedness[1]{\algcontract{Unbiasedness}{#1}}

\title{Instance-optimal estimation of $L_2$-norm}
\author{Tomer Adar\thanks{Technion - Israel Institute of Technology, Israel. Email: \href{mailto:tomer-adar@campus.technion.ac.il}{tomer-adar@campus.technion.ac.il}.}}

\begin{document}

\begin{titlepage}
    \maketitle
    \thispagestyle{empty}
    \pagestyle{empty}

    \begin{abstract}
        The $L_2$-norm, or collision norm, is a core entity in the analysis of distributions and probabilistic algorithms. Batu and Canonne (FOCS 2017) presented an extensive analysis of algorithmic aspects of the $L_2$-norm and its connection to uniformity testing. However, when it comes to estimating the $L_2$-norm itself, their algorithm is not always optimal compared to the instance-specific second-moment bounds, $O(1/(\eps\|\mu\|_2) + t_\mu/\eps^2)$, for $t_\mu = \|\mu\|_3^3 / \|\mu\|_2^4 - 1$, as stated by Batu (WoLA 2025, open problem session).

        In this paper, we present an unbiased $L_2$-estimation algorithm whose sample complexity matches the instance-specific second-moment analysis. Additionally, we show that $\Omega(1/(\eps \|\mu\|_2) + t_\mu / \eps^2)$ is indeed the per-instance lower bound for estimating the norm of a distribution $\mu$ by sampling (even for non-unbiased estimators).
    \end{abstract}

    \newpage
    \nonumber
    \tableofcontents
\end{titlepage}

\section{Introduction}

The collision norm $\|\mu\|_2^2$, which equals to the probability that two independent samples from a distribution $\mu$ are the same, is a fundamental entity in the analysis of distributions and probabilistic algorithms. For example, it lies in the algorithmic core of uniformity testing \cite{paninski2008coincidence}, as the uniform distribution has the lowest collision norm among the distributions over the same domain. Uniformity testing is used as an essential module in various testing algorithms, such as \cite{batuFFKRW2001,DK16,goldreich2020uniform}. A few additional perspectives of uniformity testing, as well as its generalization of testing identity to a known distribution, have been studied by \cite{vv17automatic, DKN15, ADK15, DGPP16}.

Instance-specific testing, in which the complexity bounds are analyzed for every input (more accurately, for every \emph{congruence class} of inputs), has been studied in \cite{valiant2011testing} for histogram testing and in \cite{VV13} for identity testing, in which we test an unknown input distribution for being equal to an explicit distribution $\mu$. For example, for identity testing, while the worst-case complexity is known to be $\Theta(\sqrt{N}/\eps^2)$ (both upper and lower bound) where $\mu$ is defined over $N$-element domain, the worst-case upper bound for an individual hard-coded $\mu$ (and an unknown input $\nu$) can be much lower, since we can use knowledge about the exact structure of $\mu$ to enhance the test.

In \cite{BC17} there is a study of uniformity testing of distributions over unknown, possibly infinite discrete domains. As stated by Batu\footnote{Tugkan Batu, eg. \cite{BC17}} in the open-problem session of WOLA 2025\footnote{Workshop on Local Algorithms at Toyota Technological Institute at Chicago (TTIC), August 18-20, 2025, Chicago, IL. \url{https://people.csail.mit.edu/joanne/WOLA25}}, even though second-moment analysis shows that $O(1/\eps \|\mu\|_2 + (\|\mu\|_3^3 - \|\mu\|_2^4)/\eps^2\|\mu\|_2^4)$ samples suffice for estimating the collision norm within a $(1 \pm \eps)$-multiplicative factor, there is no algorithm that guarantees this sample complexity (in expectation) for every given $\mu$ and $\eps > 0$. The reason is that without any prior knowledge about $\mu$, it is difficult to figure out this sufficient number of samples. In particular, estimating $(\|\mu\|_3^3 - \|\mu\|_2^4)/\|\mu\|_2^4$ within a $(1 \pm O(1))$-factor seems to be as hard as the initial task of estimating $\|\mu\|_2^2$ within a $(1 \pm \eps)$-factor.

Note that, when drawing an element $i$ according to $\mu$, the expected value of $\mu(i)$ is $\|\mu\|_2^2$ and the variance is $\|\mu\|_3^3 - \|\mu\|_2^4$. In other words, the term $(\|\mu\|_3^3 - \|\mu\|_2^4)/\|\mu\|_2^4$, denoted by $t_\mu$, corresponds to Chebyshev's ratio $\Var/\E^2$.

When drawing samples from an unknown distribution $\mu$ over an unknown domain, the algorithm can only base its decisions on the \emph{fingerprint}, also known as ``histogram of histograms''. In other words, it is a sequence $(a_1,\ldots)$ for which, for every $i \ge 1$, there are exactly $a_i$ elements that appear exactly $i$ times in the sample sequence. ($\sum_{i=1}^\infty i \cdot a_i$ is the number of samples). When the number of samples is not determined in advance, for every $m \ge 1$ the algorithm can consider the fingerprint of the first $m$ samples and decide whether to terminate or to draw the $m+1$st sample.

The following algorithm for estimating $\|\mu\|_2$ appears in \cite{BC17}: we choose some integer $k$, and then draw samples until we have at least $k$ collisions. The result is $k/\binom{M}{2}$, where $M$ is the number of samples. Considering the fingerprint $(a_1,\ldots)$, the number of collisions is $\sum_{i=2}^\infty \binom{i}{2} a_i$.

When fixing the number of samples to be $m$, the expected number of collisions is $m\|\mu\|_2^2$. If $\binom{m}{2} \approx k/\|\mu\|_2^2$, then this expected value is $\approx k$, and therefore, we expect to find the $k$th collision in a range ``around $m$''.

We would prefer a small choice of $k$, since the expected sample complexity of the algorithm till termination is $\Theta(\sqrt{k} / \|\mu\|_2)$. However, if $k$ is too small, then the ``likely termination range'' around $m$ would be too wide, since we have to consider the variance as well. For $m$ samples, the variance is close to $\frac{1}{2}m^2\|\mu\|_2^2 + m^3\left(\|\mu\|_3^3 - \|\mu\|_2^4\right)$.

When considering the worst-case, as \cite{BC17} do, we can use $k = C/\eps^4$. By Chebyshev's inequality, the likely termination range is $(1 \pm \eps)m_0$, where $m_0 = \Theta(1 / \eps^2 \|\mu\|_2)$ is the (non-integer) solution for $\binom{m_0}{2}\|\mu\|_2^2 = k$. Note that higher-order norms can be estimated using a similar algorithmic approach.

When considering \emph{instance-specific} sample complexity, we assume that the algorithm receives a $\Theta(1)$-approximation of $\|\mu\|_3^3 - \|\mu\|_2^4$ as an advice. This way, the knowledge about the variance, which is possibly much smaller than its worst-case, allows the algorithm to choose smaller values of $k$. Given this advice, we can use the worst-case algorithm to estimate $\|\mu\|_2^2$ within a $(1\pm 1/2)$-multiplicative error, and then choose $k = \min\{C_1 / \eps^2, C_2 (\|\mu\|_3^3 - \|\mu\|_2^4)^2 / \eps^4 \|\mu\|_2^6 \}$. This reduces the expected sample complexity to $O(1/\eps \|\mu\|_2 + (\|\mu\|_3^3 - \|\mu\|_2^4)/\eps^2 \|\mu|_2^4)$, which can be rewritten as $O(1 / \eps \|\mu\|_2 + t_\mu / \eps^2)$. % We denote the coefficient $\frac{\|\mu\|_3^3 - \|\mu\|_2^4}{\|\mu\|_2^4}$ by $t_\mu$ (short form: $t$), and rewrite the instance reference complexity as $O(1/\eps \|\mu\|_2 + t_\mu / \eps^2)$.

\begin{theorem}[note={Short form of Lemma \reflemma{estimate-L2-top-level}}]
    For every $0 < \eps \le 1$, there exists an unbiased estimator for $\|\mu\|_2^2$ whose input is a sampling access from a discrete distribution $\mu$ over an unknown domain such that:
    \begin{itemize}
        \item With probability at least $2/3$, the output is in the range $(1 \pm \eps)\|\mu\|_2^2$.
        \item The expected sample complexity is $O(1 / \eps \|\mu\|_2 + t_\mu / \eps^2)$, where $t_\mu = \|\mu\|_3^3 / \|\mu\|_2^4 - 1$.
    \end{itemize}
\end{theorem}

To complete our result, we show that the second-moment reference complexity $1/\eps \|\mu\|_2 + t_\mu / \eps^2$ is indeed a lower-bound for every $\mu$.

\begin{theorem}[note={Short form of Lemma \reflemma{lbnd-eps-mu2}}] \label{th:lbnd-eps-mu2}
    For every explicitly-given discrete distribution $\mu$, there exists a lower bound of $\Omega(1 / \eps \|\mu\|_2)$ for distinguishing between $\mu$ and distributions $\nu$ for which $\|\nu\|_2^2 \notin (1 \pm O(\eps))\|\mu\|_2^2$.
\end{theorem}

\begin{theorem}[note={Short form of Lemma \reflemma{lbnd-t-eps2}}] \label{th:lbnd-t-eps2}
    For every explicitly-given discrete distribution $\mu$, there exists a lower bound of $\Omega(t_\mu / \eps^2)$ for distinguishing between $\mu$ and distributions $\nu$ for which $\|\nu\|_2^2 \notin (1 \pm O(\eps))\|\mu\|_2^2$, where $t_\mu = \|\mu\|_3^3 / \|\mu\|_2^4 - 1$.
\end{theorem}

Combined, Theorems \ref{th:lbnd-eps-mu2} and \ref{th:lbnd-t-eps2} state that there is an instance-specific $\Omega(1 / \eps \|\mu\|_2 + t_\mu / \eps^2)$ lower-bound for estimating $\|\mu\|_2^2$ within $(1 \pm \eps)$-multiplicative error by sampling a given discrete distribution $\mu$ over an unknown domain.

\section{Preliminaries}
In this section we provide terms and notations that we use across the paper. A few standard terms and notations are fully defined in Appendix \ref{apx:std} to strict the formalism, and only mentioned in this section.

A \emph{discrete distribution} over a domain $\Omega$ is a function $\mu : \Omega \to [0,1]$ for which $\sum_{i\in\Omega} \mu(i) = 1$.

For every $k \ne 0$, the \emph{$L_k$-norm} of a distribution $\mu$ is defined as one of the following entities: $\|\mu\|_k^k = \sum_{i \in \Omega} (\mu(i))^k$ and $\|\mu\|_k = \left(\sum_{i \in \Omega} (\mu(i))^k\right)^{1/k}$. When relevant, we explicitly mention the entity we refer to.

\paragraph{Common distributions}
We refer to two common distributions: the binomial distribution $\Bin(n,p)$, defined over $\{0,\ldots,n\}$ as $\Pr[i] = \binom{n}{i} p^i (1-p)^{n-i}$, and Poisson distribution $\Poi(\lambda)$, defined over $\mathbb N$ (including zero) as $\Pr[i] = \frac{\lambda^i}{i!} e^{-\lambda}$.

\paragraph{Divergence measures} For the lower-bound analysis, we refer to a few divergence measures.
\begin{itemize}
    \item Total-variation distance (a metric): $\dtv(\mu, \tau) = \frac{1}{2}\sum_{i \in \Omega_\mu \cup \Omega_\tau} \abs{\mu(i) - \tau(i)}$.
    \item Kullback-Leibler divergence: $\DKL{\mu}{\tau} = \E_\mu[\log (\mu(i)/\tau(i))] = \sum_{i \in \Omega_\mu} \mu(i) \log (\mu(i)/\tau(i))$.
    \item $\chi^2$-divergence: $\chi^2(\mu,\tau) = \E_\mu[(\tau(i)/\mu(i) - 1)^2] = \sum_{i \in \Omega_\mu} \frac{(\tau(i)-\mu(i))^2}{\mu(i)}$.
\end{itemize}

\paragraph{Large-deviation bounds} Across the paper we use the well-known deviation bounds of Markov, Chebyshev and Chernoff, all stated explicitly in Appendix \ref{apx:std}.

\paragraph{Estimation behaviors} When estimating a value $\mathit{ans}$ using a random variable $X$, we consider the following behaviors (which the estimator may or may not have):
\begin{itemize}
    \item Non-negativeness: $X \ge 0$ with probability $1$.
    \item Correctness of estimation: with high probability, such as $\Omega(1)$ or $1-\eta$ (for an explicit error parameter $\eta$), the random estimation $X$ is considered accurate. For example, $\abs{X - \mathit{ans}} \le \eps$ or $\abs{X/\mathit{ans} - 1} \le \eps$ (for an explicit accuracy parameter $\eps$).
    \item Unbiasedness: the expected estimation $\E[X]$ equals to the estimated value $\mathit{ans}$.
    \item Moment preserving: for $r \in \mathbb R$, an estimator is said to \emph{preserve the $r$th moment} if $\E[X^r] = O(\mathit{ans}^r)$. An estimator is said to \emph{strongly preserve the $r$th moment} if for every $\ell \ge 0$ there exists $C_\ell > 0$ for which $\E[X^r (\log X^r)^\ell] = O(C_\ell \cdot \mathit{ans}^r \cdot (\log \mathit{ans}^r)^\ell)$.
\end{itemize}
Observe that, by Jensen's inequality, if a non-negative estimator preserves the $r$th moment, then it strongly preserves every moment $r'$ which is strictly between $0$ and $r$ (that is, $0 < r' < r$ or $r < r' < 0$, depending on the sign of $r$).

\section{Non-technical overview}

This paper is written in a top-down form. After the preliminaries, this non-technical overview and the technical overview, each section provides full details about its main statement (either an algorithmic module or a lower bound), as well as a full pseudo-code and verifiable proofs. Statements labeled as ``technical lemma'' are proved in the last subsection of the section in which they appear.

\paragraph{Upper bound}
The second-moment analysis results in an algorithm that takes an advice $s$ and uses $O(1/\eps\|\mu\|_2 + s/\eps^2)$ samples to estimate $\|\mu\|_2^2$, such that if $s \ge t_\mu = \|\mu\|_3^3/\|\mu\|_2^4 - 1$, then the estimation is in the range $(1 \pm \eps)\|\mu\|_2^2$ with high probability.

For having an instance-optimal sample complexity without any prior knowledge about $t_\mu$, we have to algorithmically obtain a random variable $s$, representing an advice, such that $s \ge t_\mu$ with high probability (for correctness) and $\E[s] = O(t_\mu + \eps / \|\mu\|_2)$ (for complexity). We use different strategies for estimating the third moment of $\mu$ depending on a rough estimation of the second moment.

If $\|\mu\|_2 = O(\eps)$, then it suffices to: (1) estimate $\|\mu\|_2$ within a $(1 \pm O(1))$-multiplicative factor, which is not a bottleneck, and (2) estimate $\|\mu\|_3^3$ within an $(1 \pm O(1))$-multiplicative factor or determine that $\|\mu\|_3^3 = O(\eps \|\mu\|_2^3)$. In the first case, we can obtain a fixed-factor estimation of $t_\mu+1 = \|\mu\|_3^3/\|\mu\|_2^4$, but since $\|\mu\|_2 = O(\eps)$, this is actually a fixed-factor estimation of $t_\mu + O(\eps/\|\mu\|_2)$. In the second case, we can determine that $t_\mu = O(\eps/\|\mu\|_2)$ and avoid a more accurate (and expensive) estimation.

If $\|\mu\|_2 = \Omega(\eps^3 \cdot \poly(\log (1/\eps)))$ and also $O(\eps^{3/5} / \poly(\log (1/\eps)))$, then we can estimate $t_\mu$ directly by its definition, by estimating $\|\mu\|_2^2$ and $\|\mu\|_3^3$ within a $(1 \pm O(\eps/\|\mu\|_2))$-multiplicative factor. The cost of this estimation is $O(1/\eps \|\mu\|_2)$ for this range of ``medium'' $\|\mu\|_2$s. For higher $\|\mu\|_2$ this resolution is too accurate (and expensive), and for lower $\|\mu\|_2$ we cannot take advantage of the $\Omega(1)$-accuracy.

Estimating larger $\|\mu\|_2$s is more intricate. In particular, it involves reduction to distributions defined over an explicit finite domain. If $\|\mu\|_2 = \Omega(\eps \cdot \poly(\log (1/\eps)))$, then we can algorithmically find a partition of the domain to a finite set $A$ for which $\mu_A$ is ``friendly'' (a term that we define more precisely below) and a set $B$ whose total mass is $O(t_\mu)$, such that $\abs{t_\mu - t_{\mu_A}} = O(\mu(B))$.

For distributions over a finite domain of $N$ elements, we can estimate $N\|\mu\|_2^2 - 1$ within a $(1 \pm O(1))$-multiplicative factor by iteratively looking for a resolution $\hat\eps$ for which $(1 - \hat\eps)\|\mu\|_2^2 > (1 +\hat\eps)/N$. If the sequence of resolutions starts with $\Omega(1)$ and decreases exponentially, then $\hat\eps = \Theta(N\|\mu\|_2^2 - 1)$ for the earliest $\hat\eps$ for which $(1 - \hat\eps)\|\mu\|_2^2 > (1 +\hat\eps)/N$. Note that this logic fails if $\|\mu\|_2^2 = (1 + \Omega(1))/N$, but this case is easy to detect and to resolve, since it implies that $N\|\mu\|_2^2 - 1 = \Theta(N\|\mu\|_2^2)$.

We consider a distribution over a finite domain of $N$ elements as \emph{friendly} if every element has probability at least $\Omega(1/N)$. If $\mu$ is friendly, then we can estimate $t_\mu$ by learning the mass of every element within a $(1 \pm O(1))$-multiplicative factor, at the cost of $O(N \log N)$ samples, and then embed them in an explicit formula that involves the (estimated) individual masses of $\mu$ and the (estimated) value of $N\|\mu\|_2^2 - 1$.

\paragraph{Lower bound}
While tight instance-specific identity testing bounds are already provided by \cite{VV13}, the collision norm of their hard-to-distinguish input distributions is not necessarily outside the range $(1 \pm \Omega(\eps))\|\mu\|_2^2$ with sufficiently high probability. Therefore, we use an ad-hoc construction based on their ideas.

For the $\Omega(1/\eps\|\mu\|_2)$ part of the lower bound, we recall a common construction for showing the $\Omega(\sqrt{N}/\eps^2)$ lower bound for uniformity testing over an explicitly given domain (for example, \cite{paninski2008coincidence}) and generalize it to non-uniform distributions as a per-instance $\Omega(1/\eps^2\|\mu\|_2)$ lower bound. We observe that the collision norm of the constructed distribution is usually in the range $(1 \pm \Theta(\eps^2))\|\mu\|_2^2$, and therefore, we use $\eps' \approx \sqrt{\eps}$ to obtain an $\Omega(1/\eps\|\mu\|_2)$-sample lower-bound for distinguishing $\mu$ from distributions whose collision norm is $(1 \pm \Theta(\eps))\|\mu\|_2^2$. Conceptually, for every element $i$, we construct a distribution with $\nu(i) \approx \mu(i) \pm \sqrt{\eps} \mu(i)$.

For the $\Omega(t_\mu / \eps^2)$ part of the lower bound, we use another construction in which the bias in the probability mass of $i$, instead of being linear in $\sqrt{\eps}$ and $\mu(i)$, is linear in $\eps$, in $1/t_\mu$ and in the deviation $\mu(i) / \|\mu\|_2^2 - 1$.

\section{Technical overview}
In this section we add more details to the non-technical overview. Note that some details only appear in the technical part.

\subsection{Elementary tools}

\paragraph{Amplification} The success probability of our procedures is a parameter $\eta$. For some procedures, the base success probability is $2/3$ and they use a generic amplification patch (which appears in the pseudocode as a preamble declaration rather than an actual code) to match the $1-\eta$ bound. For $\eta < 1/3$, amplification of the success probability from $1-1/3$ to $1-\eta$ is done by taking the median of $O(\log(1/\eta))$ independent estimations.

\begin{lemma}{amplify-1/3-to-eta}
    If $0 < \eta < 1$ and $q \ge 18 \ln \eta^{-1}$ is an integer, then $\Pr[\Bin(q,2/3) \le q/2] \le \eta$. In other words, if a sample is ``good'' with probability at least $2/3$, then the median of $q$ independent samples is ``good'' with probability at least $1 - \eta$.
\end{lemma}

\begin{lemma}{median-expected-value}
    Let $X$ be a non-negative random variable, and let $Y$ be the median of $k \ge 1$ independent variables distributing the same as $X$. If $k$ is even, then we use the ``lower median'' (the value of rank $\floor{k/2}$, where the minimum has rank $1$). In this setting, $\E[Y] \le C \cdot \E[X]$, for some constant $C$ independent of $X$ and $k$.
\end{lemma}

\paragraph{Other tools}
In addition to the amplification patch, we use a few standard tools:
\begin{itemize}
    \item An unbiased additive-error estimation of an indicator (Lemma \reflemma{estimate-indicator-additive}, Page \pagenumlemma{estimate-indicator-additive}).
    \item An effective rejection-sampling (Lemma \reflemma{rejection-sampling-concentration}, Page \pagenumlemma{rejection-sampling-concentration}).
    \item High-moment bounds through exponential tail (Lemma \reflemma{exponential-tail}, Page \pagenumlemma{exponential-tail}).
\end{itemize}

\subsection{Reference $\|\mu\|_2^2$-estimators}
We provide a few sub-optimal $\|\mu\|_2^2$-estimators. Each of these estimators has different guarantees about its behavior, even though all of them guarantee $(1 \pm \eps)$-factor accuracy with probability at least $1-\eta$ for a given parameter $0 < \eta \le 1/3$. Here we only overview the estimator behaviors by referring to their correctness lemmas. Section \ref{sec:reference-L2-estimators} provides more details.

In the following table, the ``L/ALG'' column describes the lemma number and the algorithm number.
\begin{table}[H]
    \centering
    \begin{tabular}{l|l|l|l|l}
         \textbf{Estimator} & \textbf{Complexity} & \textbf{Preserved moments} & \textbf{L/ALG} & \textbf{Comments} \\ \hline
         \proc{estimate-L2-BC} & $O(\log(1/\eta) / \eps^2 \|\mu\|_2)$ & $-\infty < r \le 0$ & L\reflemma{estimate-L2-BC}/A\refalg{estimate-L2-BC} & \\
         \proc{estimate-L2-base} & $O(1 / \eta \eps^2 \|\mu\|_2)$ & unbiased ($0\le r\le 1$) & L\reflemma{estimate-L2-base}/A\refalg{estimate-L2-base} & No advice \\
         \proc{estimate-L2-base} & $O\left(\frac{1}{\sqrt{\eta} \eps^2 \|\mu\|_2} + \frac{s}{\eta \eps^2}\right)$ & unbiased ($0 \le r\le 1$) & L\reflemma{estimate-L2-base}/A\refalg{estimate-L2-base} & Advice $s \ge t_\mu$ \\
         \proc{estimate-L2-moments} & $O(\log (1/\eta) / \eps^2 \|\mu\|_2)$ & $-\infty < r \le 1$ & L\reflemma{estimate-L2-moments}/A\refalg{estimate-L2-moments} &
    \end{tabular}
\end{table}

\subsection{$\|\mu\|_3^3$-estimators}
Some of our subroutines estimate $\|\mu\|_3^3$ as a part of their algorithmic logic. We provide two sub-optimal $\|\mu\|_3^3$-estimators, each one has different guarantees about its behavior. Here we only overview the estimator behaviors by referring to their correctness lemmas. Section \ref{sec:L3-estimators} provides more details.

% In the following table, the ``L/ALG'' column describes the lemma number and the algorithm number.
\begin{table}[H]
    \centering
    \begin{tabular}{l|l|l|l|l}
         \textbf{Estimator} & \textbf{Parameter} & \textbf{Error} & \textbf{Complexity} & \textbf{L/ALG} \\ \hline
         \proc{estimate-L3} & $\eps$ & $\pm \eps \|\mu\|_3^3$ & $O(1 / \eta \eps^2 \|\mu\|_2)$ & L\reflemma{estimate-L3}/A\refalg{estimate-L3} \\
         \proc{estimate-L3-magnitude} & $a$ & $\pm \max\{a^3, \|\mu\|_3^3/1000\}$ & $O(1 / \eta a)$ & L\reflemma{estimate-L3-magnitude}/A\refalg{estimate-L3-magnitude}
    \end{tabular}
\end{table}

\subsection{Top-level algorithm}
\label{sec:ext-abs:subsec:top-level-algorithm}
\defproc{estimate-L2-top-level}{Estimate-$L_2$-Top-Level}

Our core task for estimating $\|\mu\|_2^2$ is finding an advice $s \ge t_\mu$ where $t_\mu = \|\mu\|_3^3/\|\mu\|_2^4 - 1$, and then use the $L_2$ base-estimator. We analyze three (overlapping) cases: $\|\mu\|_2 = O(\eps)$ (``small''), $\|\mu\|_2 = O(\eps^{3/5}/\poly(\log(1/\eps))), \Omega(\eps \cdot \poly(\log(1/\eps)))$ (``medium'') and $\|\mu\|_2 = O(\eps \cdot \poly(\log(1/\eps)))$ (``large''). We overview each case in its own subsection.

Since the expected complexity of the procedure handling each case is suboptimal in other cases (unless the input belongs to an overlapping part), we have to figure out the relevant case with a very high probability. Distinguishing between the first case (small $\|\mu\|_2$) and the other cases requires $1-O(\eps^2\|\mu\|_2^2)$ success probability. Distinguishing between the medium and the large cases only requires $1-O(\eps^2)$.

Unfortunately, to test $\|\mu\|_2$ for being $O(\eps)$ versus $\Omega(\eps)$ with an $O(\eps^2 \|\mu\|_2^2)$-error using our reference estimators, we must draw either $O(\log (1/\eps\|\mu\|_2) / \|\mu\|_2)$ or $O(\log (1/\eps\|\mu\|_2) / \eps)$ samples, which are incompatible with the bi-criteria $O(1/\eps \|\mu\|_2)$ bound. Instead, we first use an $(1-O(\eps^2))$-error estimation to find a ``probable case'', and if we take the branch of a medium- or large-$\|\mu\|_2$ case, then we make another test to verify this choice with success probability $1-O(\|\mu\|_2^2)$.

\defproc{test-L2-magnitude}{Test-Small-$\|\mu\|_2$}
The ``second-thoughts'' test estimates $(1 \pm O(1))\|\mu\|_2^2$ with a constant success probability, and if the result estimation is smaller than $((3/2)\eps)^2$, then it uses additional $O(\log (1/\|\mu\|_2) / \eps)$ samples to make sure, with probability $1-O(\|\mu\|_2^2)$, that $\|\mu\|_2$ is indeed small. \algforprocshort{test-L2-magnitude}

\begin{lemma}{test-L2-magnitude}
    Procedure \proc*{test-L2-magnitude} correctly distinguishes between the case where $\|\mu\|_2 \le \eps$ (accept with probability $1 - \min\{O(\|\mu\|_2^2),\eta\}$) and the case where $\|\mu\|_2 \ge 2\eps$ (reject with probability $1-\eta$) at the expected cost of $O(\log (1/\eta) / \eps \|\mu\|_2)$ samples.
\end{lemma}

\algforprocshort{estimate-L2-top-level}

\begin{lemma}{estimate-L2-top-level}
    Procedure \proc*{estimate-L2-top-level} returns a number in the range $(1 \pm \eps)\|\mu\|_2^2$ with probability at least $1-\eta$. Moreover, the expected output is $\|\mu\|_2^2$ and the expected complexity is $O\left(\frac{1}{\eta} \left(\frac{1}{\eps \|\mu\|_2} + \frac{t_\mu}{\eps^2}\right)\right)$.
\end{lemma}

\subsection{Finding an advice when $\|\mu\|_2$ is small}
\label{sec:ext-abs:subsec:finding-advice-small-mu2}
\defproc{find-advice-small-mu2}{Find-Advice-Small-$\|\mu\|_2$}

We estimate $\|\mu\|_2^2$ within a fixed factor and then use a procedure \proc{estimate-L3-magnitude} to estimate $\|\mu\|_3^3$ within a $\max\{\|\mu\|_3^3/1000,\eps\|\mu\|_2^3\}$ additive error. If $\|\mu\|_3^3$ is small, then the $\eps\|\mu\|_2^3$-part dominates the estimation error and we obtain that $t_\mu = O(\eps/\|\mu\|_2)$. If $\|\mu\|_3^3$ is large, then the multiplicative error is dominant and we can use $t_\mu + 1 = \Theta(\|\mu\|_3^3 / \|\mu\|_2^4)$. In the case where $\|\mu\|_2 = O(\eps)$, we obtain that $t_\mu = \Theta(t_\mu + 1)$. \algforprocshort{find-advice-small-mu2}

\begin{lemma}{find-advice-small-mu2}
    Let $X$ be the random output of $\proc*{find-advice-small-mu2}(\eta; \mu, \eps)$.
    \begin{itemize}
        \item With probability at least $1-\eta$, $X \ge t_\mu$.
        \item $\E[X] = O(t_\mu + \eps/\|\mu\|_2 + 1)$.
        \item The sample complexity is $O(1/\eta \eps^{1/3} \|\mu\|_2)$.
    \end{itemize}
\end{lemma}

\subsection{Finding an advice when $\|\mu\|_2$ is medium}
\label{sec:ext-abs:subsec:finding-advice-medium-mu2}
\defproc{find-advice-medium-mu2}{Find-Advice-Medium-$\|\mu\|_2$}

Intuitively, we would wish to estimate $t_\mu$ directly by its definition, $\|\mu\|_3^3 / \|\mu\|_2^4 - 1$, with additive error $\pm O(\eps/\|\mu\|_2)$. \algforprocshort{estimate-t-directly}

Estimating $t$ directly by its definition requires, at least, an $(1 \pm O(\eps/\|\mu\|_2))$-multiplicative estimation of $\|\mu\|_3^3$. There are two ranges for which this estimation is too expensive:
\begin{itemize}
    \item If $\|\mu\|_2 = O(\eps)$, then our analysis of the $\|\mu\|_3$-estimator cannot take advantage of the $\Omega(1)$-accuracy. However, this exceeds the $O(1/\eps \|\mu\|_2)$ bound only for $\|\mu\|_2 = O(\eps^3)$.
    \item If $\|\mu\|_2 = \Omega(\eps^{3/5})$, then $O(\eps/\|\mu\|_2)$ is too accurate for a budget of $O(1/\eps\|\mu\|_2)$ samples.
\end{itemize}
The above ranges define the medium-$\|\mu\|_2$ case as the range of $\|\mu\|_2$s for which direct estimation of $t$ is effective. The extra polylogarithmic factors come from the need of amplification. \algforprocshort{find-advice-medium-mu2}

\begin{lemma}{find-advice-medium-mu2}
    Let $X$ be the output of $\proc*{find-advice-medium-mu2}(\eta; \mu,\eps)$. For every $0 < \eta \le 1/3$, a discrete distribution $\mu$ and $0 < \eps \le 1$,
    \begin{itemize}
        \item With probability at least $1 - \eta$, $X \ge t_\mu$.
        \item The expected output is $O(t_\mu + \eps/\|\mu\|_2)$.
        \item The expected sample complexity is $O(\log \frac{1}{\eta \eps} \cdot (\|\mu\|_2^{2/3} / \eps^2 + 1/ \|\mu\|_2^{4/3}))$.
    \end{itemize}
\end{lemma}

\subsection{Distributions over a finite domain}
In this subsection and in the next one we focus on distributions $\mu$ over an explicitly given finite domain of $N$ elements. Without loss of generality, we assume that $\Omega = \{1,\ldots,N\}$. For every element $i \in \Omega$, we define $\delta_i = N\mu(i) - 1$, so that $\mu(i) = (1 + \delta_i) / N$.

\paragraph{Algebraic behavior of $t$}
In the finite-domain setting, we use an algebraic approach.

\begin{lemma}{mu22-explicit-by-deltas}
    $\|\mu\|_2^2 = \frac{1}{N}\left(1 + \frac{1}{N}\sum_{i=1}^N \delta_i^2\right)$.
\end{lemma}

\begin{lemma}{t-explicit-by-deltas}
    $t_\mu = \frac{\|\mu\|_3^3}{\|\mu\|_2^4} - 1 = \frac{\frac{1}{N}\left(\sum_{i=1}^N \delta_i^2 + \sum_{i=1}^N \delta_i^3 - \frac{1}{N}\left(\sum_{i=1}^N \delta_i^2\right)^2\right)}{\left(1 + \frac{1}{N}\sum_{i=1}^N \delta_i^2\right)^2}$.
\end{lemma}

\paragraph{Estimating the sum of squares}
\defproc{estimate-sum-squares}{Estimate-Sum-Squares}

We would wish to estimate each $\delta_i$ separately and then aggregate the results to obtain an estimation for $\sum_{i=1}^N \delta_i^2$. However, the sum of squares can be too small to effectively estimate by learning individual elements. Instead, we use Lemma \reflemma{mu22-explicit-by-deltas}, which states that $N \|\mu\|_2^2 - 1 = \frac{1}{N}\sum_{i=1}^N \delta_i^2$.

Our sum-of-squares algorithm uses an initial estimation of $\|\mu\|_2^2$ and refines it until it can distinguish between $1/N$ and $\|\mu\|_2^2$. If $\eps'$ is the largest accuracy magnitude for which $(1+\eps')/N \not\approx (1-\eps')\|\mu\|_2^2$, then $\eps' = \Theta(N\|\mu\|_2^2 - 1) = \Theta(\frac{1}{N}\sum_{i=1}^N \delta_i^2)$.

%More precisely, we use a decreasing sequence of $\eps$s, with the last element bounded by the given $\eps$. The algorithm initially estimates $p_0 \in (1 \pm 1/12)\|\mu\|_2^2$ using the base estimator and then repeatedly refines its estimation until being able to separate between $\|\mu\|_2^2$ from $1/N$ (or reaching a $(1 \pm O(\eps))$-factor, followed by a report that $t$ is too small). 
Our core idea is tracking an additional decreasing sequence $t_i$ of upper bounds for $t_\mu$, using the knowledge about $\|\mu\|_2^2$ obtained in past iterations.
%
%For $0 \le i \le k = \ceil{\log_2 \eps^{-1}}$, we define $\eps_i = 2^{-i}$. Additionally, for every $1 \le i \le k$, we define $t_i = \eps_{i-1} \sqrt{N}$. 
\algforprocshort{estimate-sum-squares}

\paragraph{Estimating the sum of cubes}
Estimating the sum of cubes is not natural, since it can be negative. Even if it is positive, it can be asymptotically smaller than the sum of squares, which is intricate enough to require an iterative estimation logic. Instead, we estimate only the sum of large cubes, $\sum_{i : \delta_i \ge 1} \delta_i^3$, with an additive error that may depend on the sum of squares. To do that, we learn the input distribution $\mu$ using $\tilde{O}(N)$ samples to estimate each $\delta_i$ individually and return the sum of the sufficiently-large cubes. \algforprocshort{estimate-sum-cubes}

\subsection{Friendly distributions}

We define a class of \emph{friendly} distributions.

\begin{definition}[Friendly distribution]{friendly-distribution}
    A discrete distribution $\mu$ is \emph{friendly} if:
    \begin{itemize}
        \item It is defined over a finite domain $\Omega$.
        \item All elements in $\mu$ have probability at least $7/(13\abs{\Omega})$.
    \end{itemize}
\end{definition}

Note that instead of $7/13$ we could use every constant strictly greater than $1/2$. The lack of rare elements implies a lower bound for $t_\mu$ based on the sum of squares, the sum of large cubes and the magnitude of $\|\mu\|_2^2$.

\begin{lemma}{t-lbnd-by-sum-squares-sum-cubes}
    Let $\mu$ be a friendly distribution over $\Omega = \{1,\ldots,N\}$. For every $i \in \Omega$, let $\delta_i = N\mu(i) - 1$. In this setting, $t_\mu \ge \frac{1}{90(N\|\mu\|_2^2)^2} \cdot \frac{1}{N}\left(\sum_{i=1}^N \delta_i^2 + \sum_{i : \delta_i \ge 1} \delta_i^3\right)$.
\end{lemma}

\defproc{estimate-t-friendly}{Estimate-$t$-Friendly}

Lemma \reflemma{t-lbnd-by-sum-squares-sum-cubes} provides the mechanism for lower-bounding $t_\mu$ by obtaining lower bounds for $\sum_{i=1}^N \delta_i^2$ and $\sum_{i : \delta_i \ge 1} \delta_i^3$ and an upper bound for $\|\mu\|_2^2$. \algforprocshort{estimate-t-friendly}

\begin{lemma}{estimate-t-friendly}
    Let $X$ be the output of $\proc*{estimate-t-friendly}(\eta; \mu, \eps)$.
    \begin{itemize}
        \item If $\mu$ is friendly, then $X \ge t_\mu$ with probability at least $1-\eta$.
        \item If $\mu$ is friendly, then $\E[X] = O(t_\mu + \eps)$. Otherwise, $\E[X] = O(\sqrt{N})$.
        \item The sample complexity is $O\left(\log \frac{1}{\eta} \cdot \left(\frac{\sqrt{N}}{\eps} + \frac{1}{\eps\|\mu\|_2}\right) + N \log \frac{N}{\eta\eps}\right)$.
    \end{itemize}
\end{lemma}

\subsection{Finding an advice when $\|\mu\|_2$ is large}
\label{sec:ext-abs:subsec:finding-advice-large-mu2}

We follow a different approach for estimating $t$ when a direct estimation is too expensive. We observe that we can lower-bound $t$ by the mass of testable sets, as stated in the following key lemma.

\begin{lemma}{t-is-chebyshev}
    For every discrete distribution $\mu$, $t_\mu = \|\mu\|_3^3/\|\mu\|_2^4 - 1 \ge \sup_{\alpha > 0} \alpha^2 \Pr\left[\mu(i) \notin (1 \pm \alpha) \|\mu\|_2^2\right]$.
\end{lemma}

\subsubsection{Good partitions}
Let $A \cup B$ be a partition of the domain $\Omega$, obtained by an explicit construction of $A$ and an implicit definition of $B$ as its complement. Such a partition is considered \emph{good} if:
\begin{itemize}
    \item $A \subseteq \{ i : \mu(i) > \frac{11}{20} \|\mu\|_2^2 \}$.
    \item $B \subseteq \{ i : \mu(i) < \frac{2}{3} \|\mu\|_2^2 \}$.
\end{itemize}
That is, $A$ only contains large-mass elements and $B$ only contains a small-mass elements (note the overlap). Good partitions provide a few useful behaviors, as stated in the following lemmas. % Note that we could use every two constants in $(\frac{1}{2}, 1)$ leaving some overlap (to allow probabilistic testing of $\mu(i)$).

\begin{lemma}{good-partition--mu-B-small}
    Let $\mu$ be a discrete distribution and let $A \cup B$ be a good partition. In this setting, $\mu(B) \le 9t_\mu$.
\end{lemma}

\begin{lemma}{good-partition--A-is-friendly}
    Let $\mu$ be a discrete distribution and let $A \cup B$ be a good partition. If $t_\mu \le 1/900$, then the conditional distribution $\mu_A$ is friendly.
\end{lemma}

\begin{lemma}{good-partition--tmu-by-tmuA-muB}
    Let $\mu$ be a discrete distribution and let $A \cup B$ be a good partition. If $t_\mu \le 1/90$, then $t_{\mu_{\neg B}} \in t_\mu \pm 5\mu(B)$.
\end{lemma}

(Note that the constants $9$, $1/900$, $1/90$ and $5$ are chosen to fit the constant-factor choices in the definitions of friendly distributions and good partitions).

\subsubsection{The reduction}
\defproc{find-advice-large-mu2}{Find-Advice-Large-$\|\mu\|_2$}

We first test whether or not $t_\mu = \Omega(1)$. If we find that $t_\mu \ge 1/900$, then we estimate it directly and use the result. The rest of the algorithm assumes that $t_\mu \le 1/900$.

We learn the input distribution $\mu$ using $\tilde{O}(1 / \|\mu\|_2^2)$ samples to construct a set $A$ such that $A \cup (\Omega\setminus A)$ is a good partition with high probability. The algorithm can access $B = \Omega\setminus A$ only through the belonging oracle ($i \in^? B$), which is implemented as the negation of belonging to $A$ ($i \in B \leftrightarrow i \notin A$).

The rest of the algorithm is straightforward: we use the additive indicator estimation to estimate $\mu(B)$ within a $\pm \eps$ additive error, and then estimate $t_{\mu_A}$ using \proc*{estimate-t-friendly}. Since $\mu(A) \ge 99/100$, we can use rejection sampling to draw each $\mu_A$-sample at the expected cost of $O(1)$ $\mu$-samples. \algforprocshort{find-advice-large-mu2}

\begin{lemma}{find-advice-large-mu2}
    Let $X$ be the output of $\proc*{find-advice-large-mu2}(\eta; \mu,\eps)$. For every $0 < \eta \le 1/3$, a discrete distribution $\mu$ and $0 < \eps \le 1$,
    \begin{itemize}
        \item With probability at least $1 - \eta$, $X \ge t_\mu$.
        \item The expected output is $O(t_\mu + \eps/\|\mu\|_2)$.
        \item The expected sample complexity is $O\left(\frac{\log (1/\eta)}{\eps} + \frac{t \log (1/\eta)}{\eps^2} + \frac{\log (1 / \eta \eps \|\mu\|_2)}{\|\mu\|_2^2}\right)$.
    \end{itemize}
\end{lemma}

\subsection{An $\Omega(1 / \eps \|\mu\|_2)$ lower-bound}

At first, we show a lower bound for extremely skewed distributions.

\begin{lemma}{eps2-lower-bound-extreme-mu}
    Let $\mu$ be a discrete distribution over $\Omega$, and assume that there exists an element $i \in \Omega$ for which $\mu(i) \ge \frac{1}{8}\|\mu\|_2$. If $\eps \le 1/500$, then there exists a lower bound of $\Omega(1/\eps\|\mu\|_2)$ samples to distinguish between $\mu$ and distributions $\nu$ for which $\|\nu\|_2^2 \le (1 - (9/4)\eps)\|\mu\|_2$ (and in particular, $(1 + \eps)\|\nu\|_2^2 < (1 - \eps)\|\mu\|_2^2$).
\end{lemma}

Second, we show a lower bound for distributions with a specific structure. We assume that we can partition $\mu$'s elements into pairs such that the first element is not smaller than the second, but also not more than double. This allows us to move mass between the elements in each pair ($\mu(x_1) \pm \delta$, $\mu(x_2) \mp \delta$) independently, while keeping $\delta$ large enough (with respect to the bigger element of each pair) to have an effect on the collision norm. More specifically, for every such a pair, we use $\delta = \Theta(\sqrt{\eps} \mu(x_2))$.

\begin{lemma}{base-deviation-construction-for-eps-mu2}
    Let $\mu$ be a distribution over $\Omega = \{1,2,\ldots,\cdots\}$, and assume that for every $j \ge 1$, $\mu(2j) \le \mu(2j-1) \le \sqrt{2} \mu(2j)$. In this setting, for every $\eps \le 1/8000$, there exists a distribution $\mathcal D$ of distributions over $\Omega$ for which:
    \begin{itemize}
        \item When drawing $\nu$ from $\mathcal D$, with probability at least $3/4$, $\|\nu\|_2^2 \notin (1 \pm (5/2)\eps)\|\mu\|_2^2$.
        \item Any algorithm that distinguishes between $\mu$ and an input distribution $\nu$ drawn from $\mathcal D$ with total-variation distance greater than $1/12$ must draw $\Omega(1/\eps \|\mu\|_2)$ samples.
    \end{itemize}
\end{lemma}

For general-form $\mu$ distributions, we observe that unless $\mu$ is extremely skewed (there exists an element with mass at least $\frac{1}{8}\|\mu\|_2$), we can transform it into the ``pairwise'' shape by erasing a few elements, whose total mass is bounded by $3/4$ and whose contribution to the collision norm is at most $\frac{1}{9}\|\mu\|_2^2$.

\begin{lemma}{lbnd-eps-mu2}
    Any algorithm whose input is a discrete distribution $\mu$ and $\eps > 0$ which outputs a number in the range $(1 \pm \eps)\|\mu\|_2^2$ with probability at least $2/3$ must draw $\Omega(1/\eps \|\mu\|_2)$ samples.
\end{lemma}

\subsection{An $\Omega(t_\mu / \eps^2)$ lower-bound}
Given a discrete distribution $\mu$ over a domain $\Omega$, a parameter $\eps > 0$, a sign $s \in \{+1, -1\}$ and a coefficient $a$, we define the following distribution over $\Omega$:
\[ \nu_{s,a}(i) = \mu(i) \left(1 + \frac{a \eps}{t_\mu} \left(\frac{\mu(i)}{\|\mu\|_2^2} - 1\right)\right) \]

This construction is valid if $t_\mu \ge a\eps / \|\mu\|_2$, but if $t_\mu$ is smaller, then we can use the already-known lower bound $\Omega(1 / \eps \|\mu\|_2)$, which is not worse than $\Omega(t_\mu / a \eps^2)$.

We show two features of the construction:

\begin{lemma}{three-of-four-nu-have-far-mu22}
    Let $\eps > 0$, and assume that $t_\mu \ge 8\eps/\|\mu\|_2$. At least three (out of four) distributions $\nu_{s,a}$ for $s \in \{+1, -1\}$ and $a \in \{3,8\}$ have $\abs{\|\nu\|_2^2 - \|\mu\|_2^2} \ge 3\eps\|\mu\|_2^2$.
\end{lemma}

\begin{lemma}{lbnd-t-over-eps2-hardness-pair}
    Let $\eps > 0$ and $a > 0$, and assume that $t_\mu \ge a \eps/\|\mu\|_2$. For an integer $1 \le q \le \frac{t_\mu}{500 \eps^2}$, $\dtv\left(\frac{1}{2}\nu_{+1,a}^q + \frac{1}{2}\nu_{-1,a}^q, \mu^q\right) \le \frac{1}{12}$. % In other words, distinguishing between $\mu$ and an input distribution uniformly drawn from $\{\nu_{+1,a}, \nu_{-1,a}\}$ requires $\Omega(t_\mu/\eps^2)$ samples.
\end{lemma}

We combine these results to show that:
\begin{lemma}{lbnd-t-eps2}
    Let $\eps > 0$ and $\mu$ be a discrete distribution over $\Omega$. Distinguishing between $\mu$ and an input distribution $\nu$ for which $\abs{\|\nu\|_2^2 - \|\mu\|_2^2} \ge 3\eps\|\mu\|_2^2$ requires $q = \Omega(t_\mu / \eps^2)$ samples.
\end{lemma}

\section{Elementary tools for the upper bound}
\label{sec:ubnd-elementary-tools}

\paragraph{Amplification}

\repprovelemma{amplify-1/3-to-eta}
\begin{proof}
    By Chernoff's bound, $\Pr[\Bin(q,2/3) \le q/2] \le e^{-2(2/3-1/2)^2 q} = e^{-q/18} \le \eta$.
\end{proof}

\paragraph{Expected value of an amplified estimator}

\repprovelemma{median-expected-value}
\begin{proof}
    For $k \le 9$, it is easy to see that $\E[Y] \le 9\E[X]$, since the median of $k$ non-negative variables cannot be greater than their sum. We proceed with $k \ge 10$.

    Let $k \ge 10$ and $r = \floor{k/2}$, so that $r-1 \ge (2/5)(k-1)$.
    \begin{eqnarray*}
        \frac{\partial}{\partial p} \Pr\left[\Bin(k,p) \ge r\right]
        &=& \frac{\partial}{\partial p} \sum_{i=r}^k \binom{k}{r} p^i (1-p)^{k-i} \\
        &=& \sum_{i=r}^{k-1} \binom{k}{i} \left(i p^{i-1} (1-p)^{k-i} - (k-i) p^i (1-p)^{k-i-1}\right) + \binom{k}{k} \cdot k p^{k-1} \\
        &\le& \sum_{i=r}^{k} \binom{k}{i} i p^{i-1} (1-p)^{k-i} \\
        &\le& k^2 \sum_{j=r-1}^{k-1} \binom{k-1}{j} p^j (1-p)^{(k-1)-j}
        = k^2 \Pr\left[\Bin(k-1,p) \ge r-1\right]
    \end{eqnarray*}
    For $k \ge 10$, $r-1 \ge (2/5)(k-1)$, and therefore, $\frac{\partial}{\partial p} \Pr\left[\Bin(k,p) \ge r\right] \le k^2 e^{-\Omega(k)}$ for $0 \le p \le 1/3$. Hence, there exists a constant $C_1$ for which $\frac{\partial}{\partial p} \Pr\left[\Bin(k,p) \ge \floor{k/2}\right] \le C_1$ for every $k \ge 10$ and $0 \le p \le 1/3$. By the standard theorem, for every $k \ge 10$ and $0 \le p_1 \le p_2 \le 1/3$, $\Pr\left[\Bin(k,p_2) \le \floor{k/2}\right] - \Pr\left[\Bin(k,p_1) \le \floor{k/2}\right] \le C_1 (p_2 - p_1)$.
    
    Let $H = \{ a \in \supp(Y) : a \ge 3\E[X] \} = \{ a \in \supp(X) : a \ge 3\E[X] \}$. We use Markov's inequality to obtain $\Pr\left[X \ge a\right] \le 1/3$ for every $a \in H$. Therefore, for every $k \ge 10$:
    \begin{eqnarray*}
        \E[Y]
        &=& \sum_{a \in \supp(Y)} a \cdot \Pr\left[Y = a\right] \\
        &=& \sum_{a \in \supp(Y)} a \cdot \left(\Pr\left[Y \ge a\right] - \Pr\left[Y > a\right]\right) \\
        &\le& 3\E[X] + \sum_H a \cdot \left(\Pr\left[Y \ge a\right] - \Pr\left[Y > a\right]\right)
    \end{eqnarray*}

    We use the definition of $Y$ as the median of $k$ rounds to obtain:
    \begin{eqnarray*}
        \E[X] &\le&
        3\E[X] + \sum_H a \cdot \left(\Pr\left[\Bin\left(k, \Pr\left[X \ge a\right]\right) \ge \floor{k/2}\right] - \Pr\left[\Bin\left(k, \Pr\left[X > a\right]\right) \ge \floor{k/2}\right]\right) \\
        &\le& 3\E[X] + \sum_H a \cdot C_1 \left(\Pr\left[X \ge a\right] - \Pr\left[X > a\right]\right) \\
        &\le& 3\E[X] + C_1\E[X]
    \end{eqnarray*}

    To complete the proof, we choose $C = \max\{9, 3+C_1\}$.
\end{proof}

\paragraph{Additive estimation of an indicator}
\defproc{estimate-indicator-additive}{Estimate-Indicator-Additive}

For a black-box $\mathcal A$ sampling an indicator whose expected value is $p$, we estimate $\hat{p} = \Theta(p)$ using $O(\log \frac{1}{\eta} / \eps)$ samples, which suffices to determine its magnitude, and then use $O(\log (1/\eta) \cdot \hat{p} / \eps^2)$ samples to estimate $p$ within $\pm \eps$-error. \algforprocshort{estimate-indicator-additive}

\begin{proc-algo}{estimate-indicator-additive}{\eta; \mathcal A, \eps}
    \alginput{$\mathcal A$ is a black-box for sampling an indicator whose expected value is $p$}
    \algoutput{$X \in p \pm \eps$ with probability $\ge 1-\eta$}
    \algunbiasedness{$\E[X] = p$}
    \algcomplexity{$O(\log \eta^{-1} \cdot (\log \eps^{-1} / \eps + p/\eps^2))$}
    \begin{code}
        \algitem Let $M_1 \gets \ceil{12 \cdot \ln (10/\eta) / \eps}$.
        \algitem Call $\mathcal A$ for $M_1$ times.
        \algitem Let $S_1$ be the number of successful calls.
        \algitem Let $M_2 \gets \ceil{6 \ln (10/\eta) (S_1/M_1 + \eps) / \eps^2}$.
        \algitem Call $\mathcal A$ for $M_2$ times.
        \algitem Let $S_2$ be the number of successful calls.
        \algitem Return $S_2 / M_2$.
    \end{code}
\end{proc-algo}

\begin{lemma}{estimate-indicator-additive}
    Let $\mathcal A$ be a black-box sampling oracle for an indicator with expected value $p$. Procedure $\proc*{estimate-indicator-additive}(\eta; \mathcal A, \eps)$ is an unbiased estimator for $p$ whose additive error is at most $\eps$ with probability $1-\eta$. Moreover, its oracle-call complexity is $O(\log \eta^{-1} \cdot (1 / \eps + p/\eps^2))$.
\end{lemma}
\begin{proof}
    For complexity, observe that:
    \begin{eqnarray*}
        M_1 + \E[M_2]
        &=& O(\log \eta^{-1} / \eps) + O\left((\E[S_1/M_1] + \eps) \cdot \frac{\log \eta^{-1}} {\eps^2}\right) \\
        &=& O\left(\log \eta^{-1} / \eps\right) + O\left(p \cdot \frac{\log \eta^{-1}} {\eps^2}\right)
        = O\left(\log \eta^{-1} \left(\frac{1}{\eps} + \frac{p}{\eps^2}\right)\right)
    \end{eqnarray*}

    Clearly, for every condition on $M_2=m$, $\E[S_2/M_2 | M_2 = m] = p$, and therefore, $\E[S_2/M_2] = p$.
    
    If $p \ge \eps$, then by Chernoff's bound,
    \begin{eqnarray*}
        \Pr\left[S_1/M_1 < p/2 \right]
        &=& \Pr\left[S_1 < \E[S_1]/2 \right] \\
        &\le& e^{-\frac{1}{12} \E[S_1]}
        = e^{-\frac{1}{12} p M_1}
        \le e^{-\frac{1}{12} p (12 \cdot \ln (10/\eta) / \eps)}
        \le e^{-\ln (10/\eta)}
        = \frac{1}{10}\eta
    \end{eqnarray*}

    And hence,
    \begin{eqnarray*}
        \Pr\left[S_2/M_2 \notin p \pm \eps \right]
        &\le& \frac{1}{10}\eta\eps + \Pr\left[\Bin(M_2,p) \notin (p \pm \eps)M_2 \cond S_1/M_1 \ge \frac{1}{2}p \right] \\
        &\le& \frac{1}{10}\eta\eps + 2e^{-\frac{1}{3} \cdot (\eps/p)^2 \cdot (6 \cdot \ln (10/\eta) \cdot (p/2) / \eps)} \\
        &\le& \frac{1}{10}\eta\eps + 2e^{-\frac{1}{3} \cdot (\eps/p)^2 \cdot (3 \ln (10/\eta) \cdot p / \eps^2)} \\
        &=& \frac{1}{10}\eta\eps + 2e^{-\ln (10/\eta) / p}
        \le \frac{1}{10}\eta\eps + 2e^{-\ln (10/\eta)}
        \le \frac{1}{10}\eta\eps + \frac{1}{5}\eta
        \le \frac{3}{10}\eta
    \end{eqnarray*}
    Combined, the probability to return an output outside the range $p \pm \eps$ is bounded by $\frac{3}{10}\eta < \eta$.

    For $p < \eps$ we focus on additive error. We only have to consider the overestimation case, since $p-\eps < 0$.
    \begin{eqnarray*}
        \Pr\left[S_2/M_2 \ne p \pm \eps \right]
        &=& \Pr\left[S_2/M_2 > p + \eps \right] \\
        &=& \Pr\left[\Bin(M_2,p) > (p + \eps) M_2\right]
        = \Pr\left[\Bin(M_2,p) > \left(1 + \frac{\eps}{p}\right)\E[\Bin(M_2,p)]\right]
    \end{eqnarray*}

    Recall that $M_2 \ge 6 \ln \frac{10}{\eta} \eps / \eps^2 = 6\ln \frac{10}{\eta} / \eps$ regardless of $S_1$. Since $\eps/p \ge 1$, by Chernoff's bound:
    \[  \Pr\left[S_2/M_2 \ne p \pm \eps \right]
        \le e^{-\frac{1}{3}(\eps/p) \cdot p M_2}
        = e^{-\frac{1}{3}\eps M_2}
        \le e^{-\frac{1}{3}\eps \cdot (6 \ln (10/\eta) / \eps)}
        \le \frac{1}{10}\eta \]
\end{proof}

\paragraph{Rejection sampling}
In some procedures we construct a subset $A \subseteq \Omega$ and run a subroutine on the conditional distribution $\mu_A$. Since our algorithm can only sample $\mu$, it has to simulate an exact sampling of $\mu_A$ by rejection sampling, which is repeatedly sampling $\mu$ until obtaining an element belonging to $A$. If $A$ has sufficiently large mass, then the rejection sampling can be done effectively.

\begin{lemma}{rejection-sampling-concentration}
    Let $0 < \eta \le 1/3$ be an explicit parameter. Let $A$ be a subset accessible through the ``$i \in^? A$'' oracle. Assume that we wish to draw an unknown number of independent samples from the conditional distribution $\mu_A$ by rejection sampling. Consider the following logic: for the $i$th requested sample from $\mu_A$, we draw samples from $\mu$ until obtaining an $A$-element or until the total number of drawn $\mu$-samples exceeds $4(i + \ceil{12 \log (1/\eta)})$, in which case the rejection-sampler crashes. In this setting, if $\mu(A) > 1/2$, then with probability at least $1-\eta$, the rejection sampling does not crash even for infinitely many sample requests.
\end{lemma}
\begin{proof}
    Assume that $\mu(A) > 1/2$. For an individual $i \ge 1$, the sampler crashes at the $i$th sample with probability:
    \[\Pr\left[\Bin(4(i + \ceil{12 \ln (1/\eta)}), 1/2) < i\right]\]

    For $i \ge 12\ln (1/\eta)$:
    \begin{eqnarray*}
        \Pr\left[\Bin(4(i + \ceil{12\ln (1/\eta)}), 1/2) < i\right]
        &\le& \Pr\left[\Bin(4i, 1/2) < i\right] \\
        &\le& \Pr\left[\Bin(4i, 1/2) < 2i - i\right]
        \le e^{-2(i)^2/(4i)}
        \le e^{-i/2}
    \end{eqnarray*}
    
    For $i \le 12\ln (1/\eta)$:
    \begin{eqnarray*}
        \Pr\left[\Bin(4(i + \ceil{12\ln (1/\eta)}), 1/2) < i\right]
        &\le& \Pr\left[\Bin(\ceil{48 \ln (1/\eta)}, 1/2) < i\right] \\
        &\le& \Pr\left[\Bin(\ceil{48 \ln (1/\eta)}, 1/2) < \frac{1}{4} \cdot \ceil{48 \ln(1/\eta)} \right] \\
        &\le& e^{-\frac{1}{12} \cdot 48\ln(1/\eta)} \\
        &=&  e^{-4\ln(1/\eta)}
    \end{eqnarray*}

    By the union bound, the probability to crash is bounded by:
    \begin{eqnarray*}
        \sum_{i=1}^{\floor{12\ln(1/\eta)}} e^{-4\ln(1/\eta)} + \sum_{i=\floor{12\ln(1/\eta)}+1}^\infty e^{-i/2}
        &\le& 12\ln(1/\eta) \cdot \eta^4 + e^{-6\ln(1/\eta)} \sum_{i=0}^\infty e^{-i/2} \\
        &\le& 12\ln(1/\eta) \cdot \eta^4 + \eta^6 \cdot 3 \\
        &\le& (12 \eta^3 \ln (1/\eta) + 3\eta^5)\eta
        \le \eta
    \end{eqnarray*}
    The last transition is correct for every $0 < \eta \le 1/3$.
\end{proof}

\paragraph{Exponential tail}

\begin{lemma}{exponential-tail}
    There exist a non-decreasing monotone function $f : \mathbb N \times [0,1) \to \mathbb R^+$ with the following property: let $X$ be any non-negative random variable for which there exists some $a$ such that for every integer $\lambda \ge 1$, $\Pr[X \ge \lambda a] \le (\Pr[X \ge a])^\lambda$. If $\Pr[X \ge a] < 1$, then for every integer $r \ge 1$, $\E[X^r] \le a^r \cdot f(r,\Pr\left[X \ge a\right])$. In other words, if $r$ and $\Pr\left[X \ge a\right]$ are considered as constants, then $\E[X^r] = O(a^r)$.
\end{lemma}
\begin{proof}
    Let $f(r,p) = 1 + \sum_{i=1}^\infty (i+1)^r p^i$. For every non-negative random variable $X$ with an appropriate $a$,
    \begin{eqnarray*}
        \E[X^r]
        &\le& a^r + \sum_{i=1}^\infty ((i+1) \cdot a)^r \Pr\left[i \cdot a \le X < (i+1) a \right] \\
        &\le& a^r + a^r \sum_{i=1}^\infty (i+1)^r \cdot (\Pr\left[X \ge a\right])^i \\
        &=& a^r \cdot \left(1 + \sum_{i=1}^\infty (i+1)^r \cdot (\Pr\left[X \ge a\right])^i\right)
        = a^r \cdot f(r, \Pr\left[X \ge a\right])
    \end{eqnarray*}
\end{proof}

\section{Reference $\|\mu\|_2^2$-estimators}
\label{sec:reference-L2-estimators}

In this section we provide a few $L_2$-estimators. Some of the estimators are a rephrasing of existing ones for the sake of self-completeness of this paper.

\paragraph{The unbiased $L_2$-estimator}
Assume that we draw $m$ independent samples from a given distribution $\mu$, and let $S_m$ be the number of collisions. More explicitly, $S_m$ is the number of choices of $1 \le i < j \le m$ for which the $i$th sample equals to the $j$th sample. The expected value of $S_m$ is exactly $\|\mu\|_2^2 \cdot \binom{m}{2}$, and its variance is bounded by $\binom{m}{2} \|\mu\|_2^2 + m^3 \left(\|\mu\|_3^3 - \|\mu\|_2^4\right)$. Note that this variance has a matching lower bound (up to constant factors). The unbiased estimator cannot be used directly, since we must have some knowledge about $\mu$ to choose an appropriate number of samples.

\begin{lemma}[Based on \cite{BC17}]{base-algorithm-variance}
    Let $\mu$ be a discrete distribution over a (possibly infinite) domain. Assume that we draw $m$ samples, and let $S_m$ be the number of collisions. In this setting, $\E[S_m] = \binom{m}{2} \|\mu\|_2^2$ and $\Var[S_m] \le \binom{m}{2} \|\mu\|_2^2 + m^3 \left(\|\mu\|_3^3 - \|\mu\|_2^4\right)$.
\end{lemma}
\begin{proof}
    For every $1 \le i < j \le m$, let $X_{ij}$ be the indicator for a collision between the $i$th sample and the $j$th sample.
    
    \begin{eqnarray*}
        E\left[X_{ij}\right] &=& \sum_{x\in \Omega} (\mu(x))^2 = \|\mu\|_2^2 \\
        \E\left[X_{ij}^2\right] &=& \sum_{x\in \Omega} (\mu(x))^2 = \|\mu\|_2^2 \\
        \E\left[X_{ij} X_{ij'}\right] &=& \sum_{x\in \Omega} (\mu(x))^3 = \|\mu\|_3^3
    \end{eqnarray*}

    By linearity of expectation, $\E[S_m] = \sum_{1 \le i < j \le m} \E[X_{ij}] = \binom{m}{2} \|\mu\|_2^2$.
    
    For the variance,
    \begin{eqnarray*}
        \E\left[S_m^2\right]
        &=& \sum_{1 \le i < j \le m} \E\left[X_{ij}^2\right] + 2 \sum_i \sum_{\substack{1 \le j < j' \le m \\ j,j' \ne i}} \E\left[X_{ij} X_{ij'}\right] + \sum_{1 \le i < j \le m} \sum_{\substack{1 \le i' < j' \le m \\ i',j' \ne i,j}} \E\left[X_{ij}\right]\E\left[X_{i'j'}\right] \\
        &=& \binom{m}{2} \cdot \|\mu\|_2^2 + 2 m \binom{m - 1}{2} \|\mu\|_3^3 + \binom{m}{2} \binom{m-2}{2} \|\mu\|_2^4 \\
        &=& \binom{m}{2} \cdot \|\mu\|_2^2 + 2 (m-2) \binom{m}{2} \|\mu\|_3^3 + \binom{m}{2} \binom{m-2}{2} \|\mu\|_2^4 \\
        &=& \binom{m}{2} \left(\|\mu\|_2^2 + 2 (m-2) \|\mu\|_3^3 + 2 \binom{m-2}{2} \|\mu\|_2^4\right) \\
        &=& \binom{m}{2} \left(\|\mu\|_2^2 + 2 (m-2) \|\mu\|_3^3 + \left(\binom{m}{2} - (2m-3)\right) \|\mu\|_2^4\right)
    \end{eqnarray*}

    We use $\Var[S_m] = \E[S_m^2] - (\E[S_m])^2$:
    \begin{eqnarray*}
        \Var\left[S_m\right]
        &=& \binom{m}{2}\left(\|\mu\|_2^2 + 2 (m-2) \|\mu\|_3^3 + \left(\binom{m}{2} - (2m-3)\right) \|\mu\|_2^4\right) - \binom{m}{2}^2\|\mu\|_2^4 \\
        &=& \binom{m}{2}\left(\|\mu\|_2^2 + 2 (m-2) \|\mu\|_3^3 - (2m-3) \|\mu\|_2^4\right) \\
        &=& \binom{m}{2}\left(\|\mu\|_2^2 + 2 (m-2) (\|\mu\|_3^3 - \|\mu\|_2^4) - \|\mu\|_2^4\right) \\
        &\le& \binom{m}{2} \|\mu\|_2^2 + m^3 (\|\mu\|_3^3 - \|\mu\|_2^4)
    \end{eqnarray*}
\end{proof}

\paragraph{The BC-estimator}
\defproc{estimate-L2-BC}{Estimate-$L_2$-BC}

The BC-estimator is the same as the $O(1/\eps^2 \|\mu\|_2)$-algorithm presented in \cite{BC17} (up to constant factors). We keep drawing samples until reaching a pre-defined number $k = O(1/\eps^4)$ of collisions, and if we obtained the $k$th collision by the $M$th sample, then we use $k/\binom{M}{2}$ as the result. \algforprocshort{estimate-L2-BC}

\begin{proc-algo}{estimate-L2-BC}{\eta; \mu}
    \algimplicitamp
    \algoutput{$X \in (1 \pm \eps)\|\mu\|_2^2$ with probability $\ge 1 - \eta$}
    \algmoments{$\E[1/X^r] = O(1/\|\mu\|_2^{2r})$ for every $r \ge 0$}
    \algcomplexity{$O(\log \eta^{-1}/\|\mu\|_2)$}
    \begin{code}
        \algitem Initialize $M \gets 0$.
        \algitem Initialize $H \gets 0$.
        \algitem Let $k \gets \ceil{10^6 / \eps^4}$.
        \begin{While}{$H < k$}
            \algitem Set $M \gets M+1$.
            \algitem Draw $X_M \sim \mu$.
            \begin{For}{$i$ from $1$ to $M-1$}
                \begin{If}{$X_i = X_M$}
                    \algitem Set $H \gets H + 1$.
                \end{If}
            \end{For}
        \end{While}
        \algitem Return $k / \binom{M}{2}$.
    \end{code}
\end{proc-algo}

\begin{lemma}[Technical lemma, Based on \cite{BC17}]{technical:estimate-L2-BC-bound-pr-mlow}
    Let $0 < \eps \le 1/2$, $k \ge 100$ and $m = \min \{ m \in \mathbb N : (1+\eps)\binom{m}{2} \|\mu\|_2^2 \ge k \}$. In this setting, $\Pr[S_{m-1} \ge k] \le \frac{6}{\eps^2}\left(\frac{1}{k} + \frac{\|\mu\|_3^3 - \|\mu\|_2^4}{\sqrt{k} \|\mu\|_2^3}\right)$.
\end{lemma}

\begin{lemma}[Technical lemma, Based on \cite{BC17}]{technical:estimate-L2-BC-bound-pr-mhigh}
    Let $0 < \eps \le 1/2$, $k \ge 8/\eps^2$ and $m = \max \{ m \in \mathbb N : (1-\eps)\binom{m}{2} \|\mu\|_2^2 \le k \}$. In this setting, $\Pr[S_m < k] \le \frac{64}{\eps^2}\left(\frac{1}{k} + \frac{\|\mu\|_3^3 - \|\mu\|_2^4}{\sqrt{k} \|\mu\|_2^3}\right)$.
\end{lemma}

\begin{lemma}{estimate-L2-BC}
    Let $X$ be the output of $\proc*{estimate-L2-BC}(\eta; \mu, \eps)$. For every input distribution $\mu$ and $0 < \eps \le 1/2$:
    \begin{itemize}
        \item $\Pr[X \in (1 \pm \eps) \|\mu\|_2^2] \ge 2/3$.
        \item All moments $-\infty < r \le 0$ are preserved.
        \item The expected complexity is $O(\log (1/\eta) / \eps^2 \|\mu\|_2)$.
    \end{itemize}
\end{lemma}
\begin{proof}
    The following analysis holds for $\eta=1/3$. For $0 < \eta < 1/3$, see Lemma \reflemma{amplify-1/3-to-eta} (amplification) and Lemma \reflemma{median-expected-value} (applied to the variable $X' = X^r$ for every moment $r \le 0$).
    
    Recall that $k = \ceil{10^6/\eps^4} \ge \max\{100, 8/\eps^2\}$, and let $m_\mathrm{low} = \min\{ m : (1 + \eps)\binom{m}{2}\|\mu\|_2^2 \ge k \}$ and $m_\mathrm{high} = \max\{ m : (1 - \eps)\binom{m}{2}\|\mu\|_2^2 \le k \}$.

    By Lemma \reflemma{technical:estimate-L2-BC-bound-pr-mlow},
    \[  \Pr\left[M < m_\mathrm{low}\right]
        \le \frac{6}{\eps^2}\left(\frac{1}{10^6 / \eps^4} + \frac{1}{10^3 / \eps^2}\right)
        \le \frac{1}{15}
    \]
    
    By Lemma \reflemma{technical:estimate-L2-BC-bound-pr-mhigh},
    \[  \Pr\left[M > m_\mathrm{high}\right]
        \le \frac{64}{\eps^2}\left(\frac{1}{10^6 / \eps^4} + \frac{1}{10^3 / \eps^2}\right)
        \le \frac{1}{15}
    \]

    Combined, with probability at least $1-2/15 \ge 1-1/5$, $m_\mathrm{low} \le M \le m_\mathrm{high}$, and therefore, $1/\binom{M}{2} \in (1 \pm \eps) \|\mu\|_2^2$.

    For any integer $\lambda \ge 1$, if we draw $m' = \lambda m_\mathrm{high}$ samples, then $\Pr[S_{m'} < k] \le (\Pr[S_{m_\mathrm{high}} < k])^\lambda$. In other words, $\Pr[k/X \ge \lambda m_\mathrm{high}] \le (\Pr[k/X \ge m_\mathrm{high}])^\lambda$. For every $r > 0$, by applying Lemma \reflemma{exponential-tail} to the variable $1/X^r$ we obtain: 
    \[  \E[k^r/X^r]
        = O\left(\binom{m_\mathrm{high}}{2}^r\right) = O(k^r/\|\mu\|_2^{2r}) \]
    Therefore, $\E[1/X^r] = O(1/(\|\mu\|_2^2)^r)$.

    For the expected complexity: $\E[M] = \E[\sqrt{1/X}] \le \sqrt{\E[1/X]} = O(1/\|\mu|\|_2)$.
\end{proof}

\paragraph{The $L_2$ base-estimator}
\defproc{estimate-L2-base}{Estimate-$L_2$-Base}

The BC-estimator has useful negative moments, but it is not unbiased. The base estimator is an unbiased potentially-optimal estimator. By ``\emph{potentially}-optimal'' we mean that the base estimator can use an external advice.

First, we use the BC-estimator to obtain a fixed-factor estimation of $\|\mu\|_2^2$, at the cost of $O(1/\|\mu\|_2)$ samples. Then, we bound the variance of an unbiased estimator for $\|\mu\|_2^2$ that uses $m$ samples (Lemma \reflemma{base-algorithm-variance}). We can do it either by referring to a given advice $s$ or by using the norm inequality $\|\mu\|_3^3 - \|\mu\|_2^4 \le \|\mu\|_2^3$. We use this bound to choose an appropriate number of samples to draw for an unbiased estimation. \algforprocshort{estimate-L2-base}

\begin{proc-algo}{estimate-L2-base}{\eta; \mu, \eps, s}
    \alginput{$s \ge \|\mu\|_3^3 / \|\mu\|_2^4 - 1$ or $s = \bot$}
    \algoutput{$X \in (1 \pm \eps)\|\mu\|_2^2$ with probability $\ge 1 - \eta$}
    \algunbiasedness{$\E[X] = \|\mu\|_2^2$}
    \algcomplexity{$O(1/\eta \eps^2 \|\mu\|_2)$ if $s = \bot$}
    \algcomplexity{$O(1/\sqrt{\eta} \eps \|\mu\|_2 + s/\eta \eps^2)$ if $r \ne \bot$}
    \begin{code}
        \begin{If}{$\eps > 1/10$}
            \algitem Set $\eps \gets 1/10$.
        \end{If}
        \algitem Let $\ell \gets \proc{estimate-L2-BC}(\eta/6; \mu)$.
        \begin{If}{$s = \bot$}
            \algitem Let $s' \gets \sqrt{2 / \ell}$.
        \end{If}
        \begin{Else}
            \algitem Let $s' \gets s$.
        \end{Else}
        \algitem Let $m \gets \ceil{\frac{1}{\sqrt{\eta}} \max\{10^3 / \sqrt{\eta} \eps \sqrt{\ell}, 10^{6} s' / \eta \eps^2\}}$.
        \algitem Draw $m$ independent samples from $\mu$.
        \algitem Let $S_m$ be the number of collisions within these samples.
        \algitem Return $S_m / \binom{m}{2}$.
    \end{code}
\end{proc-algo}

\begin{lemma}[Technical lemma]{technical:estimate-L2-base-chebyshev}
    Let $0 < \eps \le 1/10$ and $\mu$ be a discrete distribution. For $m \ge 2$, recall that $S_m$ is the number of collisions within $m$ independent samples drawn from $\mu$. In this setting, $\Pr[S_m / \binom{m}{2} \notin (1 \pm \eps)\|\mu\|_2^2] \le \frac{5}{\eps^2 m^2 \|\mu\|_2^2} + \frac{20t_\mu}{\eps^2 m}$.
\end{lemma}

\begin{lemma}{estimate-L2-base}
    Let $X$ be the output of \proc*{estimate-L2-base}. For every input distribution $\mu$ and a parameter $0 < \eps < 1/2$ (which is allowed to be greater than $1$):
    \begin{itemize}
        \item If $s = \bot$ or $s \ge \|\mu\|_3^3 / \|\mu\|_2^4 - 1$, then with probability at least $1-\eta$, $X \in (1 \pm \eps) \|\mu\|_2^2$.
        \item $\E[X] = \|\mu\|_2^2$.
        \item The expected sample complexity is $O(1/\eta \eps^2 \|\mu\|_2)$ if $s = \bot$ and $O(1 / \sqrt{\eta} \eps \|\mu\|_2 + s/ \eta \eps^2)$ if $s \ne \bot$.
    \end{itemize}
\end{lemma}

\begin{proof}
    By Lemma \reflemma{estimate-L2-BC} (plus the explicit error parameter), with probability at least $1-\eta/6$, $\ell \in (1 \pm 1/2)\|\mu\|_2^2$. Additionally, $\E[1 / \ell] = O(\|\mu\|_2^2)$.

    Let $t_\mu = \frac{\|\mu\|_3^3 - \|\mu\|_2^4}{\|\mu\|_2^4}$. If the input advice $s$ is empty ($\bot$ symbol), then we use:
    \[  s'
        = \sqrt{2 / \ell}
        \ge \sqrt{2 / ((3/2)\|\mu\|_2^2)}
        \ge \frac{1}{\|\mu\|_2}
        \ge \frac{\|\mu\|_3^3 - \|\mu\|_2^4}{\|\mu\|_2^4}
        = t_\mu
    \]
    Additionally, due to Jensen's inequality, $\E[s'] = \E[\sqrt{2/\ell}] \le \sqrt{\E[2/\ell]} = O(1 / \|\mu\|_2)$.
    
    If the input advice $s$ is not empty, then we are allowed to assume that it is not smaller than $t_\mu$.

    Since $\sqrt{\ell} \le \sqrt{(3/2)\|\mu\|^2_2} \le (5/4)\|\mu\|_2$, we can use Lemma \reflemma{technical:estimate-L2-base-chebyshev} to obtain:
    \begin{eqnarray*}
        \Pr\left[X \notin (1 \pm \eps) \|\mu\|_2^2 \cond m \right]
        &\le& \frac{5}{\eps^2 m^2 \|\mu\|_2^2} + \frac{20t_\mu}{\eps^2 m} \\
        &\le& \frac{5}{\eps^2 (10^3 / \sqrt{\eta} \eps \sqrt{\ell})^2 \|\mu\|_2^2} + \frac{20t_\mu}{\eps^2 (10^6 s' / \eta \eps^2)} \\
        &\le& \frac{5}{\eps^2 (6.4 \cdot 10^5 / \eta \eps^2 \|\mu\|_2^2) \|\mu\|_2^2} + \frac{20\eta t_\mu}{10^6 t_\mu}
        \left(\frac{1}{10^5} + \frac{1}{5 \cdot 10^4}\right)\eta
        < \frac{1}{6}\eta
    \end{eqnarray*}

    That is, with probability at least $1-\eta/6-\eta/6 \ge 1-\eta$, $\ell$ is in the correct range and $X \in (1 \pm \eps)\|\mu\|_2^2$. Observe that $\E\left[S_m/\binom{m}{2}\right] = \|\mu\|_2^2$, since for every choice of $m$, $\E[S_m] = \|\mu\|_2^2 \binom{m}{2}$.
    
    For complexity: if $s \ne \bot$, then the expected complexity is $O(\log \eta^{-1}/\|\mu\|_2)$ (from \proc{estimate-L2-BC}, Lemma \reflemma{estimate-L2-BC}), plus:
    \begin{eqnarray*}
        \E[m]
        = O\left(\frac{\E[1/\sqrt{\ell}]}{\sqrt{\eta}\eps} + \frac{s}{\eta \eps^2}\right)
        = O\left(\frac{\sqrt{\E[1/\ell]}}{\sqrt{\eta} \eps} + \frac{s}{\eta \eps^2}\right)
        = O\left(\frac{1}{\sqrt{\eta} \eps\|\mu\|_2} + \frac{s}{\eta \eps^2}\right)
    \end{eqnarray*}

    If $s=\bot$, then the expected complexity is $O(\log \eta^{-1}/\|\mu\|_2)$ plus:
    \begin{eqnarray*}
        \E[m]
        = O\left(\frac{\E[1/\sqrt{\ell}]}{\sqrt{\eta} \eps} + \frac{\E[s']}{\eta \eps^2}\right)
        = O\left(\frac{\sqrt{\E[1/\ell]}}{\sqrt{\eta} \eps} + \frac{1/\|\mu\|_2}{\eta \eps^2}\right)
        = O\left(\frac{1}{\eta \eps^2 \|\mu\|_2}\right)
    \end{eqnarray*}
\end{proof}

\paragraph{The moment-preserving generic estimator}
\defproc{estimate-L2-moments}{Estimate-$L_2$-moments}

The BC-estimator preserves the negative moments, whereas the base estimator is unbiased (and in particular, preserving the first moment). For some complex algorithms, we have to generate a single random variable that preserves both negative moments and the first moment.

The moment-preserving estimation algorithm repeatedly runs two estimation procedures independently until the obtained estimations are multiplicatively close. Although our implementation is specific to $L_2$-estimation, it can trivially be generalized for every pair of algorithms that preserve different moments. \algforprocshort{estimate-L2-moments}

\begin{proc-algo}{estimate-L2-moments}{\eta; \mu, \eps}
    \algimplicitamp
    \algoutput{$X \in (1 \pm \eps) \|\mu\|_2^2$ with probability $\ge 1-\eta$}
    \algmoments{$\E[X^r] = O(\|\mu\|_2^{2r})$ for $0 \le r \le 1$}
    \algcomplexity{$O(\log(1/\eta)/\eps\|\mu\|_2)$}
    \begin{code}
        \begin{If}{$\eps > 1/5$}
            \algitem Set $\eps \gets 1/5$.
        \end{If}
        \begin{While}{True}
            \algitem Let $X^+ \gets \proc{estimate-L2-base}(1/6; \mu, \eps, \bot)$.
            \algitem Let $X^- \gets \proc{estimate-L2-BC}(1/6; \mu, \eps)$.
            \begin{If}{$X^-/2 \le X^+ \le 2X^-$}
                \algitem Return $X^+$.
            \end{If}
        \end{While}
    \end{code}
\end{proc-algo}

\begin{lemma}{estimate-L2-moments}
    Let $X$ be the output of $\proc*{estimate-L2-moments}(\eta; \mu, \eps)$. For every input distribution $\mu$, a parameter $0 < \eps \le 1/5$ and an error parameter $0 < \eta \le 1/3$:
    \begin{itemize}
        \item With probability at least $1 - \eta$, $X \in (1 \pm \eps)\|\mu\|_2^2$.
        \item All moments $-\infty < r \le 1$ are preserved.
        \item The sample complexity is $O(\log(1/\eta) / \eps^2 \|\mu\|_2)$.
    \end{itemize}
\end{lemma}

\begin{proof}
    The following analysis holds for $\eta=1/3$. For $0 < \eta < 1/3$, see Lemma \reflemma{amplify-1/3-to-eta} (amplification) and Lemma \reflemma{median-expected-value} (applied to the variable $X' = X^r$ for every moment $r \le 1$).

    The complexity of a single round is $O(1/\eps^2 \|\mu\|_2)$ (Lemma \reflemma{estimate-L2-base}, Lemma \reflemma{estimate-L2-BC}).

    The probability to return a number in the range $(1 \pm \eps)\|\mu\|_2^2$ in the first round is at least:
    \[\Pr\left[X^- \in (1 \pm \eps)\|\mu\|_2^2 \wedge X^+ \in (1 \pm \eps) \|\mu\|_2^2 \wedge X^-/2 \le X^+ \le 2 X^- \right]\]

    Since $\eps \le 1/5$, the first two conditions imply the third one:
    \[ [\cdots] = \Pr\left[X^- \in (1 \pm \eps)\|\mu\|_2^2 \wedge X^+ \in (1 \pm \eps) \|\mu\|_2^2 \right] \ge 1-2/6 = 2/3 \]

    This is also a lower bound for the probability to terminate in a single round, and therefore, the expected number of rounds is $O(1)$.
    
    For positive moments, since both $X^+$ and $X^-$ are non-negative with probability $1$,
    \begin{eqnarray*}
        \E[X^r]
        &=& \E\left[(X^+)^r \cond X^-/2 \le X^+ \le 2X^-\right] \\
        &\le& \frac{\E\left[(X^+)^r\right]}{\Pr\left[X^-/2 \le X^+ \le 2X^-\right]} \\
        &\le& \frac{1}{1-\eta}\E\left[(X^+)^r\right]
        \le 2\E\left[(X^+)^r\right]
    \end{eqnarray*}
    In particular, for $r=1$, $\E[X] \le 2\E[X^+] = O(\|\mu\|_2^2)$ (Lemma \reflemma{estimate-L2-base}).

    For negative moments, since both $X^+$ and $X^-$ are non-negative with probability $1$,
    \begin{eqnarray*}
        \E[1/X^r]
        &=& \E\left[1/(X^+)^r \cond X^-/2 \le X^+ \le 2X^-\right] \\
        &\le& \E\left[(2/X^-)^r \cond X^-/2 \le X^+ \le 2X^-\right] \\
        &\le& 2^r \frac{\E\left[1/(X^-)^r\right]}{\Pr\left[X^-/2 \le X^+ \le 2X^-\right]} \\
        &\le& \frac{2^r}{1-\eta}\E\left[1/(X^-)^r\right]
        \le 2^{r+1} \E\left[1/(X^-)^r\right]
    \end{eqnarray*}
    That is, for every hard-coded $r > 0$, $\E[1/X^r] = O((1/\|\mu\|_2)^{2r})$ (Lemma \reflemma{estimate-L2-BC}).
\end{proof}

\subsection{Deferred proofs of technical lemmas}
The next section begins at Page \pageref{sec:L3-estimators}.

\repprovelemma{technical:estimate-L2-BC-bound-pr-mlow}
\begin{proof}
    Observe that for $0 < \eps \le 1$, $m \ge \sqrt{k} / \|\mu\|_2 \ge 10$. By Chebyshev's inequality,
    \begin{eqnarray*}
        \Pr\left[S_{m-1} \ge k\right]
        &\le& \Pr\left[S_{m_-1} - \E[S_{m-1}] \ge k - \binom{m-1}{2}\|\mu\|_2^2\right] \\
        &\le& \Pr\left[\abs{S_{m-1} - \E[S_{m-1}]} \ge \eps \binom{m-1}{2}\|\mu\|_2^2\right] \\
        &\le& \frac{\Var\left[S_{m-1}\right]}{\eps^2 \binom{m-1}{2}^2 \|\mu\|_2^4} \\
        \text{[Lemma \reflemma{base-algorithm-variance}]} &\le& \frac{\binom{m-1}{2} \|\mu\|_2^2 + (m-1)^3 (\|\mu\|_3^3 - \|\mu\|_2^4)}{\eps^2 \cdot \left(\binom{m-1}{2}\|\mu\|_2^2\right)^2}
    \end{eqnarray*}
    
    Since $m \ge 10$:
    \begin{eqnarray*}    
        \Pr\left[S_{m-1} \ge k\right] &\le& \frac{3}{\eps^2}\left(\frac{1}{\binom{m}{2}\|\mu\|_2^2} + \frac{(m-1)\|\mu\|_2^2 \cdot (\|\mu\|_3^3 - \|\mu\|_2^4)}{\binom{m}{2}\|\mu\|_2^2 \cdot \|\mu\|_2^4} \right) \\
        &\le& \frac{3}{\eps^2}\left(\frac{1+\eps}{k} + \frac{2}{m} \cdot \frac{\binom{m}{2} \|\mu\|_2^2 \cdot (\|\mu\|_3^3 - \|\mu\|_2^4)}{\binom{m}{2}\|\mu\|_2^2 \cdot \|\mu\|_2^4} \right) \\
        &=& \frac{3}{\eps^2}\left(\frac{1+\eps}{k} + \frac{2}{m} \cdot \frac{\|\mu\|_3^3 - \|\mu\|_2^4}{\|\mu\|_2^4} \right)
    \end{eqnarray*}

    Since $m \ge \sqrt{k}/\|\mu\|_2$ as well,
    \[  \Pr\left[S_{m-1} \ge k\right]
        \le \frac{3}{\eps^2}\left(\frac{1+\eps}{k} + \frac{2\|\mu\|_2}{\sqrt{k}} \cdot \frac{\|\mu\|_3^3 - \|\mu\|_2^4}{\|\mu\|_2^4} \right)
        \le \frac{6}{\eps^2}\left(\frac{1}{k} + \frac{\|\mu\|_3^3 - \|\mu\|_2^4}{\sqrt{k} \|\mu\|_2^3}\right) \]
\end{proof}

\repprovelemma{technical:estimate-L2-BC-bound-pr-mhigh}
\begin{proof}
    Observe that for $0 < \eps \le 1/2$, $m \ge \sqrt{k} / \|\mu\|_2$.

    Since $\binom{m+1}{2} > k/\|\mu\|_2^2$, we obtain that $m \ge \max\{\sqrt{2k} - 1,k/\|\mu\|_2\} \ge \max\{4/\eps - 1, ,k/\|\mu\|_2\}$, and therefore, $\frac{m-1}{\eps(m+1) - 2} \le \frac{2}{\eps}$.
    
    By Chebyshev's inequality,
    \begin{eqnarray*}
        \Pr\left[S_m < k\right]
        &=& \Pr\left[\E[S_m] - S_m > \E[S_m] - k\right] \\
        &\le& \Pr\left[\E[S_m] - S_m > \left(1 - (1 - \eps) \frac{m+1}{m-1}\right)\binom{m}{2} \|\mu\|_2^2 \right] \\
        &=& \Pr\left[\E[S_m] - S_m > \frac{\eps(m+1) - 2}{m - 1}\binom{m}{2} \|\mu\|_2^2 \right] \\
        &\le& \frac{(m-1)^2}{(\eps(m+1)-2)^2} \cdot \frac{\Var[S_m]}{\binom{m}{2}^2 \|\mu\|_2^4}
    \end{eqnarray*}

    By Lemma \reflemma{base-algorithm-variance},
    \begin{eqnarray*}
        \Pr\left[S_m < k\right] &\le& \frac{4}{\eps^2} \cdot \frac{\binom{m}{2} \|\mu\|_2^2 + m^3 (\|\mu\|_3^3 - \|\mu\|_2^4)}{\binom{m}{2}^2 \|\mu\|_2^4} \\
        &\le& \frac{4}{\eps^2} \cdot \left(\frac{m+1}{m-1} \cdot \frac{1}{\binom{m+1}{2}\|\mu\|_2^2} + \frac{m^3 (\|\mu\|_3^3 - \|\mu\|_2^4)}{\binom{m}{2}^2 \|\mu\|_2^4}\right) \\
        \text{[Since $m \ge 2$]} &\le& \frac{4}{\eps^2} \cdot \left(3 \cdot \frac{1-\eps}{k} + \frac{16}{m} \frac{\|\mu\|_3^3 - \|\mu\|_2^4}{\|\mu\|_2^4}\right)
    \end{eqnarray*}

    Since $m \ge \sqrt{k}/\|\mu\|_2$ as well,
    \[  \Pr\left[S_m < k\right]
        \le \frac{4}{\eps^2} \cdot \left(3 \cdot \frac{1-\eps}{k} + \frac{16}{\sqrt{k}} \frac{\|\mu\|_3^3 - \|\mu\|_2^4}{\|\mu\|_2^3}\right)
        \le \frac{64}{\eps^2} \cdot \left(\frac{1}{k} + \frac{\|\mu\|_3^3 - \|\mu\|_2^4}{\sqrt{k} \|\mu\|_2^3}\right) \]
\end{proof}

\repprovelemma{technical:estimate-L2-base-chebyshev}
\begin{proof}
    Recall that $\E[S_m] = \binom{m}{2} \|\mu\|_2^2$
    By Chebyshev's inequality,
    \begin{eqnarray*}
        \Pr\left[\frac{S}{\binom{m}{2}} \notin (1 \pm \eps) \|\mu\|_2^2 \right]
        &=& \E\left[S_m \notin (1 \pm \eps) \E[S_m] \right] \\
        &\le& \frac{\Var[S_m]}{\eps^2 (\E[S_m])^2} \\
        \text{[Lemma \reflemma{base-algorithm-variance}]} &\le& \frac{\binom{m}{2} \|\mu\|_2^2 + m^3 (\|\mu\|_3^3 - \|\mu\|_2^4)}{\eps^2 \binom{m}{2}^2 \|\mu\|_2^4} \\
        &\le& \frac{1}{\eps^2 \binom{m}{2} \|\mu\|_2^2} + \frac{m^3 t_\mu}{\eps^2 \binom{m}{2}^2} \\
        \text{[Since $m \ge 2$]} &\le& \frac{5}{\eps^2 m^2 \|\mu\|_2^2} + \frac{20t_\mu}{\eps^2 m}
    \end{eqnarray*}
\end{proof}

\section{$\|\mu\|_3^3$-estimators}
\label{sec:L3-estimators}

In this section we provide a few $L_3$-estimators.

\paragraph{The unbiased $L_3$-estimator}
Some of our procedures require estimation of $\|\mu\|_3^3$. The core of our $L_3$-estimators is the unbiased $L_3$-estimator, in which we draw $m$ samples and let $T_m$ be the number of three-way collisions.

\begin{lemma}{L3-algorithm-variance}
    Let $\mu$ be a distribution over a (possibly infinite) domain. Assume that we draw $m$ samples, and let $T_m$ be the number of three-way collisions. In this setting, $\E[T_m] = \binom{m}{3} \|\mu\|_3^3$ and $\Var[T_m] \le \max\{ m^3 \|\mu\|_3^3, m^5 \|\mu\|_3^5\}$.
\end{lemma}
\begin{proof}
    For every $1 \le i < j < k \le m$, let $X_{ijk}$ be the indicator for a collision between the $i$th sample, the $j$th sample and the $k$th sample.

    Clearly, $\E\left[X_{ijk}\right] = \sum_i (\mu(i))^3 = \|\mu\|_3^3$. Also, for $1 \le i' < j' < k' \le m$ for which $\abs{\{i,j,k,i',j',k'\}} \le 5$, $\E[X_{ijk} X_{i'j'k'}] = \sum_i (\mu(i))^{\{i,j,k,i',j',k'\}} = \|\mu\|_{\abs{\{i,j,k,i',j',k'\}}}^{\abs{\{i,j,k,i',j',k'\}}}$.
    
    By linearity of expectation, $\E\left[T_m\right] = \binom{m}{3}\E\left[X_{1,2,3}\right] = \binom{m}{3}\|m\|_3^3$.
    
    For the variance,
    \begin{eqnarray*}
        \E\left[T_m^2\right]
        &=& \binom{m}{3}\E\left[X_{1,2,3}^2\right] + \binom{m}{4} \cdot \binom{4}{3}\E\left[X_{1,2,3} X_{1,2,4}\right] + \binom{m}{5} \cdot \binom{5}{3}\E\left[X_{1,2,3} X_{1,4,5}\right] \|m\|_4^4 + \cdots\\&& \binom{m}{6} \cdot \binom{m}{3} \E\left[X_{1,2,3} X_{4,5,6}\right] \\
        &=& \binom{m}{3}\|\mu\|_3^3 + 4 \binom{m}{4} \|\mu\|_4^4 + 10 \binom{m}{5} \|\mu\|_5^5 + 20 \binom{m}{6} \|\mu\|_3^6
    \end{eqnarray*}

    Hence,
    \begin{eqnarray*}
        \Var\left[T_m\right]
        &=& \E\left[T_m^2\right] - \left(\E\left[T_m\right]\right)^2 \\
        &=& \binom{m}{3}\|\mu\|_3^3 + 4 \binom{m}{4} \|\mu\|_4^4 + 10 \binom{m}{5} \|\mu\|_5^5 + \cdots\\&& 20 \binom{m}{6} \|\mu\|_3^6 - \binom{m}{3}^2 \|\mu\|_3^6 \\
        &\le& \binom{m}{3}\|\mu\|_3^3 + 4 \binom{m}{4} \|\mu\|_4^4 + 10 \binom{m}{5} \|\mu\|_5^5 \\
        \text{[Since $\|\mu\|_5 \le \|\mu\|_4 \le \|\mu\|_3$]} &\le& \frac{1}{6} m^3 \|\mu\|_3^3 (1 + m \|\mu\|_3 + m^2 \|\mu\|_3^2) \\
        &\le& \max\{m^3 \|\mu\|_3^3, m^5 \|\mu\|_3^5\}
    \end{eqnarray*}
\end{proof}

\paragraph{The $L_3$-estimator}
\defproc{estimate-L3}{Estimate-$L_3$}

As in the $L_2$-estimator, we obtain a fixed-factor estimation of $\|\mu\|_2^2$ and use it to bound both $\|\mu\|_3^3$ and the variance of the number of three-way collisions (as a function of $\|\mu\|_3^3$) within $m$ independent samples drawn from $\mu$. \algforprocshort{estimate-L3}

\begin{proc-algo}[hbtp]{estimate-L3}{\eta; \mu, \eps}
    \algoutput{$(1 \pm \eps)\|\mu\|_3^3$ with probability $\ge 1-\eta$}
    \algunbiasedness{$\E[X] = \|\mu\|_3^3$}
    \algcomplexity{$O(1/\eta \eps^2 \|\mu\|_2^{4/3})$}
    \begin{code}
        \begin{If}{$\eps > 1/10$}
            \algitem Set $\eps \gets 1/10$.
        \end{If}
        \algitem Let $\ell \gets \proc{estimate-L2-BC}(\eta/6; \mu, 1/2)$.
        \algitem Let $m \gets \ceil{10^{12}/ \eta \eps^2 \ell^{2/3}}$.
        \algitem Draw $m$ independent samples from $\mu$.
        \algitem Let $T_m$ be the number of three-way collisions within these samples.
        \algitem Return $T_m / \binom{m}{3}$.
    \end{code}
\end{proc-algo}

\begin{lemma}[Technical lemma]{technical:estimate-L3-chebyshev}
    Let $0 < \eps \le 1/10$ and $\mu$ be a discrete distribution. For $m \ge \max\{10, 1/\|\mu\|_3\}$, recall that $T_m$ is the number of three-way collisions within $m$ independent samples drawn from $\mu$. In this setting, $\Pr[T_m / \binom{m}{3} \notin (1 \pm \eps)\|\mu\|_3^3] \le \frac{100}{\eps^2 m \|\mu\|_3}$.
\end{lemma}

\begin{lemma}{estimate-L3}
    Let $X$ be the output of \proc*{estimate-L3}. For every input distribution $\mu$ and a parameter $0 < \eps < 1$:
    \begin{itemize}
        \item With probability at least $1-\eta$, $X \in (1 \pm \eps) \|\mu\|_3^3$.
        \item $\E[X] = \|\mu\|_3^3$.
        \item The expected sample complexity is $O(1 /\eta \eps^2 \|\mu\|_2^{4/3})$.
    \end{itemize}
\end{lemma}

\begin{proof}
    By Cauchy-Schwartz inequality, $\|\mu\|_3^3 \ge \|\mu\|_2^4$, and therefore, $\|\mu\|_3 \le \|\mu\|_2^{4/3}$.
    
    By Lemma \reflemma{estimate-L2-BC}, with probability at least $1-\eta/6$, $\ell \in (1 \pm 1/2)\|\mu\|_2^2$. Additionally, $\E[1 / \ell] = O(\|\mu\|_2^2)$.

    If $\ell \in (1 \pm 1/2)\|\mu\|_2^2$, then $\ell^{2/3} \le 2 \|\mu\|_3$ and $m \ge \max\{10, 1/\|\mu\|_3\}$. Therefore, by Lemma \reflemma{technical:estimate-L3-chebyshev},
    \[  \Pr\left[X \notin (1 \pm \eps) \|\mu\|_3^3 \cond m \right]
        \le \frac{100}{\eps^2 m \|\mu\|_3}
        \le \frac{100}{\eps^2 (10^{12} / \eta \eps^2 \ell^{2/3}) \|\mu\|_3}
        \le \frac{\ell^{2/3} / \|\mu\|_3}{10^{10}}\eta
        < \frac{1}{6}\eta \]
    
    That is, with probability at least $1-\eta/6-\eta/6 \ge 1-\eta$, $X \in (1 \pm \eps)\|\mu\|_3^3$.
    
    For the exact expected value, we observe that:
    \[  \E[X]
        = \sum_{m=1}^\infty \E[X | m] \Pr[m]
        = \sum_{m=1}^\infty \frac{\E[T_m]}{\binom{m}{3}} \Pr[m]
        = \sum_{m=1}^\infty \|\mu\|_3^3 \Pr[m]
        = \|\mu\|_3^3 \]

    For sample complexity,
    \begin{eqnarray*}
        \E[m]
        = O(1/\eta \eps^2) \cdot \E[1/\ell^{2/3}]
        \le O(1/\eta \eps^2) \cdot (\E[1/\ell])^{2/3}
        = O(1/\eta \eps^2 \|\mu\|_2^{4/3})
    \end{eqnarray*}
\end{proof}

\paragraph{The amplified $L_3$-estimator}
\defproc{estimate-L3-amplified}{Estimate-$L_3$-Amplified}
The $L_3$-estimator is unbiased, but its cost is polynomial with the error probability $\eta$. We use Lemma \reflemma{amplify-1/3-to-eta} and Lemma \reflemma{median-expected-value} to obtain that:

\begin{observation} \label{obs:estimate-L3-amplified}
    We can implement procedure $\proc*{estimate-L3-amplified}(\eta; \mu, \eps)$ for estimating $(1 \pm \eps)\|\mu\|_3^3$ with probability at least $1-\eta$ while preserving the first moment ($\E[X] = O(\|\mu\|_3^3)$) at the cost of $O(\log (1/\eta) / \eps^2 \|\mu\|_2^{4/3})$.
\end{observation}

\paragraph{The $L_3$ magnitude estimator}
\defproc{estimate-L3-magnitude}{Estimate-$L_3$-Magnitude}

To estimate the magnitude of $\|\mu\|_3^3$ according to a reference magnitude $a^3$, we choose the number of samples based on $a$. The expected number of three-way collisions in $m$ samples is $\binom{m}{3}\|\mu\|_3^3 \approx a^3 \|\mu\|_3^3$, and therefore, we use $m = \Theta(1/a)$ for having $\Theta(\|\mu\|_3^3 / a^3)$ such collisions in expectation. \algforprocshort{estimate-L3-magnitude}

\begin{proc-algo}{estimate-L3-magnitude}{\eta; \mu, a}
    \algoutput{$\|\mu\|_3^3 \pm \max\{\|\mu\|_3^3/2, a^3\}$, with probability $\ge 1-\eta$}
    \algunbiasedness{$\E[X] = \|\mu\|_3^3$}
    \algcomplexity{$O(1/\eta a)$}
    \begin{code}
        \algitem Let $m \gets \ceil{10^{12}/ \eta a}$.
        \algitem Draw $m$ independent samples from $\mu$.
        \algitem Let $T_m$ be the number of three-way collisions within these samples.
        \algitem Return $T_m / \binom{m}{3}$.
    \end{code}
\end{proc-algo}

\begin{lemma}{estimate-L3-magnitude}
    Let $X$ be the output of \proc*{estimate-L3-magnitude}. For every input distribution $\mu$ and a parameter $0 < a \le 1$:
    \begin{itemize}
        \item With probability at least $1-\eta$, $X \in \|\mu\|_3^3 \pm \max\{\|\mu\|_3^3/1000, a^3\}$.
        \item $\E[X] = \|\mu\|_3^3$.
        \item The expected sample complexity is $O(1 /\eta a)$.
    \end{itemize}
\end{lemma}

\begin{proof}
    The number of samples is explicitly $O(1/\eta a)$.
    
    For the exact expected value, we observe that $\E[X] = \frac{1}{\binom{m}{3}}\E[T_m] = \|\mu\|_3^3$.

    Let $\delta = \max\{\|\mu\|_3^3 / 1000, a^3\}$, so that $\|\mu\|_3^3 \le 1000\delta$ and $m \ge 10^{12}/\eta \delta^{1/3}$.

    By Lemma Chebyshev's inequality,
    \begin{eqnarray*}
        \Pr\left[X \notin \|\mu\|_3^3 \pm \delta \right]
        &=& \E\left[T_m \notin \E[T_m] \pm \binom{m}{3}\delta \right] \\
        &\le& \frac{\Var[T_m]}{\binom{m}{3}^2 \delta^2} \\
        \text{[Lemma \reflemma{L3-algorithm-variance}]} &\le& \frac{\max\{m^3 \|\mu\|_3^3, m^5 \|\mu\|_3^5\}}{\binom{m}{3}^2 \delta^2} \\
        &\le& \frac{\max\{m^3 \cdot 1000\delta, m^5 (1000\delta)^{5/3}\}}{\binom{m}{3}^2 \delta^2} \\
        &\le& 10^5 \max\left\{ \frac{m^3}{\binom{m}{3}^2 \delta}, \frac{1}{\binom{m}{3}^2 \delta^{1/3}} \right\} \\
        \text{[Since $m \ge 3$]} &\le& 10^8 \max\left\{\frac{1}{m^3 \delta}, \frac{1}{m \delta^{1/3}}\right\} \\
        &\le& 10^8 \max\left\{\frac{1}{(10^{12} / \eta \delta^{1/3})^3 \delta}, \frac{1}{(10^{12} / \eta \delta^{1/3}) \delta^{1/3}}\right\}
        \le \frac{10^8}{10^{12}}\eta
        \le \eta
    \end{eqnarray*}
\end{proof}

\subsection{Deferred proofs of technical lemmas}
The next section begins at Page \pageref{sec:top-level-algorithm}.

\repprovelemma{technical:estimate-L3-chebyshev}
\begin{proof}
    By Chebyshev's inequality,
    \begin{eqnarray*}
        \Pr\left[X \notin (1 \pm \eps) \|\mu\|_3^3 \right]
        &=& \E\left[T_m \notin (1 \pm \eps) \E[T_m] \right] \\
        &\le& \frac{\Var[T_m]}{\eps^2 (\E[T_m])^2} \\
        \text{[Lemma \reflemma{L3-algorithm-variance}]} &\le& \frac{\max\{m^3 \|\mu\|_3^3, m^5 \|\mu\|_3^5\}}{\eps^2 \binom{m}{3}^2 \|\mu\|_3^6} \\
        &\le& \frac{1}{\eps^2} \max\left\{\frac{m^3}{\binom{m}{3}^2\|\mu\|_3^3}, \frac{m^5}{\binom{m}{3}^2 \|\mu\|_3} \right\} \\
        \text{[Since $m \ge 10$]} &\le& \frac{100}{\eps^2} \max\left\{\frac{1}{m^3 \|\mu\|_3^3}, \frac{1}{m \|\mu\|_3}\right\} \\
        \text{[Since $m \ge 1/\|\mu\|_3$]} &\le& \frac{100}{\eps^2 m \|\mu\|_3}
    \end{eqnarray*}
\end{proof}

\section{Top-level algorithm}
\label{sec:top-level-algorithm}

In this section we provide the top-level logic of our estimation algorithm. The algorithm distinguishes between three (overlapping) cases:
\begin{itemize}
    \item $\|\mu\|_2 = O(\eps)$.
    \item $\|\mu\|_2 = O(\eps^{3/5} / \log \eps^{-1})$ and $\|\mu\|_2 = \Omega(\eps^3 \log \eps^{-1})$.
    \item $\|\mu\|_2 = \Omega(\eps \log \eps^{-1})$.
\end{itemize}

While we can distinguish between the cases with probability $1-O(\eta)$ at the cost of $O(\log (1/\eta) / \|\mu\|_2)$ samples (for example by Lemma \reflemma{estimate-L2-BC}), we have to make sure that the expected sample complexity is bounded by $O(1/\eps \|\mu\|_2 + t_\mu/\eps^2)$ as well. This requires that, in addition to the $O(\eta)$-error, we must have the following bounds for inputs not belonging to the overlapping ranges:
\begin{itemize}
    \item If $\|\mu\|_2$ belongs to the first case (``small''), then the probability to misclassify it must be $O(\eps^2 \|\mu\|_2^2)$.
    \item If $\|\mu\|_2$ belongs to a non-first case (``medium'' or ``large''), then the probability to misclassify it must be $O(\eps^2)$.
\end{itemize}
In a straightforward test, the first constraint would require $\Omega(\log (1/\eps\|\mu\|_2) / \max\{\eps,\|\mu\|_2\})$ samples, which is incompatible with the bi-criteria bound $O(1/\eps\|\mu\|_2)$. Therefore, we use the intersection of two independent tests: one $O(\eps^2)$-error classifier and one $O(\|\mu\|_2^2)$-verifier to make sure that the first-case is not misclassified.

\paragraph{The verifier}
In the verifier \proc{test-L2-magnitude}, we use the BC-estimator to have an initial estimation of the magnitude at the cost of $O(\log (1/\eta \eps) / \|\mu\|_2)$. If $\|\mu\|_2$ seems to be small then we can accept immediately, but if it seems to be large, then we verify it using additional $O(\log (1/\eta \|\mu\|_2) / \eps)$ samples. As a result, if $\|\mu\|_2$ is small then we accept with probability $1-\min\{\eta, O(\|\mu\|_2^2)\}$, and if $\|\mu\|_2$ is large, then we reject with probability at least $1-\eta$, preserving correctness with high probability. \algforprocshort{test-L2-magnitude}

\begin{proc-algo}{test-L2-magnitude}{\eta; \mu, \eps}
    \algoutput{\accept with probability $\ge 1-\min\{O(\|\mu\|_2^2), \eta\}$, if $\|\mu\|_2 \le \eps$}
    \algoutput{\reject with probability $\ge 1-\eta$, if $\|\mu\|_2 \ge 2\eps$}
    \algcomplexity{$O(\log (1/\eta) / \eps \|\mu\|_2)$}
    \begin{code}
        \algitem Let $Q$ be the expected sample complexity of $\proc{estimate-L2-BC}(1/4; \nu, 1/4)$, maximized over $\{ \nu : \|\nu\|_2 \ge 2\eps \}$.
        \begin{If}{$\eps > 1/4$}
            \algitem Set $\eps \gets 1/4$.
        \end{If}
        \algitem Let $\ell_1 \gets \proc{estimate-L2-BC}(\eta/2; \mu, 1/4)$.
        \begin{If}{$\sqrt{\ell_1} \le (3/2)\eps$}
            \algitem Return \accept.
        \end{If}
        \algitem Let $\ell_2 \gets \proc{estimate-L2-moments}(1/2; \mu, 1/4)$.
        \algitem Let $R \gets \ceil{18 \ln (2/\min\{1,\ell_2\} \eta)}$.
        \algitem Initialize $s \gets 0$.
        \begin{For}{$i$ from $1$ to $R$}
            \algitem Let $\ell'' \gets \proc{estimate-L2-BC}(1/4; \mu, 1/4)$, or $\ell'' \gets \bot$ after exceeding the bound of $12Q$ samples.
            \begin{If}{$\ell' \ne \bot$ and $\sqrt{\ell'} \ge (3/2)\eps$}
                \algitem Set $s \gets s + 1$.
            \end{If}
        \end{For}
        \begin{If}{$s \ge \floor{R/2}$}
            \algitem Return \reject.
        \end{If}
        \algitem Return \accept.
    \end{code}
\end{proc-algo}

\repprovelemma{test-L2-magnitude}
\begin{proof}
    We first analyze the verification phase, regardless of whether or not we actually take its branch.
    
    By Markov's inequality, the probability of an iteration to successfully run the BC-estimation is at least $1-1/12$. Combined with Lemma \reflemma{estimate-L2-BC}:
    \begin{itemize}
        \item If $\|\mu\|_2 \ge 2\eps$, then with probability at least $(3/4)-(1/12) = 2/3$, $\ell' \ge (3/4)\|\mu\|_2 \ge ((3/2)\eps)^2$.
        \item If $\|\mu\|_2 \le \eps$, then with probability at least $3/4 > 2/3$, $\ell' = \bot$ or $\ell' \le (5/4)\|\mu\|_2^2 \le ((3/2)\eps)^2$.
    \end{itemize}

    That is, the loop performs an $R$-round amplification of a $1/3$-error test. By Lemma \reflemma{amplify-1/3-to-eta}, the probability of the test to provide the wrong answer is bounded by $\eta \min\{1,\ell_2\} / 2$. Since $\ell_2$ is obtained from $\proc{estimate-L2-moments}$, the error probability of the inner test is bounded by
    \[  \frac{1}{2}\eta \E[\min\{1,\ell_2\}]
        \le \frac{1}{2}\eta \min\{\E[\ell_2],1\}
        \le \frac{1}{2}\eta, O(\|\mu\|_2^2) \]

    For the outer test: if $\|\mu\|_2 \le \eps$, then with some probability we accept, and otherwise we take the ``if-true'' branch and run the inner test, which accepts with probability at least $1 - \min\{O(\|\mu\|_2^2), \eta\}$.

    If $\|\mu\|_2 \ge 2\eps$, then with probability at least $1 - \eta/2$ we take the ``if-true'' branch and run the inner test, which rejects with probability at least $1-\eta/2$. By the union bound, the probability to reject is at least $1 - \eta$.

    For complexity: the estimation of $\ell_1$ costs $O(\log(1/\eta)/\|\mu\|_2)$ (Lemma \reflemma{estimate-L2-moments}).

    By Lemma \reflemma{estimate-L2-BC}, $Q = O(1/\eps)$. Therefore, the expected cost of the inner test, if executed, is
    \begin{eqnarray*}
        Q \cdot O(\E[\log (1/\eta \max\{1,\ell_2\})])
        &=& O(1/\eps) \cdot O(\log (\E[1/\max\{1,\ell_2\}]/\eta)]) \\
        &=& O(1/\eps) \cdot O(\log ((1 + \E[1/\ell_2])/\eta)]) \\
        &=& O(1/\eps) \cdot O(\log ((1 + 1/\|\mu\|_2^2)/\eta)]) \\
        &=& O(\log(1/\eta\|\mu\|_2)/\eps) \\
        &=& O(\log(1/\eta) / \eps \|\mu\|_2)
    \end{eqnarray*}

    Both the initial estimation and the inner test cost $O(\log(1/\eta)/\eps\|\mu\|_2)$.
\end{proof}

\paragraph{Back to the top-level logic}
The top-level algorithm chooses an advice $s \ge t_\mu$ and then calls the advised procedure \proc{estimate-L2-base}. To keep the expected sample complexity of the advice-finding part low, while also keeping the expected advice low, we must correctly classify a first-case input with probability $O(\eps^2 \|\mu\|_2^2)$ and correctly classify a second-case or a third-case input with probability $1 - O(\eps^2)$ (but never less than $1-\eta/4$, to keep the $(1-\eta)$-correctness).

First, we estimate $\|\mu\|_2$ with success probability $1-O(\eps^2)$, and use it to classify the case of the input. If we believe that we are not in the first case, then we call \proc{test-L2-magnitude} as a verifier, to reduce the probability to misclassify a first-case input to $O(\eps^2 \|\mu\|_2^2)$. To reduce the expected advice, if we believe that we are in a non-first case and the verification test fails, then the algorithm gives up and chooses a worthless advice $s=0$. \algforprocshort{estimate-L2-top-level}

\begin{proc-algo}{estimate-L2-top-level}{\eta; \mu, \eps}
    \algoutput{A number in the range $(1 \pm \eps)\|\mu\|_2^2$}
    \algunbiasedness{$\E[X] = \|\mu\|_2^2$}
    \algcomplexity{$O\left(\frac{1}{\sqrt{\eta} \eps \|\mu\|_2} + \frac{t}{\eta \eps^2}\right)$, where $t = (\|\mu\|_3^3 - \|\mu\|_2^4) / \|\mu\|_2^4$}
    \begin{code}
        \algitem Let $\ell \gets \proc{estimate-L2-BC}(\min\{\eps^2, \eta/4\}; \mu; 1/4)$.
        \begin{If}{$\sqrt{\ell} \le 4\eps$} 
            \algitem Set $s \gets \proc{find-advice-small-mu2}(\eta/4; \mu, \eps)$.
        \end{If}
        \begin{Else}
            \begin{If}{$\proc{test-L2-magnitude}(\eta/4; \mu, \eps)$ accepts}
                \algitem Set $s \gets 0$.
            \end{If}
            \begin{Else}
                \begin{If}{$\sqrt{\ell} \le 2\eps^{2/3}$}
                    \algitem Set $s \gets \proc{find-advice-medium-mu2}(\eta/4; \mu, \eps)$.
                \end{If}
                \begin{Else}
                    \algitem Set $s \gets \proc{find-advice-large-mu2}(\eta/4; \mu, \eps)$.
                \end{Else}
            \end{Else}
        \end{Else}
        \algitem Return $\proc{estimate-L2-base}(\eta/4; \mu, s)$.
    \end{code}
\end{proc-algo}

\begin{lemma}{estimate-L2-top-level--first-case-misclassification}
    Consider the run of $\proc{estimate-L2-top-level}(\eta; \mu,\eps)$ where $\|\mu\|_2 \le \eps$. The probability of the algorithm to both have $\sqrt{\ell} > 4\eps$ and fail the small-$\|\mu\|_2$ test is $O(\eps^2 \|\mu\|_2^2)$.
\end{lemma}
\begin{proof}
    By Lemma \reflemma{estimate-L2-BC}, the probability that $\ell > 16\eps^2 \ge 16\|\mu\|_2^2$ is bounded by $\eps^2$. Independently, by Lemma \reflemma{test-L2-magnitude}, the probability to fail the small-$\|\mu\|_2$ test is $O(\|\mu\|_2^2)$. The probability to make both misclassifications is therefore bounded by $O(\eps^2 \|\mu\|_2^2)$.
\end{proof}

\begin{lemma}{estimate-L2-top-level--medium-branch-contribution}
    Consider the run of $\proc{estimate-L2-top-level}(\eta; \mu,\eps)$. The contribution of the medium-$\eps$ branch to the expected complexity is $O(\log(1/\eta) / \eps \|\mu\|_2)$.
\end{lemma}
\begin{proof}
    By Lemma \reflemma{find-advice-medium-mu2}, the sample complexity of $\proc{find-advice-medium-mu2}$, if it is executed, is $O(\log (1/\eta \eps) \cdot (\|\mu\|_2^2 / \eps^2 + 1/\|\mu\|_2^{4/3}))$.

    Case I (bad). $\|\mu\|_2 \le \eps$. In this case, the probability to take the medium-$\eps$ branch is $O(\eps^2 \|\mu\|_2^2)$ (Lemma \reflemma{estimate-L2-top-level--first-case-misclassification}). Therefore, the contribution of the medium-$\eps$ branch in this case is bounded by:
    \[  \eps^2 \|\mu\|_2^2 \cdot O\left(\log \frac{1}{\eta\eps} \cdot \left(\frac{\|\mu\|_2^2}{\eps^2} + \frac{1}{\|\mu\|_2^{4/3}}\right)\right)
        = O\left(\log \frac{1}{\eta\eps} \cdot \left(\|\mu\|_2^4 + \|\mu\|_2^{2/3}\eps^2 \right)\right)
        = O(1) \]

    Case II (good). $\eps \le \|\mu\|_2 \le \eps^{2/3}$. In this case, the expected complexity of the medium-$\eps$ branch is:
        \begin{eqnarray*}
        O\left(\log \frac{1}{\eta\eps} \cdot \left(\frac{\|\mu\|_2^2}{\eps^2} + \frac{1}{\|\mu\|_2^{4/3}}\right)\right)
        &=& O\left(\log \frac{1}{\eta\eps} \cdot \frac{1}{\eps \|\mu\|_2} \cdot \left(\frac{\|\mu\|_2^3}{\eps} + \frac{\eps}{\|\mu\|_2^{1/3}}\right)\right) \\
        &=& O\left(\log \frac{1}{\eta\eps} \cdot \frac{1}{\eps \|\mu\|_2} \left(\frac{(\eps^{2/3})^3}{\eps} + \frac{\eps}{\eps^{1/3}}\right)\right) \\
        &=& O\left(\log \frac{1}{\eta\eps} \cdot \frac{1}{\eps \|\mu\|_2} \cdot \eps^{2/3} \right)
        = O\left(\frac{\log (1/\eta)}{\eps\|\mu\|_2}\right)
    \end{eqnarray*}

    Case III (bad). $\|\mu\|_2 \ge \eps^{2/3}$. In this case, the probability to enter the medium-$\eps$ branch (by $\eps < (1/4)\sqrt{\ell}$ and $\eps \ge 2\ell_2^{3/4}$) is bounded by $\eps^2$ (Lemma \reflemma{estimate-L2-BC}). Therefore, the contribution of the medium-$\eps$ branch in this case is bounded by:
    \begin{eqnarray*}
        \eps^2 \cdot O\left(\log \frac{1}{\eta \eps} \cdot \left(\frac{\|\mu\|_2^2}{\eps^2} + \frac{1}{\|\mu\|_2^{4/3}}\right)\right)
        &=& O\left(\log \frac{1}{\eta \eps} \cdot \left(\|\mu\|_2^2 + \frac{\eps^2}{\|\mu\|_2^{4/3}}\right)\right) \\
        &=& O\left(\log \frac{1}{\eta \eps} \cdot \left(1 + \frac{\eps^2}{(\eps^{2/3})^{4/3}}\right)\right) \\
        &=& O\left(\log \frac{1}{\eta \eps} \cdot (1 + \eps^{10/9})\right)
        = O\left(\log \frac{1}{\eta \eps}\right)
        = O\left(\frac{\log (1/\eta)}{\eps}\right)
    \end{eqnarray*}

    In all cases, the contribution to the expected sample complexity is bounded by $O(\log (1/\eta) / \eps \|\mu\|_2)$.
\end{proof}

\begin{lemma}{estimate-L2-top-level--small-branch-contribution}
    Consider the run of $\proc{estimate-L2-top-level}(\eta; \mu,\eps)$. The contribution of the small-$\eps$ branch to the expected complexity is $O(\log(1/\eta) \cdot (1/\eps \|\mu\|_2 + t/\eps^2))$.
\end{lemma}
\begin{proof}
    By Lemma \reflemma{find-advice-large-mu2}, the sample complexity of \proc{find-advice-large-mu2}, if it is executed, is:
    \[O\left(\frac{\log (1/\eta)}{\eps} + \frac{t_\mu \log (1/\eta)}{\eps^2} + \frac{\log (1 / \eta \eps \|\mu\|_2)}{\|\mu\|_2^2}\right)\]
    Note that, since $t_\mu \le \|\mu\|_3^3/\|\mu\|_2^4 \le 1/\|\mu\|_2$, this can be relaxed (when appropriate) to:
    \[ O\left(\frac{\log (1/\eta)}{\eps^2 \|\mu\|_2} + \frac{\log (1/\eta \eps)}{\|\mu\|_2^2} \right) \]

    Case I (bad). $\|\mu\|_2 \le \eps$. In this case, the probability to take the small-$\eps$ branch is bounded by $O(\eps^2 \|\mu\|_2^2)$ (Lemma \reflemma{estimate-L2-top-level--first-case-misclassification}). Therefore, the contribution of the small-$\eps$ branch in this case is bounded by:
    \[  \eps^2 \|\mu\|_2^2 \cdot O\left(\frac{\log (1/\eta)}{\eps^2 \|\mu\|_2} + \frac{\log (1/\eta \eps)}{\|\mu\|_2^2} \right)
        = O\left(\|\mu\|_2 \log (1/\eta) + \eps^2 \log(1/\eta\eps\|\mu\|_2) \right)
        = O(\log (1/\eta)/\|\mu\|_2)
    \]

    Case II (bad). $\eps \le \|\mu\|_2 \le \eps^{2/3}$. In this case, the probability to enter the small-$\eps$ branch (by $\eps < (1/4)\sqrt{\ell}$ and $\eps < 2\ell_2^{3/4}$) is bounded by $\eps^2$ (Lemma \reflemma{estimate-L2-BC}). Therefore, the contribution of the small-$\eps$ branch in this case is bounded by:
    \begin{eqnarray*}
        \eps^2 \cdot O\left(\frac{\log (1/\eta)}{\eps^2 \|\mu\|_2} + \frac{\log (1/\eta \eps)}{\|\mu\|_2^2} \right)
        &=& O\left(\frac{\log(1/\eta)}{\|\mu\|_2} + \frac{\eps^2 \log (1/\eta\eps)}{\|\mu\|_2^2}\right) \\
        \text{[Since $\|\mu\|_2 \ge \eps$]} &=& O\left(\frac{\log(1/\eta)}{\|\mu\|_2} + \frac{\eps \log (1/\eta\eps)}{\|\mu\|_2}\right)
        = O\left(\frac{\log (1/\eta)}{\eps \|\mu\|_2}\right)
    \end{eqnarray*}

    Case III (good). $\|\mu\|_2 \ge \eps^{2/3}$. The expected complexity of the small-$\eps$ branch is:
    \begin{eqnarray*}
        O\left(\frac{\log (1/\eta)}{\eps} + \frac{t \log (1/\eta)}{\eps^2} + \frac{\log (1 / \eta \eps \|\mu\|_2)}{\|\mu\|_2^2}\right)
        &=& O\left(\frac{\log (1/\eta)}{\eps} + \frac{t \log (1/\eta)}{\eps^2} + \frac{\log (1 / \eta \eps)}{\eps^{2/3} \|\mu\|_2}\right) \\
        &=& O\left(\frac{\log (1/\eta)}{\eps} + \frac{t \log (1/\eta)}{\eps^2} + \frac{\log (1 / \eta)}{\eps \|\mu\|_2}\right) \\
        &=& O\left(\log \frac{1}{\eta} \cdot \left(\frac{1}{\eps \|\mu\|_2} + \frac{t}{\eps^2}\right)\right)
    \end{eqnarray*}
\end{proof}

\repprovelemma{estimate-L2-top-level}
\begin{proof}
    In the following table, we describe the relevant branches in each range of $\eps$, given that $\ell \in (1 \pm 1/4)\|\mu\|_2^2$ with probability at least $1 - \eta/4$. There is an additional column of what the large-$\eps$ test determines with probability at least $1-\eta/4$. Some of the ranges can be empty if $\|\mu\|_2 \approx 1$.
    \[\begin{array}{llllll}
        \text{Minimum $\|\mu\|_2$} & \text{Maximum $\|\mu\|_2$} & \text{Large-$\eps$} & \text{Medium-$\eps$} & \text{Small-$\eps$} & \text{\proc*{test-L2-magnitude}} \\
        0^+ & \eps & \text{V} & \text{} & \text{} & \accept \\
        \eps & 2\eps & \text{V} & \text{} & \text{} & \text{*} \\
        2\eps & 8\eps & \text{V} & \text{V} & \text{} & \reject \\
        8\eps & \eps^{2/3} & \text{} & \text{V} & \text{} & \reject \\
        \eps^{2/3} & 4\eps^{2/3} & \text{} & \text{V} & \text{V} & \reject \\
        4\eps^{2/3} & 1 & \text{} & \text{} & \text{V} & \reject
    \end{array}\]

    Therefore, for every magnitude of $\|\mu\|_2$, the algorithm enters the correct branch with probability at least $1 - \eta/4$, and executes it with probability at least $1 - \eta/4$. (Recall that the small-$\|\mu\|_2$ branch is always executed if taken, since \proc*{test-L2-magnitude} only applies to the medium-$\|\mu\|_2$ and the large-$\|\mu\|_2$ branches).

    Additionally, if the algorithm takes the correct branch and executes it, then there is a probability of $1-\eta/4$ to correctly obtain $s \ge t_\mu$ (Lemmas \reflemma{find-advice-small-mu2}, \reflemma{find-advice-medium-mu2}, \reflemma{find-advice-large-mu2}) and a probability of $1-\eta/4$ to obtain an output in the range $(1 \pm \eps)\|\mu\|_2^2$ (Lemma \reflemma{estimate-L2-base}). By the union bound, the probability to have a correct output is at least $1-4\eta/4 = 1-\eta$.

    Moreover, the expected output is $\|\mu\|_2^2$, regardless of the value of $s$ (Lemma \reflemma{estimate-L2-base}).

    For complexity:
    \begin{itemize}
        \item The estimation of $\ell$ costs $O(\log(1/\eta\eps)/\|\mu\|_2)$ samples (Lemma \reflemma{estimate-L2-BC}).
        \item The small-$\|\mu\|_2$ test, if executed, costs $O(\log(1/\eta) / \eps \|\mu\|_2)$ samples (Lemma \reflemma{test-L2-magnitude}).
        \item The large-$\eps$ branch, if executed, costs $O(\log(1/\eta) / \eps^{1/3} \|\mu\|_2)$ samples (Lemma \reflemma{find-advice-small-mu2}).
        \item The medium-$\|\mu\|_2$ branch contributes $O(\log(1/\eta) / \eps \|\mu\|_2)$ samples to the expected complexity (Lemma \reflemma{estimate-L2-top-level--medium-branch-contribution}).
        \item The large-$\|\mu\|_2$ branch contributes $O(\log(1/\eta) \cdot (1 / \eps \|\mu\|_2 + t/\eps^2))$ samples to the expected complexity (Lemma \reflemma{estimate-L2-top-level--small-branch-contribution}).
        \item Generating the output costs $O(1/\sqrt{\eta}\eps\|\mu\|_2 + \E[s]/\eta\eps^2)$.
    \end{itemize}
    That is, the expected sample complexity is $O\left(\frac{1}{\sqrt{\eta} \eps \|\mu\|_2} + \frac{\E[s]}{\eta \eps^2}\right)$. 

    For $\E[s]$:
    \begin{itemize}
        \item All branches but the small-$\|\mu\|_2$ branch provide $\E[s] = O(t_\mu + \eps/\|\mu\|_2)$ (Lemmas \reflemma{find-advice-medium-mu2}, \reflemma{find-advice-large-mu2}).
        \item The large-$\eps$ branch provide $\E[s] = O(t_\mu + \eps/\|\mu\|_2 + 1)$ if executed.
        \item If $\|\mu\|_2 \ge 8\eps$, then the probability to enter the small-$\|\mu\|_2$ branch is bounded by $\eps^2$, and therefore, its expected contribution to $\E[s]$ is reduced to $O(t_\mu + \eps/\|\mu\|_2 + \eps)$.
        \item If $\|\mu\|_2 \le 8\eps$, then the contribution of the small-$\|\mu\|_2$ branch, which is $O(t_\mu + \eps/\|\mu\|_2 + 1)$, can be relaxed to $O(t_\mu + \eps/\|\mu\|_2)$.
    \end{itemize}
    In all cases, the expected value of $s$ is bounded by $O(t_\mu + \eps/\|\mu\|_2)$. Therefore, the expected sample complexity of \proc*{estimate-L2-top-level} is $O\left(\frac{1}{\eta} \cdot \left(\frac{1}{\eps \|\mu\|_2} + \frac{t}{\eps^2}\right)\right)$.
\end{proof}

\section{Finding an advice when $\|\mu\|_2$ is small}
\label{sec:small-mu2}

If $\|\mu\|_2 = O(\eps)$, then $t_\mu = \Omega(1)$ or $t_\mu + 1 = O(\eps/\|\mu\|_2)$. Therefore, an estimation of $\|\mu\|_3^3 / \|\mu\|_2^4 = t_\mu + 1$ is sufficiently accurate as an estimation of $t_\mu$. If $\|\mu\|_2 = O(\eps)$, then we only obtain an estimation of $t_\mu + 1$, which can still be used as an advice since it is greater than $t_\mu$, but more expensive.

To estimate $t_\mu$, we first estimate $\ell_2 \approx \|\mu\|_2^2$ within a $(1 \pm O(1))$-factor and then estimate $\ell_3 \approx \|\mu\|_3^3$ within additive error bounded by $\max\{ \eps \ell_2^{3/2} , O(1) \cdot \|\mu\|_3^3\}$.

If the error is dominated by the first term $\eps \ell^{3/2} \approx \eps \|\mu\|_2^3$, then we obtain a test of $\|\mu\|_3^3 / \|\mu\|_2^4$ for being smaller than $O(\eps/\|\mu\|_2)$. Passing this test results in an output of $O(\eps/\|\mu\|_2 + 1)$.

If the error is dominated by the second term $O(1) \cdot \|\mu\|_3^3$, then we obtain a $(1+O(1))$-estimation of $\|\mu\|_3^3 / \|\mu\|_2^4 = (1 \pm O(1))t_\mu \pm O(1)$, which is sufficiently accurate in the small-$\|\mu\|_2$ case. \algforprocshort{find-advice-small-mu2}

\begin{proc-algo}{find-advice-small-mu2}{\eta; \mu, \eps}
    \algoutput{$t \le X \le 2t_\mu + 3\eps/\|\mu\|_2 + 2$ with probability $\ge 1-\eta$}
    \algmoments{$\E[X] = O(t + \eps/\|\mu\|_2 + 1)$}
    \algcomplexity{$O(1/\eta \eps^{1/3} \|\mu\|_2)$}
    \begin{code}
        \algitem Let $\ell_2 \gets \proc{estimate-L2-moments}(\eta/2; \mu, 1/1000)$.
        \algitem Let $a \gets \sqrt[3]{(9/2)\eps \ell_2^{3/2}}$.
        \algitem Let $\ell_3 \gets \proc{estimate-L3-magnitude}(\eta/2; \mu, a)$.
        \algitem Return $(1+1/200)(\ell_3 + a^3) / \ell_2^2$.
    \end{code}
\end{proc-algo}

\repprovelemma{find-advice-small-mu2}
\begin{proof}
    By Lemma \reflemma{estimate-L2-moments}: the expected cost of estimating $\ell_2$ is $O(\log (1/\eta)/ \|\mu\|_2)$, and $\E[\ell_2] = O(\|\mu\|_2^2)$ and $\E[1/\ell_2] = O(1/\|\mu\|_2^2)$. Also, with probability at least $1-\eta/2$, $\ell_2 \le (1 + 1/1000)\|\mu\|_2^2$, and in this case, $a^3 \ge (9/2) \eps \ell_2^{3/2} \ge 3\eps \|\mu\|_2^3$.

    The expected value of $a$ is:
    \begin{eqnarray*}
        \E[a]
        = (9/2)^{1/3} \eps^{1/3} \E[\sqrt{\ell_2}]
        \le (9/2)^{1/3} \eps^{1/3} \sqrt{\E[\ell_2]}
        = O(\eps^{1/3}\|\mu\|_2)
    \end{eqnarray*}
    
    The expected value of $1/a$ is:
    \begin{eqnarray*}
        \E[1/a]
        = (2/9)^{1/3} \E[\sqrt{1/\ell_2}] / \eps^{1/3}
        \le (2/9)^{1/3} \sqrt{\E[1/\ell_2]} / \eps^{1/3}
        = O(1/\eps^{1/3}\|\mu\|_2)
    \end{eqnarray*}

    The expected cost of estimating $\ell_3$ is $O(\E[1/a] / \eta) = O(1 / \eta \eps^{1/3} \|\mu\|_2)$, and $\E[\ell_3] = \|\mu\|_3^3$. Additionally, with probability at least $1-\eta/2$, $\ell_3 \in \|\mu\|_3^3 \pm \max\{ \|\mu\|_3^3/1000, a^3 \}$.
    
    And:
    \begin{eqnarray*}
        X
        &=& (1 + 1/200) \frac{(\ell_3 + a^3)}{\ell_2^2} \\
        &\ge& (1 + 1/200) \frac{(((1-1/1000)\|\mu\|_3^3 - a^3) + a^3}{(1+1/1000)^2 \|\mu\|_2^4}
        \ge \frac{(1+1/200)(1-1/1000)}{(1+1/1000)^2} \frac{\|\mu\|_3^3}{\|\mu\|_2^4}
        \ge \frac{\|\mu\|_3^3}{\|\mu\|_2^4} \ge t_\mu
    \end{eqnarray*}

    For the expected value,
    \[  \E[X]
        = O((\E[\ell_3] + \E[a]) \E[1/\ell_2^2])
        = O((\|\mu\|_3^3 + \eps\|\mu\|_2^3) / \|\mu\|_2^4)
        = O((t + 1) + \eps/\|\mu\|_2) \]
\end{proof}

\section{Finding an advice when $\|\mu\|_2$ is medium}
\label{sec:medium-mu2}

In this section we provide the logic of estimating $t_\mu$ directly by its definition, $t_\mu = \|\mu\|_3^3 / \|\mu\|_2^4 - 1$, and use this estimation to find an advice when $\|\mu\|_2$ is considered ``medium'' with respect to $\eps$.

\paragraph{Direct estimation of $t$}
\defproc{estimate-t-directly}{Estimate-$t$-Directly}

We estimate $t_\mu = \frac{\|\mu\|_3^3}{(\|\mu\|_2^2)^2} - 1$ directly by estimating $\|\mu\|_2^2$ and $\|\mu\|_3^3$. The result is between $\Omega(t_\mu - \delta)$ and $O(t_\mu + \delta)$ with high probability, but can be negative. Therefore, we use the maximum of the result value and $0$. Unfortunately, this can increase the expected value of the output by $\Omega(1)$. To reduce the additive penalty in the expected output to $O(\delta)$, we amplify the success probability to $O(\log \delta^{-1})$ and take the minimum with an alternative non-negative random variable whose expected value is $O(t_\mu + 1)$. \algforprocshort{estimate-t-directly}

\begin{proc-algo}{estimate-t-directly}{\eta; \mu, \delta}
    \alginput{$0 < \delta \le 1$}
    \algoutput{$X \ge t_\mu$ with probability $\ge 1-\eta$}
    \algmoments{$\E[X] = O(t_\mu + \delta)$}
    \algcomplexity{$O(\log (1/\eta \delta) / \delta^2 \|\mu\|_2^{4/3})$}
    \begin{code}
        \algitem Let $\eta' = \min\{\eta, \delta\}$.
        \algitem Let $\ell_{22} \gets \proc{estimate-L2-moments}(\eta'/4; \mu, \delta/40)$.
        \algitem Let $\ell_{33} \gets \proc{estimate-L3}(\eta'/4; \mu, \delta/30)$.
        \algitem Let $Y_1 \gets \max\{0, \ell_{33} / (\ell_{22})^2 - 1\}$.
        \algitem Let $\ell'_{22} \gets \proc{estimate-L2-moments}(\eta'/4; \mu, 1/2)$.
        \algitem Let $\ell'_{33} \gets \proc{estimate-L3-amplified}(\eta'/4; \mu, 1/2)$.
        \algitem Let $Y_2 \gets (2\ell'_{33}) / (\ell'_{22}/(3/2))^2$.
        \algitem Return $\min\{2(Y_1 + \delta), Y_2\}$.
    \end{code}
\end{proc-algo}

\begin{lemma}{estimate-t-directly--combined-estimation}
    For $0 < \delta \le 1$, let $\ell_{22} \in (1 \pm \delta/40)\|\mu\|_2^2$ and $\ell_{33} \in (1 \pm \delta/30)\|\mu\|_3^3$. In this setting, $\ell_{33} / \ell_{22}^2 - 1 \in [t_\mu/2 - \delta, 2t_\mu + \delta]$.
\end{lemma}
\begin{proof}
    Recall that $t_\mu = \|\mu\|_3^3 / \|\mu\|_2^4 - 1$.

    For the upper bound:
    \[  \frac{\ell_{33}}{\ell_{22}^2} - 1
        \le \frac{(1 + \delta/30)\|\mu\|_3^3}{(1 - \delta/40)^2 \|\mu\|_2^4} - 1
        \le \left(1 + \frac{1}{3}\delta\right) \frac{\|\mu\|_3^3}{\|\mu\|_2^4} - 1
        = \left(1 + \frac{1}{3}\delta\right)\left(\frac{\|\mu\|_3^3}{\|\mu\|_2^4} - 1\right) + \frac{1}{3}\delta
        \le 2t_\mu + \delta
    \]

    For the lower bound:
    \[  \frac{\ell_{33}}{\ell_{22}^2} - 1
        \ge \frac{(1 - \delta/30)\|\mu\|_3^3}{(1 + \delta/40)^2\|\mu\|_2^4} - 1
        \ge \left(1 - \frac{1}{3}\delta\right)\frac{\|\mu\|_3^3}{\|\mu\|_2^4} - 1
        = \left(1 - \frac{1}{3}\delta\right)\left(\frac{\|\mu\|_3^3}{\|\mu\|_2^4} - 1\right) - \frac{1}{3}\delta
        \ge \frac{1}{2} t_\mu - \delta \]
\end{proof}

\begin{lemma}{estimate-t-directly}
    Let $X$ be the output of $\proc{estimate-t-directly}(\eta; \mu, \delta)$. For every error parameter $\eta$, an input distribution $\mu$ and $ 0 < \delta \le 1$, considering $t_\mu = (\|\mu\|_3^3 - \|\mu\|_2^4) / \|\mu\|_2^4$:
    \begin{itemize}
        \item With probability at least $1 - \eta$, $t \le X \le 15t_\mu + 3\delta$.
        \item The expected output is $\E[X] = O(t + \delta)$.
        \item The expected sample complexity is $O(\log (1/\eta \delta) / \delta^2 \|\mu\|_2^{4/3})$.
    \end{itemize}
\end{lemma}
\begin{proof}
    Recall that $t_\mu = \|\mu\|_3^3 / \|\mu\|_2^4 - 1$.

    For the upper bound:
    \[  \frac{\ell_{33}}{\ell_{22}^2} - 1
        \le \frac{(1 + \delta/30)\|\mu\|_3^3}{(1 - \delta/40)^2 \|\mu\|_2^4} - 1
        \le \left(1 + \frac{1}{3}\delta\right) \frac{\|\mu\|_3^3}{\|\mu\|_2^4} - 1
        = \left(1 + \frac{1}{3}\delta\right)\left(\frac{\|\mu\|_3^3}{\|\mu\|_2^4} - 1\right) + \frac{1}{3}\delta
        \le 2t_\mu + \delta
    \]

    For the lower bound:
    \[  \frac{\ell_{33}}{\ell_{22}^2} - 1
        \ge \frac{(1 - \delta/30)\|\mu\|_3^3}{(1 + \delta/40)^2\|\mu\|_2^4} - 1
        \ge \left(1 - \frac{1}{3}\delta\right)\frac{\|\mu\|_3^3}{\|\mu\|_2^4} - 1
        = \left(1 - \frac{1}{3}\delta\right)\left(\frac{\|\mu\|_3^3}{\|\mu\|_2^4} - 1\right) - \frac{1}{3}\delta
        \ge \frac{1}{2} t_\mu - \delta \]
\end{proof}

\begin{proof}
    For correctness: by the union bound, with probability at least $1-4(\eta'/4) \ge 1-\eta$:
    \begin{itemize}
        \item $\ell_{22} \in (1 \pm \delta/40)\|\mu\|_2^2$.
        \item $\ell_{33} \in (1 \pm \delta/30)\|\mu\|_3^3$.
        \item $\ell'_{22} \le (3/2) \|\mu\|_2^2$.
        \item $\ell'_{33} \ge (1/2) \|\mu\|_3^3$.
    \end{itemize}
    Therefore, by Lemma \reflemma{estimate-t-directly--combined-estimation}, $Y_1 \ge \max\{0, \frac{1}{2}t_\mu - \delta\}$ and $Y_1 \le 2t_\mu + \delta$. Also, by directly using the definition, $Y_2 \ge (2 \cdot (1/2)\|\mu\|_3^3) / ((3/2)\|\mu\|_2^2 / (3/2))^2 = \|\mu\|_3^3 / \|\mu\|_2^4 = t_\mu + 1$. Combined, $\min\{2(Y_1 + \delta), Y_2\} \ge t_\mu$.

    For expected output: since $Y_1$ and $Y_2$ are independent,
    \begin{eqnarray*}
        \E[X]
        = \E[\min\{2(Y_1 + \delta), Y_2\}]
        &\le& (4t_\mu + 4\delta) + \Pr[Y_1 > 2t_\mu + \delta]\E[Y_2] \\
        &\le& (4t_\mu + 4\delta) + \eta' \cdot \E[\ell'_{33}]\E[1/(\ell'_{22})^2] \\
        (*) &\le& (4t_\mu + 4\delta) + \delta \cdot O(\|\mu\|_3^3) \cdot O(1/\|\mu\|_2^4) \\
        &=& (4t_\mu + 4\delta) + O(\delta) \cdot O(t_\mu + 1)
        = O(t_\mu + \delta)
    \end{eqnarray*}
    $(*)$: for this transition we use Observation \ref{obs:estimate-L3-amplified} (preserving the first moment of $\|\mu\|_3^3$) and Lemma \reflemma{estimate-L2-moments} (preserving the second negative moment of $\|\mu\|_2^2$).
    
    For complexity, we consider the cost of each estimation:
    \begin{itemize}
        \item $\ell_{22}$: $O(\log (1/\eta \delta) / \delta^2 \|\mu\|_2)$ (Lemma \reflemma{estimate-L2-moments}).
        \item $\ell_{33}$: $O(\log (1/\eta \delta) / \delta^2 \|\mu\|_2^{4/3})$ (Observation \ref{obs:estimate-L3-amplified}).
        \item $\ell'_{22}$: $O(\log (1/\eta \delta) / \|\mu\|_2)$ (Lemma \reflemma{estimate-L2-moments}).
        \item $\ell'_{33}$: $O(\log (1/\eta \delta) / \|\mu\|_2^{4/3})$ (Observation \ref{obs:estimate-L3-amplified}).
    \end{itemize}
\end{proof}

\paragraph{Finding an advice}
First, we roughly estimate $\|\mu\|_2^2$ and use it to define $\delta \approx \eps/\|\mu\|_2$ (both with high probability and in expectation). Then we use this $\delta$ to estimate $t_\mu$ through \proc{estimate-t-directly}. \algforprocshort{find-advice-medium-mu2}

\begin{proc-algo}{find-advice-medium-mu2}{\eta; \mu, \eps}
    \algoutput{$X \ge t_\mu$ with probability $\ge 1-\eta$}
    \algmoments{$\E[X] = O(t + \eps/\|\mu\|_2)$}
    \algcomplexity{$O(\log \frac{1}{\eta \eps} \cdot (\|\mu\|_2^{2/3} / \eta \eps^2 + 1/\eta \|\mu\|_2^{4/3}))$}
    \begin{code}
        \algitem $\ell_2 \gets \proc{estimate-L2-moments}(1/3; \mu, 1/10)$.
        \algitem Let $\delta \gets \min\{1, \eps / \sqrt{\ell_2}\}$.
        \algitem Return $\proc{estimate-t-directly}(\eta; \mu, \delta)$.
    \end{code}
\end{proc-algo}

\repprovelemma{find-advice-medium-mu2}
\begin{proof}
    For correctness, observe that by Lemma \reflemma{estimate-t-directly}, the result is not-smaller than $t_\mu$ with probability at least $1-\eta$, regardless of $\delta$.

    For expected value:
    \[  \E[X]
        = O(t_\mu + \E[\delta])
        = O(t_\mu + \eps \E[1/\sqrt{\ell_2}]
        = O(t_\mu + \eps / \|\mu\|_2)
    \]
    Where the first transition is by Lemma \reflemma{estimate-t-directly} and the second transition is by Lemma \reflemma{estimate-L2-moments} (preserving the negative moments).

    For complexity: by Lemma \reflemma{estimate-L2-moments}, the cost of estimating $\delta$ is $O(1/\|\mu\|_2)$. By Lemma \reflemma{estimate-t-directly}, the expected sample complexity of the $t$-estimation call is:
    \begin{eqnarray*}
        \E\left[\frac{\log (1/\eta \delta)}{\delta^2 \|\mu\|_2^{4/3}}\right]
        \le \frac{\log (1/\eta)}{\|\mu\|_2^{4/3}} + \frac{\log (1/\eta \eps)}{\|\mu\|_2^{4/3}} \cdot \frac{\E[\ell_2]}{\eps^2}
        = O\left(\frac{\log (1/\eta)}{\|\mu\|_2^{4/3}} + \frac{\log (1/\eta \eps)\|\mu\|_2^{2/3}}{\eps^2}\right)
    \end{eqnarray*}
    Where in the second transition we use Lemma \reflemma{estimate-L2-moments} (preserving the first moment).
\end{proof}

\section{Distributions over a finite domain}
\label{sec:finite-domain}

\subsection{Algebraic behavior of $t$}
\label{sec:finite-domain:subsec:algebraic-behavior}

In this subsection we prove the statements about the algebraic behavior of $t$.

\repprovelemma{mu22-explicit-by-deltas}
\begin{proof}
    Observe that $\sum_{i=1}^N \delta_i = 0$, and therefore,
    \begin{eqnarray*}
        \|\mu\|_2^2
        = \sum_{i=1}^N \frac{(1+\delta_i)^2}{N^2}
        = \frac{1}{N} + \frac{1}{N^2} \sum_{i=1}^N \delta_i^2
        = \frac{1}{N}\left(1 + \frac{1}{N} \sum_{i=1}^N \delta_i^2\right)
    \end{eqnarray*}
\end{proof}

\begin{lemma}{mu33-explicit-by-deltas}
    $\|\mu\|_3^3 = \frac{1}{N^2}\left(1 + \frac{3}{N}\sum_{i=1}^N \delta_i^2 + \frac{1}{N}\sum_{i=1}^N \delta_i^3\right)$
\end{lemma}
\begin{proof}
    Observe that $\sum_{i=1}^N \delta_i = 0$, and therefore,
    \[\|\mu\|_3^3
        = \sum_{i=1}^N \frac{(1+\delta_i)^3}{N^3}
        = \frac{1}{N^2} + \frac{3}{N^3} \sum_{i=1}^N \delta_i^2 + \frac{1}{N^3} \sum_{i=1}^N \delta_i^3
        = \frac{1}{N^2}\left(1 + \frac{3}{N} \sum_{i=1}^N \delta_i^2 + \frac{1}{N} \sum_{i=1}^N \delta_i^3 \right)
    \]
\end{proof}

\repprovelemma{t-explicit-by-deltas}
\begin{proof}
    Just arithmetics over Lemma \reflemma{mu22-explicit-by-deltas} and Lemma \reflemma{mu33-explicit-by-deltas}.
\end{proof}

\subsection{Estimating the sum of squares}
\label{sec:finite-domain:subsec:sum-of-squares}

Let $k = \ceil{\log \eps^{-1}}$ be the number of iterations in the algorithm. For $0 \le i \le k$, we define a sequence of resolutions $\eps_i = 2^{-i}$, and for $1 \le i \le k$ we also define a sequence $t_i = \eps_{i-1} \sqrt{N}$.

Before the first iteration we estimate $\|\mu\|_2^2$ within a $(1 \pm 1/12)$-multiplicative factor. If the result indicates that $\|\mu\|_2^2 = (1 + \Omega(1))/N$, then we deduce that $\frac{1}{N} \sum_{i=1}^N \delta_i = \Theta(N\|\mu\|_2^2)$, and use an additional estimation of $\|\mu\|_2^2$.

For $1 \le i \le k$, the $i$th iteration considers the $i-1$st estimation and distinguishes between the cases where $\|\mu\|_2^2 < (1+\frac{3}{4}\eps_{i-1})/N$ and $\|\mu\|_2^2 \ge (1+\frac{1}{4}\eps_{i-1})/N$ (note the overlap).

If we believe that $\|\mu\|_2^2 < (1 + \frac{3}{4}\eps_{i-1}) / N$, then we estimate $\|\mu\|_2^2$ again, now with the better accuracy $\eps_i / 12$. We also observe that, in this case, $t_\mu \le t_i$ (see Lemma \reflemma{bound-iteration-t-by-iteration-mu22}), which reduces the growing rate of the $t/\eps^2$-part of the estimation complexity. More precisely, since $t_i = O(\eps_{i-1}\sqrt{N}) = O(\eps_i \sqrt{N})$, the term $t_i / \eps_i^2$ is $O(\sqrt{N}/\eps_i)$, whose dependence on $1/\eps_i$ is only linear rather than quadratic.

If we believe that $\|\mu\|_2^2 \ge (1 + \frac{1}{4}\eps_{i-1}) / N$ for the first time, which means that $\|\mu\|_2^2 < (1 + \frac{3}{4}\eps_{i-1}) / N$ as well, then we deduce that $\|\mu\|_2^2 = (1 + \Theta(\eps_{i-1})) / N$. Therefore, we return $2\eps_{i-1}$ as our estimation. If the algorithm does not terminate within $k$ iterations, then we deduce that $\|\mu\|_2^2 = (1 + O(\eps))/N$, and return $2\eps_k \in [\eps, 2\eps]$ as our estimation.

\begin{proc-algo}{estimate-sum-squares}{\eta; \mu, \eps}
    \alginput{A friendly distribution $\mu$ over $\Omega = \{1,\ldots,N\}$}
    \algoutput{$X \ge \frac{1}{N}\sum_{i=1}^N \delta_i^2$ with probability $\ge 1-\eta$}
    \algmoments{$\E[X] = O(\frac{1}{N}\sum_{i=1}^N \delta_i^2 + \eps)$}
    \algcomplexity{$O(\log (1/\eta) \cdot (1/\eps\|\mu\|_2 + t/\eps^2 + \log(1/\eps\|\mu\|_2)/\|\mu\|_2^2))$}
    \begin{code}
        \algitem Let $\eps_0 \gets 1$.
        \algitem Let $p_0 \gets \proc{estimate-L2-BC}(3^{-k} \eta/32; \mu, \eps_0 / 12, \bot)$.
        \begin{If}{$p_0 \ge (1 + 1/4)/N$}
            \algitem Let $p' \gets \proc{estimate-L2-BC}(\min\{\eta/2,\eps\}; \mu, 1/12, \bot)$.
            \algitem Return $\max\{0, 3(Np' - 1)\}$.
        \end{If}
        \algitem Let $k \gets \ceil{\log_2 (1/\eps)}$.
        \begin{For}{$i$ from $1$ to $k$}
            \algitem Let $\eps_i \gets \eps_{i-1} / 2$.
            \begin{If}{$p_{i-1} < \frac{1 + \frac{1}{2}\eps_{i-1}}{N}$}
                \algitem Let $t_i \gets \eps_{i-1} \sqrt{N}$.
                \algitem Let $p_i \gets \proc{estimate-L2-BC}(3^{i-k} \eta / 32; \mu, \eps_i / 12, t_i)$.
            \end{If}
            \begin{Else}
                \algitem Return $2 \eps_{i-1}$.
            \end{Else}
        \end{For}
        \algitem Return $2 \eps_k$. \algcomment{(The ``$k+1$st'' iteration)}
    \end{code}
\end{proc-algo}

\begin{lemma}{bound-iteration-t-by-iteration-mu22}
    Consider \proc*{estimate-sum-squares}. If, for some $1 \le i \le k$, we believe that $\|\mu\|_2^2 \le \frac{1 + \frac{3}{4}\eps_{i-1}}{N}$, then we can deduce the belief that $t_\mu \le t_i$.
\end{lemma}
\begin{proof}
    By Lemma \reflemma{t-explicit-by-deltas},
    \[  t
        = \frac{\|\mu\|_3^3}{\|\mu\|_2^4} - 1
        = \frac{\frac{1}{N}\left(\sum_{i=1}^N \delta_i^2 + \sum_{i=1}^N \delta_i^3 - \frac{1}{N}\left(\sum_{i=1}^N \delta_i^2\right)^2\right)}{\left(1 + \frac{1}{N}\sum_{i=1}^N \delta_i^2\right)^2} \]

    For the numerator:
    \begin{eqnarray*}
        \frac{1}{N}\left(\sum_{i=1}^N \delta_i^2 + \sum_{i=1}^N \delta_i^3 - \frac{1}{N}\left(\sum_{i=1}^N \delta_i^2\right)^2\right)
        &\le& \frac{1}{N}\left(\sum_{i=1}^N \delta_i^2 + (\max_i \delta_i) \sum_{i=1}^N \delta_i^2\right) \\
        &=& \frac{1}{N}\left(\sum_{i=1}^N \delta_i^2 + (N \max_ \mu(i) - 1) \sum_{i=1}^N \delta_i^2\right) \\
        &=& \frac{1}{N} (N \max_i \mu(i)) \cdot \sum_{i=1}^N \delta_i^2
        = \max_ \mu(i) \cdot \sum_{i=1}^N \delta_i^2
        \le \|\mu\|_2 \sum_{i=1}^N \delta_i^2
    \end{eqnarray*}

    For the denominator, by Lemma \reflemma{mu22-explicit-by-deltas}:
    \begin{eqnarray*}
        \left(1 + \frac{1}{N}\sum_{i=1}^N \delta_i^2\right)^2
        = N\|\mu\|_2^2 \cdot \left(1 + \frac{1}{N}\sum_{i=1}^N \delta_i^2\right)
        \ge N\|\mu\|_2^2 \cdot 1
    \end{eqnarray*}

    For the whole expression, using $\|\mu|_2^2 \ge 1/N$:
    \begin{eqnarray*}
        t
        \le \frac{\|\mu\|_2 \sum_{i=1}^N \delta_i^2}{N\|\mu\|_2^2}
        = \frac{1}{N\|\mu\|_2} \cdot \sum_{i=1}^N \delta_i^2
        = \frac{1}{\|\mu\|_2} \cdot \frac{1}{N}\sum_{i=1}^N \delta_i^2
        \le \sqrt{N} \cdot \frac{3}{4}\eps_{i-1}
        < t_i
    \end{eqnarray*}
\end{proof}

\begin{lemma}{estimate-t-friendly-termination-goal}
    Let $i^*$ be the smallest non-negative integer for which $(1 - 2^{-i^*} / 12) \|\mu\|_2^2 \ge (1 + 2^{-i^*} / 2) / N$, or infinite if $\|\mu\|_2^2 = 1/N$. With probability at least $1 - \eta/6$, \proc{estimate-sum-squares} terminates after at most $i^*+1$ iterations, and for every $0 \le i \le i_\mathrm{max}-1$, where $i_\mathrm{max} \le i^* + 1$ is the iteration causing termination, $p_i \in (1 \pm \eps_i / 12)\|\mu\|_2^2$.
\end{lemma}
\begin{proof}
    For every $0 \le i \le i^*-1$,
    \[  \|\mu\|_2^2
        <   \frac{1 + 2^{-i}/2}{1 - 2^{-i}/12} \cdot \frac{1}{N}
        \le \frac{1 + \frac{3}{4} \cdot 2^{-i}}{N}
        =   \frac{1 + \frac{3}{4} \eps_i}{N}
    \]
    By Lemma \reflemma{bound-iteration-t-by-iteration-mu22}, $t_\mu \le t_i$ for every $1 \le i \le i^*$.

    The probability that $p_i \notin (1 \pm \eps_i/12)\|\mu\|_2^2$ even for one $0 \le i \le \min\{i^*,k\}$ is bounded by $\sum_{i=0}^{\min\{i^*,k\}} 3^{i-k}\eta / 32 \le \eta/6$.

    In the $i^* + 1$st iteration, if it is executed (in particular, we can assume that $i^*+1 \le k$), we have
    \[ p_{i^*}
        \ge (1 - \eps_{i^*}/12) \|\mu\|_2^2
        =   (1 - 2^{-i^*}/12) \|\mu\|_2^2
        \ge (1 + 2^{-i^*} / 2) / N
        =   (1 + \eps_{i^*} / 2) / N \]
    And therefore, the algorithm branches into an estimation of $t_\mu$ (through $s$) and terminates.
\end{proof}

\begin{lemma}{estimate-sum-squares--correctness}
    Let $\mu$ be a distribution over $N$ elements. For every $0 < \eta \le 1/3$ and $0 < \eps < 1$, the output of $\proc*{estimate-sum-squares}(\eta; \mu, \eps)$ is greater than $X \ge \frac{1}{N}\sum_{i=1}^N \delta_i^2$ with probability at least $1-\eta$.
\end{lemma}
\begin{proof}
    Let $i^*$ be the smallest non-negative integer for which $(1 - 2^{-i^*} / 12) \|\mu\|_2^2 \ge (1 + 2^{-i^*} / 2) / N$, or infinite if $\|\mu\|_2^2 = 1/N$. By Lemma \reflemma{estimate-t-friendly-termination-goal}, with probability at least $1-\eta/6$, the algorithm terminates after executing the $i_\mathrm{max}$th iteration for some $i_\mathrm{max} \le i^*+1$, and additionally, $p_0,\ldots,p_{i_\mathrm{max}-1}$ are all in their desired ranges.

    If $p_0 \ge (1 + 1/4)/N$, then we deduce that $\|\mu\|_2^2 \ge (15/13)/N$. In this case, we return:
    \begin{eqnarray*}
        3(Np' - 1)
        &\ge& 3((1 - 1/12)N\|\mu\|_2^2 - 1) \\
        &=& 3((1 - 1/12)(N\|\mu\|_2^2 - 1) - 1/12) \\
        &\ge& 3((1 - 1/12)(N\|\mu\|_2^2 - 1) - (13/24)(N\|\mu\|_2^2 - 1)) \\
        &=& 3 \cdot (1 - 1/12 - 13/24)(N\|\mu\|_2^2 - 1)
        \ge (N\|\mu\|_2^2 - 1)
    \end{eqnarray*}
    
    If $\|\mu\|_2^2 \ge 2/N$, then $p_0$ should be greater than $(1+1/4)/N$. The rest of the analysis assumes that $\|\mu\|_2^2 \le 2/N$.

    If $i^* = 0$, then we deduce that:
    \[  \frac{1}{N}\sum_{i=1}^N \delta_i^2
        = N\|\mu\|_2^2 - 1
        \le N \cdot (2/N) - 1
        = 1
        = \eps_{i^*} \]
    
    If $i^* \ge 1$, then we deduce that:
    \[  \frac{1}{N}\sum_{i=1}^N \delta_i^2
        = N\|\mu\|_2^2 - 1
        \le 2^{-(i^*-1)}
        = 2\eps_{i^*} \]

    Since the terminating iteration is at most $i^* + 1$, its output is at least $2 \eps_{(i^*+1)-1} = 2 \eps_{i^*} \ge (N\|\mu\|_2^2-1)$.
\end{proof}

\begin{lemma}{estimate-sum-squares--expected-value}
    Let $\mu$ be a distribution over $N$. For every $0 < \eta \le 1/3$ and $0 < \eps < 1$, the expected output of $\proc*{estimate-sum-squares}(\eta; \mu, \eps)$ is $O(\frac{1}{N}\sum_{i=1}^N \delta_i^2 + \eps)$.
\end{lemma}
\begin{proof}
    We first consider the contribution of the large-$p_0$ branch (``if $p_0 \ge (1+1/4)/N$'') to the expected output.

    If $\|\mu\|_2^2 \ge (1 + 1/12)/N$, then $N\|\mu\|_2^2 = O(N\|\mu\|_2^2 - 1)$. Hence, by Lemma \reflemma{estimate-L2-BC}:
    \[  \E[\max\{0, 3(Np' - 1)\}]
        \le 3N\E[p']
        = 3N \cdot O(\|\mu\|_2^2)
        = O(N\|\mu\|_2^2)
        = O(N\|\mu\|_2^2 - 1)
    \]

    If $\|\mu\|_2^2 < (1 + 1/12)/N$, then the probability to enter the large-$p_0$ branch is bounded by $\eps$. If we enter this branch, then the expected output is bounded by $3N\|\mu\|_2^2 \le 4$. Therefore, the contribution of the large-$p_0$ branch is bounded by $\eps \cdot 4 = O(\eps)$.
    
    Therefore, it suffices to show that the expected output of the main branch is $O(\frac{1}{N}\sum_{i=1}^N \delta_i^2 + \eps)$.

    Let $i^*$ be the smallest non-negative integer for which $(1 - 2^{-i^*} / 12) \|\mu\|_2^2 \ge (1 + 2^{-i^*} / 2) / N$, or infinite if $\|\mu\|_2^2 = 1/N$.

    For every $0 \le i \le i^*-1$,
    \[  \|\mu\|_2^2
        <   \frac{1 + 2^{-i}/2}{1 - 2^{-i}/12} \cdot \frac{1}{N}
        \le \frac{1 + \frac{3}{4} \cdot 2^{-i}}{N}
        =   \frac{1 + \frac{3}{4} \eps_i}{N} \]
    And therefore, by Lemma \reflemma{estimate-L2-BC}, for every $1 \le i \le i^* + 1$,
    \[  \Pr\left[p_{i-1} \ge \frac{1+\frac{1}{2}\eps_{i-1}}{N} \right]
        \le \Pr\left[p_{i-1} \ge (1 + \eps_{i-1} / 12)\|\mu\|_2^2 \right]
        \le 3^{(i-1)-k}\eta/32 \]

    The expected output is bounded by:
    \begin{eqnarray*}
        \E[X]
        &\le& \sum_{i=1}^k 2\eps_{i-1} \Pr\left[p_{i-1} \ge \frac{1+\frac{1}{2}\eps_{i-1}}{N}\right] + 2\eps_k \\
        &\le& 2\sum_{i=1}^{\min\{i^*+1,k\}} \eps_{i-1} \cdot \frac{3^{(i-1)-k}\eta}{32} + \sum_{i=\min\{i^*+1,k\}+1}^k \eps_{i-1} + O(\eps) \\
        &\le& \frac{3^{-k}}{16}\eta \cdot \sum_{i=1}^{\min\{i^*+1,k\}} 2^{-(i-1)} \cdot 3^{i-1} + \sum_{i=\min\{i^*+1,k\}+1}^k \eps_{i-1} + O(\eps) \\
        &\le& \frac{3^{-k}}{16}\eta \cdot \sum_{i=1}^{\min\{i^*+1,k\}} (3/2)^{i-1} + O(2^{-\min\{i^*,k\}}) + O(\eps) \\
        &\le& \frac{3^{-k}}{8}\eta \cdot (3/2)^{\min\{i^*+1,k\}} + O(2^{-i^*} + \eps)
    \end{eqnarray*}

    If $i^* \le k-1$, then:
    \begin{eqnarray*}
        \E[X]
        &\le& \frac{3^{-k}}{8}\eta \cdot (3/2)^{i^*+1} + O(2^{-i^*} + \eps) \\
        &\le& \frac{1}{8}\eta \cdot 2^{-(i^*+1)} + O(2^{-i^*} + \eps)
        = O((1+\eta)2^{-i^*} + \eps)
        = O(2^{-i^*} + \eps)
    \end{eqnarray*}

    If $i^* \ge k$, then:
    \begin{eqnarray*}
        \E[X]
        &\le& \frac{3^{-k}}{8}\eta \cdot (3/2)^{k} + O(2^{-i^*} + \eps) \\
        &\le& O(2^{-k}\eta + 2^{-i^*} + \eps)
        = O(\eps \eta + 2^{-i^*} + \eps)
        = O(2^{-i^*} + \eps)
    \end{eqnarray*}
    
    In both cases, $\E[X] = O(2^{-i^*} + \eps)$. By Lemma \reflemma{mu22-explicit-by-deltas}, $\frac{1}{N}\sum_{i=1}^N \delta_i^2 = N\|\mu\|_2^2 - 1 = \Omega(2^{-i^*})$, and therefore, $\E[X] = O(\frac{1}{N}\sum_{i=1}^N + \eps)$ as required.
\end{proof}

\begin{lemma}{estimate-sum-squares--complexity}
    Let $\mu$ be a distribution over $N$. For every $0 < \eta \le 1/3$ and $0 < \eps < 1$, the expected sample complexity of $\proc*{estimate-sum-squares}(\eta; \mu, \eps)$ is $O\left(\log \frac{1}{\eta} \cdot \left(\frac{\sqrt{N}}{\eps} + \frac{1}{\eps\|\mu\|_2}\right)\right)$.
\end{lemma}
\begin{proof}
    The sample complexity required for $p_0$ is:
    \[  O(\log (3^k/\eta)/\|\mu\|_2)
        = O((k + \log (1/\eta))/\|\mu\|_2)
        = O(\log (1/\eps \eta) / \|\mu\|_2) \]

    If we execute the large-$p_0$ branch (``if $p_0 \ge (1 + 1/4)/N$''), then the sample complexity of the additional estimation is $O(\log (1/\eps\eta) / \|\mu\|_2)$ (Lemma \reflemma{estimate-L2-BC}), which can be relaxed to $O(\log (1/\eta) / \eps \|\mu\|_2)$.
    
    For $1 \le i \le k$, the expected query complexity of the $i$th iteration (if it is executed) is:
    \begin{eqnarray*}
        O(\log (32 \cdot 3^{k-i}/\eta)) \cdot \left(\frac{1}{\eps_i \|\mu\|_2} + \frac{t_i}{\eps_i^2} \right)
        &=& O((1 + k - i) \log (1/\eta)) \cdot \left(\frac{1}{\eps_i \|\mu\|_2} + \frac{\sqrt{N}}{\eps_i} \right) \\
        &=& O\left(\log (1/\eta) \cdot (1 + k - i) \cdot \left(\sqrt{N} + \frac{1}{\|\mu\|_2}\right) \cdot \frac{1}{\eps_i}\right)
    \end{eqnarray*}

    The cost of the loop is bounded by:
    \begin{eqnarray*}
        O\left(\log \frac{1}{\eta} \cdot \left(\sqrt{N} + \frac{1}{\|\mu\|_2}\right)\right) \cdot \sum_{i=1}^k \frac{1 + k - i}{\eps_i}
        &=& O\left(\log \frac{1}{\eta} \cdot \left(\sqrt{N} + \frac{1}{\|\mu\|_2}\right)\right) \cdot \sum_{i=1}^k \frac{i}{\eps_{k+1-i}} \\
        &=& O\left(\log \frac{1}{\eta} \cdot \left(\sqrt{N} + \frac{1}{\|\mu\|_2}\right)\right) \cdot \sum_{i=1}^k 2^{k+1-i} i \\
        &=& O\left(\log \frac{1}{\eta} \cdot \left(\sqrt{N} + \frac{1}{\|\mu\|_2}\right) \cdot 2^k\right) \cdot \sum_{i=1}^k 2^{-i} i \\
        \text{[Since $\sum_{i=1}^\infty 2^{-i}i = 2$]} &=& O\left(\log \frac{1}{\eta} \cdot \left(\sqrt{N} + \frac{1}{\|\mu\|_2}\right) \cdot 2^k\right) \\
        \text{[Since $k = \log \eps^{-1} + O(1)$]} &=& O\left(\log \frac{1}{\eta} \cdot \left(\frac{\sqrt{N}}{\eps} + \frac{1}{\eps\|\mu\|_2}\right)\right)
    \end{eqnarray*}
\end{proof}

\begin{lemma}{estimate-sum-squares}
    Let $X$ be the output of $\proc*{estimate-sum-squares}(\eta; \mu, \eps)$ for a distribution $\mu$ over $N$ elements and $0 < \eps < 1$.
    \begin{itemize}
        \item $X \ge \frac{1}{N} \sum_{i=1}^N \delta_i^2$ with probability $1-\eta$.
        \item $\E[X] = O(\frac{1}{N} \sum_{i=1}^N \delta_i^2 + \eps)$.
        \item The expected sample complexity is $O\left(\log \frac{1}{\eta} \cdot \left(\frac{\sqrt{N}}{\eps} + \frac{1}{\eps\|\mu\|_2}\right)\right)$.
    \end{itemize}
\end{lemma}

\begin{proof}
    See Lemma \reflemma{estimate-sum-squares--correctness} (correctness), Lemma \reflemma{estimate-sum-squares--expected-value} (expected output) and Lemma \reflemma{estimate-sum-squares--complexity} (complexity).
\end{proof}

\subsection{Estimating the sum of large cubes}
\label{sec:friendly:subsec:sum-of-cubes}

\defproc{estimate-sum-cubes}{Estimate-Sum-Cubes}

To estimate the sum of large cubes, we learn $\mu$ using $O(N \log N)$ samples. This suffices to detect all elements with $\delta_i < 1/4$ (to ignore) and a superset of all elements with $\delta_i \ge 1$ (to consider), for which we also obtain a multiplicative bound.

\begin{proc-algo}{estimate-sum-cubes}{\eta; \mu}
    \alginput{A distribution $\mu$ over $\Omega = \{1,\ldots,N\}$}
    \algoutput{If $X \ge \frac{1}{N}\sum_{i : \delta_i \ge 1} \delta_i^3$ with probability $\ge 1-\eta$}
    \algmoments{$\E[X] = O(\frac{1}{N}\sum_{i=1}^N \delta_i^2 +\frac{1}{N}\sum_{i : \delta_i \ge 1} \delta_i^3 + \eta)$}
    \algcomplexity{$O(N \log (N/\eta))$}
    \begin{code}
        \algitem Let $q \gets \ceil{1000N \ln (N^4/\eta)}$.
        \algitem Draw $q$ independent samples from $\mu$.
        \begin{For}{every $i \in \Omega$}
            \algitem Let $P_i$ be the number of $i$ samples drawn.
            \algitem Let $\hat p_i \gets P_i / q$.
            \algitem Let $\hat\delta_i \gets N\hat p_i - 1$.
        \end{For}
        \algitem Return $\frac{8}{N} \sum_{i : \hat\delta_i \ge 1/2} \hat\delta_i^3$.
    \end{code}
\end{proc-algo}

\begin{lemma}[Technical lemma]{technical:estimate-sum-cubes--large-delta-chernoff}
    Considering a run of $\proc*{estimate-sum-cubes}(\eta; \mu)$ for $0 < \eta \le 1/3$ and a distribution $\mu$ over $\Omega$ of size $N$. For every $i \in \Omega$ for which $\delta_i > 1/4$, the probability that $\hat{\delta}_i \notin (1 \pm 1/2)\delta_i$ is bounded by $2\eta/N^4$. Moreover, $\E\left[\hat\delta_i^3\right] \le 4\delta_i^3 + 2\eta/N$
\end{lemma}

\begin{lemma}[Technical lemma]{technical:estimate-sum-cubes--small-delta-chernoff}
    Considering a run of $\proc*{estimate-sum-cubes}(\eta; \mu)$ for $0 < \eta \le 1/3$ and a distribution $\mu$ over $\Omega$ of size $N$. For every $i \in \Omega$ for which $\delta_i \le 1/4$, $\E\left[\hat\delta_i^3 \cdot \mathbf{1}_{\hat\delta_i \ge 1/2}\right] \le 2\eta/N$.
\end{lemma}

\begin{lemma}{estimate-sum-cubes}
    Let $X$ be the output of $\proc*{estimate-sum-cubes}(\eta; \mu, \eps)$ for a distribution $\mu$ over $N$ elements and $0 < \eps < 1$.
    \begin{itemize}
        \item $X \ge \frac{1}{N} \sum_{i : \delta_i \ge 1} \delta_i^3$ with probability $1-\eta$.
        \item $\E[X] = O(\frac{1}{N}\sum_{i=1}^N \delta_i^3 + \frac{1}{N}\sum_{i=1}^N \delta_i^2 + \eta)$.
        \item The expected sample complexity is $O(N \log (N/\eta))$.
    \end{itemize}
\end{lemma}

\begin{proof}
    Sample complexity is trivial since the number of samples is explicitly $O(N \log (N/\eta))$.

    For correctness: by Lemma \reflemma{technical:estimate-sum-cubes--large-delta-chernoff}, with probability at least $1 -N \cdot 2\eta/N^4 \ge 1-\eta$ (unless $N=1$, and then all bounds are trivial), $\hat\delta_i \ge \frac{1}{2}\delta_i \ge \frac{1}{2}$ for every $i$ for which $\delta_i \ge 1$, in which case the output is correct since:
    \[  \sum_{i : \hat\delta_i \ge 1/2} \hat\delta_i^3
        \ge \sum_{i : \delta_i \ge 1} (\delta_i / 2)^3
        \ge \frac{1}{8}\sum_{i : \delta_i \ge 1} \delta_i^3 \]

    For the expected value:
    \begin{eqnarray*}
        \E\left[\sum_{i\in \Omega, \hat \delta_i \ge 1/2} \hat\delta_i^3 \right]
        &=& \E\left[\sum_{i\in \Omega} \hat\delta_i^3 \cdot \mathbf{1}_{\hat \delta_i \ge 1/2} \right] \\
        &=& \sum_{i\in \Omega} \E\left[\hat\delta_i^3 \cdot \mathbf{1}_{\hat \delta_i \ge 1/2} \right]
        = \left(\sum_{i\in \Omega, \delta_i > 1/4} \E\left[\hat\delta_i^3\right] + \sum_{i\in \Omega, \delta_i \le 1/4} \E\left[\hat\delta_i^3 \cdot \mathbf{1}_{\hat \delta_i \ge 1/2} \right]\right)
    \end{eqnarray*}

    By Lemma \reflemma{technical:estimate-sum-cubes--large-delta-chernoff} and Lemma \reflemma{technical:estimate-sum-cubes--small-delta-chernoff},
    \begin{eqnarray*}
        \E\left[\sum_{i\in \Omega, \hat \delta_i \ge 1/2} \hat\delta_i^3 \right]
        &\le& 8\left(\sum_{i\in \Omega, \delta_i > 1/4} (4\delta_i^3 + 2\eta/N) + \sum_{i\in \Omega, \delta_i \le 1/4} (\eta/N)\right) \\
        &=& \left(\sum_{i : \delta_i \ge 1} (4\delta_i^3 + 2\eta/N) + \sum_{i : 1/4 < \delta_i < 1} (4\delta_i^3 + 2\eta/N) + \sum_{i : \delta_i \le 1/4} (\eta/N)\right) \\
        (*) &\le& \left(4 \sum_{i : \delta_i \ge 1} \delta_i^3 + 4 \sum_{i : 1/4 \le \delta_i < 1} \delta_i^2 + 2\eta\right)
    \end{eqnarray*}
    $(*)$: Since $\delta_i^3 \le \delta_i^2$ for $\delta_i < 1$.

    As a result, the expected output is bounded by:
    \[ \frac{8}{N} \cdot \left(4 \sum_{i : \delta_i \ge 1} \delta_i^3 + 4 \sum_{i : 1/4 \le \delta_i < 1} \delta_i^2 + 2\eta\right)
        = O\left(\frac{1}{N}\sum_{i=1}^N \delta_i^2 + \frac{1}{N}\sum_{i : \delta_i \ge 1} \delta_i^3 + \eta \right) \]
\end{proof}

\subsection{Deferred proofs of technical lemmas}
The next section begins at Page \pageref{sec:friendly}.

\repprovelemma{technical:estimate-sum-cubes--large-delta-chernoff}
\begin{proof}
    Consider some $i$ for which $\delta_i > 1/4$. By Chernoff's bound,
    \begin{eqnarray*}
        \Pr\left[\hat\delta_i \notin \left(1 \pm \frac{1}{2}\right)\!\delta_i\right]
        &=& \Pr\left[\hat p_i \notin p_i \pm \frac{1}{2N}\delta_i\right] \\
        &=& \Pr\left[\Bin(q, p_i) \notin qp_i\left(1 \pm \frac{\delta_i}{2N \cdot p_i}\right)\right]
        \le 2 e^{-\frac{1}{3} \cdot \frac{\delta_i^2}{4N^2p_i^2} \cdot p_i q}
        = 2 e^{-\frac{1}{12} \cdot \frac{\delta_i^2}{N^2p_i} \cdot q}
    \end{eqnarray*}

    We use $q \ge 1000N \ln (N^4 / \eta)$ and $p_i = (1 + \delta_i) / N$ to obtain that:
    \[  \Pr\left[\hat\delta_i \notin \left(1 \pm \frac{1}{2}\right)\delta_i\right]
        \le 2 e^{-\frac{1}{12} \cdot \frac{\delta_i^2}{N (1 + \delta_i)} \cdot (1000N \ln (N^4 / \eta)}
        \le 2 e^{-\frac{1000}{12} \cdot \frac{\delta_i^2}{1 + \delta_i} \cdot \ln (N^4 / \eta)} \]

    Since $\delta_i > 1/4$, we can use $(1000/12) \delta_i^2 / (1+\delta_i) \ge 1$, and therefore,
    \[  \Pr\left[\hat\delta_i \notin \left(1 \pm \frac{1}{2}\right)\delta_i\right]
        \le 2 e^{-\ln (N^4 / \eta)}
        = 2\eta / N^4 \]

    For the expected value, observe that $\hat\delta_i \le N-1 < N$ with probability $1$. Therefore,
    \[  \E[\hat\delta_i^3]
        \le ((3/2)\delta_i)^3 + \Pr\left[\hat\delta_i > (3/2)\delta_i^3\right] \cdot (N)^3
        \le 4\delta_i^3 + (2\eta/N^4) \cdot N^3
        \le 4\delta_i^3 + 2\eta/N \]
\end{proof}

\repprovelemma{technical:estimate-sum-cubes--small-delta-chernoff}
\begin{proof}
    For every $i$ for which $\delta_i \le 1/4$, $p_i \le (1 + 1/4)/N = 5/(4N)$. By Chernoff's bound:
    \[  \Pr\left[\hat\delta_i \ge \frac{1}{2} \right]
        = \Pr\left[\hat p_i \ge \frac{3}{2N}\right]
        = \Pr\left[\Bin(q, p_i) \ge \frac{3}{2N}\right]
        \le \Pr\left[\Bin\left(q, \frac{5}{4N}\right) \ge \frac{3}{2N}\right]
        \le e^{-\frac{1}{12} \cdot q \cdot 5/(4N)} \]

    Since $q \ge 1000N \ln (N^4/\eta)$,
    \[  \Pr\left[\hat\delta_i \ge \frac{1}{2} \right]
        \le e^{-(5/48) \cdot 1000N \ln (N^4/\eta) }
        \le e^{-\ln (N^4 / \eta)}
        = \eta/N^4 \]

    For the expected value:
    \[  \E[\hat\delta_i^3 \cdot \mathbf{1}_{\hat\delta_i^3 \ge 1/2}]
        = \max_{\text{w.p. 1}} \hat\delta_i^3 \cdot \Pr\left[\hat\delta_i^3 \ge 1/2\right]
        \le \frac{\eta}{N^4} \cdot N^3
        \le 2\eta/N \]
\end{proof}

\section{Friendly distributions}
\label{sec:friendly}

We first prove the algebraic lower-bound of $t_\mu$ for friendly distributions.

\begin{lemma}[Technical lemma]{technical:estimate-t-friendly-advice--deltas-bound}
    Let $\mu$ be a friendly distribution over $\Omega = \{1,\ldots,N\}$. For every $i \in \Omega$, let $\delta_i = N\mu(i) - 1$. In this setting, 
    $\sum_{i=1}^N \delta_i^2 + \sum_{i=1}^N \delta_i^3 - \frac{1}{N}\left(\sum_{i=1}^N \delta_i^2\right)^2
        \ge \frac{1}{6}\sum_{i=1}^N \delta_i^2$.
\end{lemma}

\begin{lemma}[Technical lemma]{technical:estimate-t-friendly-advice--delta-cubes-bound}
    Let $\mu$ be a friendly distribution over $\Omega = \{1,\ldots,N\}$. For every $i \in \Omega$, let $\delta_i = N\mu(i) - 1$. In this setting, 
    $\sum_{i=1}^N \delta_i^2 + \sum_{i=1}^N \delta_i^3 - \frac{1}{N}\left(\sum_{i=1}^N \delta_i^2\right)^2 \ge \frac{1}{45}\sum_{i : \delta_i \ge 1} \delta_i^3$.
\end{lemma}

\repprovelemma{t-lbnd-by-sum-squares-sum-cubes}
\begin{proof}
    By Lemma \reflemma{technical:estimate-t-friendly-advice--deltas-bound},
    \[  \sum_{i=1}^N \delta_i^2 + \sum_{i=1}^N \delta_i^3 - \frac{1}{N}\left(\sum_{i=1}^N \delta_i^2\right)^2
        \ge \frac{1}{6}\sum_{i=1}^N \delta_i^2 \]
    
    By Lemma \reflemma{technical:estimate-t-friendly-advice--delta-cubes-bound},
    \[  \sum_{i=1}^N \delta_i^2 + \sum_{i=1}^N \delta_i^3 - \frac{1}{N}\left(\sum_{i=1}^N \delta_i^2\right)^2
        \ge \frac{1}{45}\sum_{i : \delta_i \ge 1} \delta_i^3 \]

    Combined,
    \[ \sum_{i=1}^N \delta_i^2 + \sum_{i=1}^N \delta_i^3 - \frac{1}{N}\left(\sum_{i=1}^N \delta_i^2\right)^2
        \ge \frac{1}{90}\left(\frac{1}{6}\sum_{i=1}^N \delta_i^2 + \sum_{i : \delta_i \ge 1} \delta_i^3\right) \]

    Therefore, by Lemma \reflemma{t-explicit-by-deltas}
    \begin{eqnarray*}
        t_\mu &=& \frac{\frac{1}{N}\left(\sum_{i=1}^N \delta_i^2 + \sum_{i=1}^N \delta_i^3 - \frac{1}{N}\left(\sum_{i=1}^N \delta_i^2\right)^2\right)}{\left(1 + \frac{1}{N}\sum_{i=1}^N \delta_i^2\right)^2} \\
        &\ge& \frac{1}{90} \cdot \frac{\frac{1}{N}\left(\sum_{i=1}^N \delta_i^2 + \sum_{i : \delta_i \ge 1}^N \delta_i^3\right)}{\left(1 + \frac{1}{N}\sum_{i=1}^N \delta_i^2\right)^2} \\
        \text{[Lemma \reflemma{mu22-explicit-by-deltas}]} &=& \frac{1}{90(N\|\mu\|_2^2)^2} \cdot \frac{1}{N}\left(\sum_{i=1}^N \delta_i^2 + \sum_{i : \delta_i \ge 1}^N \delta_i^3\right)
    \end{eqnarray*}
\end{proof}

To estimate $t_\mu$, we estimate the sum of squares $\sum_{i=1}^N \delta_i^2$ and the sum of large cubes $\sum_{i : \delta_i \ge 1} \delta_i^3$, and combine them according to Lemma \reflemma{t-lbnd-by-sum-squares-sum-cubes}.

\begin{proc-algo}{estimate-t-friendly}{\eta; \mu}
    \alginput{A distribution $\mu$ over $\Omega = \{1,\ldots,N\}$}
    \algoutput{If $\mu$ is friendly, then $X \ge t_\mu$ with probability $\ge 1-\eta$}
    \algmoments{$\E[X] = O(t_\mu + \eps)$}
    \algcomplexity{$O(N \log (N/\eta \eps) + \log (1/\eta) / \|\mu\|_2)$}
    \begin{code}
        \algitem Let $\ell_2 \gets \proc{estimate-L2-BC}(\eta/3;\mu,1/2)$.
        \algitem Let $a \gets \proc{estimate-sum-squares}(\eta/3; \mu)$.
        \algitem Let $b \gets \proc{estimate-sum-cubes}(\min\{\eta/3,\eps\}; \mu)$.
        \algitem Return $\min\{360(a + b) / (N\ell_2)^2, \sqrt{N}\}$.
    \end{code}
\end{proc-algo}

\repprovelemma{estimate-t-friendly}
\begin{proof}
    By Lemma \reflemma{estimate-L2-BC}:
    \begin{itemize}
        \item With probability at least $1-\eta/3$, $\ell_2 \le \frac{3}{2}\|\mu\|_2^2$.
        \item $\E[1/\ell_2^2] = O(1/\|\mu\|_2^4)$.
        \item The cost of obtaining $\ell_2$ is $O(\log (1/\eta) / \eps\|\mu\|_2)$.
    \end{itemize}
    
    By Lemma \reflemma{estimate-sum-squares}:
    \begin{itemize}
        \item With probability at least $1-\eta/3$, $a \ge \frac{1}{N} \sum_{i=1}^N \delta_i^2$.
        \item $\E[a] = O(\frac{1}{N} \sum_{i=1}^N \delta_i^2 + \eps)$.
        \item The cost of obtaining $a$ is $O(\log (1/\eta) \cdot (\sqrt{N}/\eps + 1/\eps\|\mu\|_2))$.
    \end{itemize}
    
    By Lemma \reflemma{estimate-sum-cubes}:
    \begin{itemize}
        \item With probability at least $1-\eta/3$, $b \ge \frac{1}{N} \sum_{i : \delta_i \ge 1} \delta_i^3$.
        \item $\E[b] = O(\frac{1}{N} \sum_{i : \delta_i \ge 1} \delta_i^3 + \frac{1}{N} \sum_{i=1}^N \delta_i^3 + \eps)$.
        \item The cost of obtaining $b$ is $O(N \log (N/\eta\eps))$.
    \end{itemize}

    By the union bound and Lemma \reflemma{t-lbnd-by-sum-squares-sum-cubes} (if $\mu$ is friendly), with probability at least $1-\eta$,
    \begin{eqnarray*}
        360(a + b) / (N\ell_2)^2
        \ge 360\frac{\frac{1}{N}\left(\sum_{i=1}^N \delta_i^2 + \sum_{i : \delta_i \ge 1} \delta_i^3\right)}{((3/2)N\|\mu\|_2^2)^2}
        \ge 160\frac{\frac{1}{N}\left(\sum_{i=1}^N \delta_i^2 + \sum_{i : \delta_i \ge 1} \delta_i^3\right)}{(N\|\mu\|_2^2)^2}
        \ge t_\mu
    \end{eqnarray*}
    If, for some reason, this expression is too high, then we use $\sqrt{N} \ge 1/\|\mu\|_2 \ge t_\mu$ instead.

    For the expected value,
    \begin{eqnarray*}
        \E[a+b]\E[1/\ell_2^2]/N^2
        = O\left(\frac{1}{N} \sum_{i=1}^N \delta_i^2 + \frac{1}{N} \sum_{i : \delta_i \ge 1} \delta_i^3 + \eps\right) \cdot O(1/\|\mu\|_2^4) / N^2
    \end{eqnarray*}
    Which is $O(t_\mu + \eps)$ by Lemma \reflemma{t-lbnd-by-sum-squares-sum-cubes}, assuming that $\mu$ is friendly. Otherwise, we cannot use the Lemma, and we can only guarantee that the output is bounded by $\sqrt{N}$.
\end{proof}

\subsection{Deferred proofs of technical lemmas}
The next section begins at Page \pageref{sec:large-mu2}.

\repprovelemma{technical:estimate-t-friendly-advice--deltas-bound}
\begin{proof}
    Since $\delta_i \ge -1/2$ for every $1 \le i \le N$, we can use the bounds $\delta_i^3 \ge -\frac{1}{2}\delta_i^2$ and $\delta_i^3 \ge \abs{\delta_i}^3 - \frac{1}{4}$.
    
    If $\sum_{i=1}^N \delta_i^2 \le \frac{1}{3}N$, then:
    \[  \sum_{i=1}^N \delta_i^2 + \sum_{i=1}^N \delta_i^3 - \frac{1}{N}\left(\sum_{i=1}^N \delta_i^2\right)^2
        \ge \sum_{i=1}^N \delta_i^2 - \frac{1}{2}\sum_{i=1}^N \delta_i^2 - \frac{1}{N} \cdot \frac{1}{3}N \cdot \sum_{i=1}^N \delta_i^2
        = \frac{1}{6} \sum_{i=1}^N \delta_i^2
    \]

    For higher sums, we need a few intermediate bounds. Since $\sum_{i=1}^N \delta_i = 0$, we can deduce that:
    \[  \sum_{i=1}^N \abs{\delta_i}
        = 2\sum_{i : \delta_i < 0} \abs{\delta_i} \le 2 \cdot N \cdot \frac{1}{2} = N \]

    By Cauchy-Schwartz inequality,
    \[  \frac{1}{N} \left(\sum_{i=1}^N \delta_i^2\right)^2
        = \frac{1}{N} \left(\sum_{i=1}^N \abs{\delta_i}^{1/2} \abs{\delta_i}^{3/2} \right)^2
        \le \frac{1}{N} \underbrace{\sum_{i=1}^N \abs{\delta_i}} \sum_{i=1}^N \abs{\delta_i}^3
        \le \frac{1}{N} \cdot N \cdot \sum_{i=1}^N \abs{\delta_i}^3
        = \sum_{i=1}^N \abs{\delta_i}^3 \]
    
    If $\sum_{i=1}^N \delta_i^2 \ge \frac{1}{3}N$, then:
    \begin{eqnarray*}
        \sum_{i=1}^N \delta_i^2 + \sum_{i=1}^N \delta_i^3 - \frac{1}{N}\left(\sum_{i=1}^N \delta_i^2\right)^2
        &\ge& \sum_{i=1}^N \delta_i^2 + \left(\underbrace{\sum_{i=1}^N \abs{\delta_i}^3} - \frac{1}{4}N\right) - \underbrace{\frac{1}{N}\left(\sum_{i=1}^N \delta_i^2\right)^2} \\
        &\ge& \sum_{i=1}^N \delta_i^2 - \frac{1}{4}N \\
        &\ge& \left(1 - \frac{1/4}{1/3}\right)\sum_{i=1}^N \delta_i^2
        = \frac{1}{4}\sum_{i=1}^N \delta_i^2
    \end{eqnarray*}
\end{proof}

\repprovelemma{technical:estimate-t-friendly-advice--delta-cubes-bound}
\begin{proof}
    Let:
    \[  \alpha = \frac{1}{N}\sum_{i=1}^N \delta_i^2, \qquad
        \beta = \frac{1}{N}\sum_{i : \delta_i \ge 1} \delta_i^3, \qquad
        \gamma = \frac{1}{N}\sum_{i : \delta_i < 1} \delta_i^3, \qquad
        \lambda = \frac{1}{N}\sum_{i : \delta_i < 0} \abs{\delta_i} = \frac{1}{N}\sum_{i : \delta_i > 0} \abs{\delta_i}
    \]

    Note that $\lambda \le \frac{6}{13}$, since $\delta_i \ge 7/13-1 \ge -6/13$ for every $i \in \Omega$.

    Case I. $\alpha \le 2$.

    Lower bound for $\gamma$:
    \[ \gamma
        = \frac{1}{N}\sum_{i : \delta_i < 1} \delta_i^3
        \ge \frac{1}{N} \cdot (-1/2) \sum_{i : \delta_i < 1} \delta_i^2
        = -\frac{1}{2}\alpha \]

    \begin{eqnarray*}
        \sum_{i=1}^N \delta_i^2 + \sum_{i=1}^N \delta_i^3 - \frac{1}{N}\left(\sum_{i=1}^N \delta_i^2\right)^2
        &=& \alpha + (\beta + \gamma) - \alpha^2 \\
        &\ge& \alpha + (\beta - \alpha/2) - \alpha^2 \\
        &=& (1/2 - \alpha)\alpha + \beta
        \ge (-3/2)\alpha + \beta
    \end{eqnarray*}

    Case I.a. If $\alpha \le 2$ and $\beta \ge 4\alpha$, then
    \begin{eqnarray*}
        \frac{1}{N}\left(\sum_{i=1}^N \delta_i^2 + \sum_{i=1}^N \delta_i^3 - \frac{1}{N}\left(\sum_{i=1}^N \delta_i^2\right)^2\right)
        &\ge& (-3/2)\alpha + \beta \\
        &=& (-3/2)\alpha + \frac{3}{8}\beta + \frac{5}{8}\beta
        \ge (-3/2)\alpha + (3/2)\alpha + \frac{5}{8}\beta
        = \frac{5}{8}\beta
    \end{eqnarray*}
    
    Case I.b. If $\alpha \le 2$ and $\beta < 4\alpha$, then we can use \reflemma{technical:estimate-t-friendly-advice--deltas-bound} to obtain that
    \[  \frac{1}{N}\left(\sum_{i=1}^N \delta_i^2 + \sum_{i=1}^N \delta_i^3 - \frac{1}{N}\left(\sum_{i=1}^N \delta_i^2\right)^2\right)
        \ge \frac{1}{6}\alpha
        > \frac{1}{24}\beta \]

    In both subcases of Case I, $\sum_{i=1}^N \delta_i^2 + \sum_{i=1}^N \delta_i^3 - \frac{1}{N}\left(\sum_{i=1}^N \delta_i^2\right)^2 \ge \frac{1}{24}\beta = \frac{1}{24}\sum_{i : \delta_i \ge 1} \delta_i^3$.
    
    Case II. $\alpha \ge 2$.

    Since $\sum_{i : \delta_i \ge 1} \delta_i^3 \ge \sum_{i : \delta_i \ge 1} \delta_i^2 = \alpha N \ge 2N$,
    \begin{eqnarray*}
        \sum_{i=1}^N \delta_i^3
        \ge \sum_{i : \delta_i \ge 1} \delta_i^3 - \frac{1}{8} N
        \ge \frac{15}{16}\sum_{i : \delta_i \ge 1} \delta_i^3
    \end{eqnarray*}

    Therefore,
    \begin{eqnarray*}
        \alpha^2 N
        &=& \frac{1}{N}\left(\sum_{i=1}^N \delta_i^2\right)^2 \\
        \text{[Cauchy-Schwartz]} &\le& \frac{1}{N} \cdot \sum_{i=1}^N \abs{\delta_i} \cdot \sum_{i=1}^N \abs{\delta_i}^3 \\
        &\le& \frac{1}{N} \cdot \left(2\sum_{i : \delta_i < 0} \abs{\delta_i}\right) \cdot \left(\sum_{i=1}^N \delta_i^3 + 2 \cdot \sum_{i : \delta_i < 0} \abs{\delta_i}^3\right) \\
        &\le& \frac{1}{N} \cdot \left(2\sum_{i : \delta_i < 0} \abs{\delta_i}\right) \cdot \left(\sum_{i=1}^N \delta_i^3 + \frac{1}{2} \cdot \sum_{i : \delta_i < 0} \abs{\delta_i}\right)
        = 2\lambda\left(\sum_{i=1}^N \delta_i^3 + \frac{1}{2}\lambda N\right)
    \end{eqnarray*}

    That is,
    \begin{eqnarray*}
        \sum_{i=1}^N \delta_i^3
        \ge \frac{\alpha^2 N}{2\lambda} - \frac{1}{2}\lambda N
        &=& \left(\frac{1}{2\lambda} - \frac{\lambda}{2\alpha^2} \right)\alpha^2 N \\
        \text{[Since $\alpha \ge 2$]} &\ge& \left(\frac{1}{2\lambda} - \frac{1}{8}\lambda \right)\alpha^2 N \\
        \text{[Since $\lambda \le 6/13$]} &\ge& \left(\frac{1}{2 \cdot (6/13)} - \frac{1}{8} \cdot (6/13)\right)\alpha^2 N
        > \frac{40}{39}\alpha^2 N
    \end{eqnarray*}

    Therefore, if $\alpha \ge 2$,
    \[  \sum_{i=1}^N \delta_i^2 + \sum_{i=1}^N \delta_i^3 - \frac{1}{N}\left(\sum_{i=1}^N \delta_i^2\right)^2
        \ge \alpha N + \sum_{i=1}^N \delta_i^3 - \alpha^2 N
        \ge 0 + \sum_{i=1}^N \delta_i^3 - \frac{39}{40}\sum_{i=1}^N \delta_i^3
        \ge \frac{1}{40} \sum_{i=1}^N \delta_i^3
        \ge \frac{1}{45} \sum_{i : \delta_i \ge 1} \delta_i^3 \]
\end{proof}

\section{Finding an advice when $\|\mu\|_2$ is large}
\label{sec:large-mu2}

Recall the key lemma relating $t_\mu$ to a large-deviation bound.

\repprovelemma{t-is-chebyshev}
\begin{proof}
    Assume that we draw a sample $i$ according to $\mu$, and let $X = \mu(i)$. Clearly, $\E[X] = \sum_{i \in \Omega} \mu(i) \cdot \mu(i) = \|\mu\|_2^2$. For the second moment, $\E[X^2] = \sum_{i \in \Omega} \mu(i) \cdot (\mu(i))^2 = \|\mu\|_3^3$. Therefore, $\Var[X] = \|\mu\|_3^3 - \|\mu\|_2^4$. By Chebyshev's inequality, for every $\alpha > 0$,
    \begin{eqnarray*}
        \Pr\left[\mu(i) \notin (1 \pm \alpha)\|\mu\|_2^2\right]
        = \Pr\left[X \notin (1 \pm \alpha)\E[X]\right]
        \le \frac{\Var[X]}{\alpha^2 (\E[X])^2}
        = \frac{\|\mu\|_3^3 - \|\mu\|_2^4}{\alpha^2 \|\mu\|_2^4}
    \end{eqnarray*}
    Therefore, $\frac{\|\mu\|_3^3}{\|\mu\|_2^4} - 1 = \frac{\|\mu\|_3^3 - \|\mu\|_2^4}{\|\mu\|_2^4} \ge \sup_{\alpha > 0} \alpha^2 \Pr\left[\mu(i) \notin (1 \pm \alpha)\|\mu\|_2^2\right]$.
\end{proof}

\subsection{Good partitions}
\label{sec:find-advice-large-mu2:subsec:reduction-facts}

In this technical subsection we prove the statements about good partitions. Let $\mu$ be a discrete distribution over a domain $\Omega$. Recall that a partition $A \cup B = \Omega$ is \emph{good} if:
\begin{itemize}
    \item $A \subseteq \{ i : \mu(i) > \frac{11}{20} \|\mu\|_2^2 \}$.
    \item $B \subseteq \{ i : \mu(i) < \frac{2}{3} \|\mu\|_2^2 \}$.
\end{itemize}

\repprovelemma{good-partition--mu-B-small}
\begin{proof}
    Recall Lemma \reflemma{t-is-chebyshev}:
    \[  t_\mu
        \ge \frac{1}{3^2} \cdot \Pr\left[\mu(i) \notin \left(1 \pm \frac{1}{3}\right)\|\mu\|_2^2 \right]
        \ge \frac{1}{9} \mu(B) \]

    Therefore, $\mu(B) \le 9t_\mu$.
\end{proof}

\repprovelemma{good-partition--A-is-friendly}
\begin{proof}
    By Lemma \reflemma{good-partition--mu-B-small}, $\mu(A) = 1 - \mu(B) \ge 1 - 9t_\mu \ge 99/100$. Observe that $\|\mu\|_2^2 \ge (1 - \mu(A))^2\|\mu_A\|_2^2 \ge \frac{99^2}{100^2}\|\mu_A\|_2^2$. For every $i \in A$,
    \begin{eqnarray*}
        \mu(i)
        \ge \frac{11}{20}\|\mu\|_2^2 
        \ge \frac{11}{20} \cdot \frac{99^2}{100^2} \|\mu_A\|_2^2 
        \ge \frac{7}{13}\|\mu_A\|_2^2
        \ge \frac{7}{13\abs{A}}
    \end{eqnarray*}
    
    Therefore, the distribution $\mu_A$ is friendly.
\end{proof}

\begin{lemma}{sub-m22-bounded-by-1/2}
    Let $\mu$ be a discrete distribution and let $i$ be an element. If $\mu(i) < \|\mu\|_2^2$, then $\mu(i) \le \frac{1}{2}$.
\end{lemma}
\begin{proof}
    Let $i$ be an element for which $\mu(i) > \frac{1}{2}$. The mass $\mu(i)$ must be maximal, since the sum of all other masses is strictly smaller than $1 - 1/2 = 1/2$. Recall that $\|\mu\|_2^2 = \E_{i \sim \mu}[\mu(i)]$. Therefore, $\mu(i) = \max_i \mu(i) \ge \E[\mu(i)] = \|\mu\|_2^2$, and as a result, if $\mu(i) < \|\mu\|_2^2$ then $\mu(i) \le \frac{1}{2}$.
\end{proof}

\begin{lemma}[Technical lemma]{technical:t-is-linear-by-erasing-small-element}
    Let $\mu$ be a discrete distribution, and let $i$ be an element for which $\mu(i) < \|\mu\|_2^2$. Let $\tau$ be the distribution obtained by conditioning $\mu$ on ``not $i$''. In this setting, if $\|\mu\|_3^3 \le 2\|\mu\|_2^4$, then $\frac{\|\tau\|_3^3}{\|\tau\|_2^4} = \frac{\|\mu\|_3^3}{\|\mu\|_2^4} \pm 4\mu(i)$. (In other words, if $t_\mu \le 1$ then $t_\tau = t_\mu \pm 4\mu(i)$).
\end{lemma}

\repprovelemma{good-partition--tmu-by-tmuA-muB}
\begin{proof}
    Let $B \subseteq \{ i : \mu(i) < \frac{2}{3}\mu(i) \}$ be a set of small elements. For convenience, without loss of generality, we assume that $B$ is infinite, and denote it by $B = \{b_1,b_2,\ldots\}$. If it is finite, then we can add infinitely many arbitrary zero-probability elements.

    Let $B = \{b_1,\ldots\}$. For every $i \ge 0$, let $B_i = \{b_1,\ldots,b_i\}$ (empty for $i=0$) and $\tau_i = \mu_{\neg B_i}$ (the same as $\mu$ for $i=0$).

    By Lemma \reflemma{good-partition--mu-B-small}, $\mu(B) \le 9t \le 1/10$. Therefore, for every $i \ge 1$:
    \begin{eqnarray*}
        \tau_{i-1}(b_i)
        = \frac{\mu(b_i)}{1 - \mu(B_{i-1})}
        \le \frac{1}{1-1/10}\mu(b_i)
        = \frac{10}{9}\mu(b_i)
    \end{eqnarray*}

    For every $i \ge 1$, $\|\mu\|_2^2 \le \frac{10}{9}\|\tau_{i-1}\|_2^2$ since:
    \begin{eqnarray*}
        \|\mu\|_2^2
        = (1 - \mu(B_i))^2 \|\mu_{\neg B_i}\|_2^2 + \sum_{j\in B_i} (\mu(j))^2
        &\le& \|\tau_i\|_2^2 + \mu(B_i) \max_{j\in B_i} \mu(i) \\
        &\le& \|\tau_i\|_2^2 + 9t \cdot \frac{2}{3}\|\mu\|_2^2
        \le \|\tau_i\|_2^2 + \frac{1}{10}\|\mu\|_2^2
    \end{eqnarray*}
    
    And therefore, for every $i \ge 1$,
    \begin{eqnarray*}
        \tau_{i-1}(b_i)
        \le \frac{10}{9}\mu(b_i)
        <   \frac{10}{9} \cdot \frac{2}{3} \|\mu\|_2^2
        \le \frac{10}{9} \cdot \frac{2}{3} \cdot \frac{10}{9} \|\tau_{i-1}\|_2^2
        \le \|\tau_{i-1}\|_2^2
    \end{eqnarray*}

    Therefore, we can apply Lemma \reflemma{technical:t-is-linear-by-erasing-small-element} for every $i \ge 1$, as long as $t_{\tau_{i-1}} \le 1$.
    
    By Definition, $t_{\tau_0} = t_\mu \le 1/90 = 1/90 + 5 \mu(B_0) \le 1/90 + 5 \mu(B) < 1$. For every $i \ge 1$, inductively by applying Lemma \reflemma{technical:t-is-linear-by-erasing-small-element}:
    \begin{eqnarray*}
        t_{\tau_i}
        = t_{\tau_{i-1}} \pm 4\tau_{i-1}(b_i)
        = t_{\tau_{i-1}} \pm 4 \cdot \frac{10}{9} \mu(b_i)
        = t_{\tau_{i-1}} \pm 5 \mu(b_i)
        = t_\mu \pm 5 \mu(B_i)
        \le t_\mu + 5 \mu(B) < 1
    \end{eqnarray*}

    By considering the limit, $t_{\mu_{\neg B}} = \lim_{i \to \infty} t_{\tau_{\neg B_i}} \in t_\mu \pm 5\mu(B)$.
\end{proof}

\subsection{The reduction}

We estimate $t_\mu$ with constant accuracy (using \proc{estimate-t-directly}) to test whether or not $t_\mu \le 1/900$, since one of our good-partition lemmas require this assumption. If $t_\mu$ is found to be larger, then we estimate it again (independently, for the ease of the analysis, using the same parameters) and return the result as an advice.

Next, we learn the input distribution $\mu$ using $\tilde{O}(1 / \|\mu\|_2^2)$ samples to construct a set $A$ such that $A \cup (\Omega\setminus A)$ is a good partition with high probability. More precisely, the algorithm tests every individual element for having mass greater than $\frac{11}{20}\|\mu\|_2^2$ or smaller than $\frac{2}{3}\|\mu\|_2^2$ (note the overlap). Note that since only the set $A$ is constructed, the algorithm can access $B = \Omega\setminus A$ only through the belonging oracle ($i \in^? B$), which is implemented as the negation of the corresponding oracle for $A$ ($i \in B \leftrightarrow i \notin A$).

Finally, we estimate $\mu(B)$ (through \proc{estimate-indicator-additive}) and $t_{\mu_A}$ (through \proc{estimate-t-friendly}, using rejection-sampling according to Lemma \reflemma{rejection-sampling-concentration}). The result advice is $\est(t_{\mu_A}) + 5(\est(\mu(B)) + \eps)$, since we have to cancel the potential error in $\mu(B)$ (at most $\eps$, additive) and then extract $t_\mu$ from the inequality promised by Lemma \reflemma{good-partition--tmu-by-tmuA-muB}. \algforprocshort{find-advice-large-mu2}

\begin{proc-algo}{find-advice-large-mu2}{\eta; \mu, \eps}
    \begin{code}
        \algitem Let $t_1 \gets \proc{estimate-t-directly}(\min\{\eta/10,\eps\};\mu,1/10^4)$.
        \begin{If}{$t_1 \ge 1/900$}
            \algitem Return $\proc{estimate-t-directly}(\eta/10;\mu,1/10^4)$.
        \end{If}
        \algitem Let $\ell \gets \proc{estimate-L2-BC}(\min\{\eta/10,\eps\}; \mu,1/100)$.
        \algitem Let $q \gets \ceil{10000 \ln (100/\eta \eps \ell) / \ell}$.
        \algitem Draw $q$ independent samples from $\mu$.
        \begin{For}{every $i \in \Omega$ (only finitely many are non-zero)}
            \algitem Let $X_i$ be the number of $i$ samples drawn.
            \algitem Let $\hat p_i \gets X_i / q$.
        \end{For}
        \algitem Let $A \gets \{ i : \hat p_i > (3/5)\ell \}$ and (virtually) $B \gets \{ i : \hat p_i \le (3/5)\ell \}$.
        \algitem Let $r_B \gets \proc{estimate-indicator-additive}(\eta/10; \bullet, \eps)$.
        \begin{Codeblock*}
            \algitem Sample: draw $i \sim \mu$, the result is an indicator for $i \notin A$.
        \end{Codeblock*}
        \algitem Let $N \gets \abs{A}$.
        \algitem Let $r_A \gets \proc{estimate-t-friendly}(\eta/10; \mu_A, \eps)$.
        \begin{Codeblock*}
            \item By rejection sampling.
            \begin{If}{The first $i$ $A$-samples require more than $4(i + \ceil{12 \ln(10/\eta)})$ $\mu$-samples}
                \algitem Terminate and return $0$.
            \end{If}
        \end{Codeblock*}
        \algitem Return $r_A + 5(r_B + \eps)$.
    \end{code}
\end{proc-algo}

\begin{lemma}{friendly-reduction-A-must-have}
    Considering \proc*{find-advice-large-mu2} (Algorithm \refalg{find-advice-large-mu2}), if $\ell$ is correctly in the range $(1 \pm 1/100)\|\mu\|_2^2$, then with probability at least $1-\eps \eta/10$, for every $i \in \Omega$ for which $\mu(i) \ge \frac{2}{3}\|\mu\|_2^2$, $i \in A$.
\end{lemma}
\begin{proof}
    \begin{eqnarray*}
        \ell &\le& (1 + 1/100)\|\mu\|_2^2 \\
        1/\ell &\ge& 1/(1 + 1/100)/\|\mu\|_2^2 \ge (1 - 1/100)/\|\mu\|_2^2 \\
        q &\ge& 10000(1-1/100) \ln(100(1-1/100)/\eta \eps \|\mu\|_2^2) / \|\mu\|_2^2
        \ge 9900 \ln(99/\eta \eps \|\mu\|_2^2) / \|\mu\|_2^2
    \end{eqnarray*}

    The threshold probability is $(3/5)\ell \le (3/5)(1+1/100)\|\mu\|_2^2$. For every $i$ for which $\mu(i) \ge \frac{2}{3}\|\mu\|_2^2$,
    \begin{eqnarray*}
        \Pr\left[\hat p_i \le (3/5)\ell \right]
        &\le& \Pr\left[\Bin(q, (2/3)\|\mu\|_2^2) \le (3/5)(1+1/100)\|\mu\|_2^2 \right] \\
        &\le& \Pr\left[\Bin(q, (2/3)\|\mu\|_2^2) \le 0.909\E[\Bin(q, (2/3)\|\mu\|_2^2)] \right] \\
        &\le& e^{-\frac{1}{3} \cdot (0.091)^2 \|\mu\|_2^2 q}
        \le e^{-\frac{1}{3} \cdot (0.091)^2 \cdot 9900 \ln(99/\eta \eps \|\mu\|_2^2)}
        \le e^{-\ln(99/\eta \eps \|\mu\|_2^2)}
        = \frac{1}{99} \eta \eps \|\mu\|_2^2
    \end{eqnarray*}
    
    There are at most $1 / ((2/3) \|\mu\|_2^2)$ such elements, and therefore, by the union bound, the probability that even one of them does not belong to $A$ is bounded by $\eps \eta/10$.
\end{proof}

\begin{lemma}{friendly-reduction-B-must-have}
    Considering \proc*{find-advice-large-mu2} (Algorithm \refalg{find-advice-large-mu2}), if $\ell$ is correctly in the range $(1 \pm 1/100)\|\mu\|_2^2$, then with probability at least $1-\eps\eta/10$, for every $i \in \Omega$ for which $\mu(i) < \frac{11}{20}\|\mu\|_2^2$, $i \in B$.
\end{lemma}

\begin{proof}
    \begin{eqnarray*}
        \ell &\ge& (1 - 1/100)\|\mu\|_2^2 \\
        1/\ell &\ge& 1/(1 + 1/100)/\|\mu\|_2^2 \ge (1 - 1/100)/\|\mu\|_2^2 \\
        q &\ge& 10000(1-1/100) \ln(100(1-1/100)/\eta\|\mu\|_2^2) / \|\mu\|_2^2
        \ge 9900 \ln(99/\eta\|\mu\|_2^2) / \|\mu\|_2^2
    \end{eqnarray*}

    We partition the set $\{ i : \mu(i) < (11/20) \|\mu\|_2^2 \}$ into $m \le \frac{1}{\frac{1}{2} \cdot \frac{11}{20}\|\mu\|_2^2} + 1 \le 5 / \|\mu\|_2^2$ sets $S_1,\ldots,S_m$, each have probability mass smaller than $(11/20) \|\mu\|_2^2$. For every $1 \le j \le m$, let $X_{S_j} = \sum_{i \in S_j} \mu(i)$ and $p_{S_j} = X_{S_j} / q$.

    The threshold probability is $(3/5)\ell \ge (3/5)(1-1/100)\|\mu\|_2^2$. For every $1 \le j \le m$,
    \begin{eqnarray*}
        \Pr\left[\hat p_{S_j} > (3/5)\ell \right]
        &\le& \Pr\left[\Bin(q, (11/20)\|\mu\|_2^2) > (3/5)(1-1/100)\|\mu\|_2^2 \right] \\
        &\le& \Pr\left[\Bin(q, (11/20)\|\mu\|_2^2) > 1.08\E[\Bin(q, (11/20)\|\mu\|_2^2)] \right] \\
        &\le& e^{-\frac{1}{3} \cdot (0.08)^2 \|\mu\|_2^2 q}
        \le e^{-\frac{1}{3} \cdot (0.08)^2 \cdot 9900 \ln(99/\eta\|\mu\|_2^2)}
        \le e^{-\ln(99/\eta\|\mu\|_2^2)}
        = \frac{1}{99} \eta \|\mu\|_2^2
    \end{eqnarray*}

    By the union bound over $m \le 5/\|\mu\|_2^2$ sets, the probability that $S_j \cap A \ne \emptyset$ is bounded by $\eta/10$.
\end{proof}

\begin{lemma}{find-advice-large-mu2--large-t}
    Consider the call to $\proc*{find-advice-large-mu2}(\eta; \mu, \eps)$, for $0 < \eta \le 1/3$ and $0 < \eps < 1$. If $t_\mu \ge 1/900$, then with probability at least $1-\eta$, the output is not smaller than $t_\mu$.
\end{lemma}
\begin{proof}
    With probability at least $1 - \eta/10$, $t_1 \ge t \ge 1/900$ (Lemma \reflemma{estimate-t-directly}), and then we take the if-true branch. This is also the probability that the second estimation of $t$ results in a value at least $t$. By the union bound, the probability of success is at least $1-2\cdot \eta/10 \ge 1-\eta$.
\end{proof}

\begin{lemma}{find-advice-large-mu2--small-t}
    Consider the call to $\proc*{find-advice-large-mu2}(\eta; \mu, \eps)$, for $0 < \eta \le 1/3$ and $0 < \eps < 1$. If $t_\mu < 1/900$, then with probability at least $1-\eta$, the output is not smaller than $t_\mu$.
\end{lemma}
\begin{proof}
    With some probability, the algorithm takes the first if-true branch. In this case, by Lemma \reflemma{estimate-t-directly}, the output is at least $t$ with probability at least $1-\eta/10 \ge 1-\eta$. The rest of the analysis refers to the case in which the algorithm does not take this branch.

    Consider the following list for the ranges of intermediate results and the probability to obtain them correctly. In every row, the ``resolved range'' is obtained by expanding the ``range'' expression assuming that all rows above are correct as well.
    
    \[\begin{array}{lllll}
         \textbf{Result} & \textbf{Range} & \textbf{Probability} & \textbf{Lemma} & \textbf{Resolved Range}  \\
         \ell & (1 \pm 1/100)\|\mu\|_2^2 & 1-\eta/10 & \text{L\reflemma{estimate-L2-BC}} & \ell \in (1 \pm 1/100)\|\mu\|_2^2 \\
         q & \ceil{10000 \ln (100/\eta \eps \ell) / \ell} & 1 & \text{explicit} &q \ge 9900 \ln (99/\eta \eps |\mu\|_2^2) / \|\mu\|_2^2 \\
         A & \supseteq \{ i : \mu(i) > (\approx (3/5))\ell \} & 1-\eps\eta/10 & \text{L\reflemma{friendly-reduction-A-must-have}} & \neg A \subseteq \{ i : \mu(i) < \frac{2}{3}\|\mu\|_2^2 \} \\
         \neg A & \supseteq \{ i : \mu(i) \le (\approx (3/5))\ell \} & 1-\eps\eta/10 & \text{L\reflemma{friendly-reduction-B-must-have}} & A \subseteq \{ i : \mu(i) \ge \frac{11}{20}\|\mu\|_2^2 \} \\
         r_B & \mu(B) \pm \eps & 1-\eta/10 & \text{L\reflemma{estimate-indicator-additive}} & r_B \in \mu(B) \pm \eps \\
         N & \abs{A} & 1 & \text{Markov} & N \le \frac{20}{11\|\mu\|_2^2} \\
         \text{crash} & 0 & 1-\eta/10 & \text{L\reflemma{rejection-sampling-concentration}} & \text{No crash due to rejection sampling} \\
         r_A & \ge t_{\mu_A} & 1-\eta/10 & \text{L\reflemma{estimate-t-friendly}} & r_A \ge t_{\mu_A} \\
         \text{output} & r_A + 5(r_B + \eps) & 1 & \text{explicit} & \text{See below}
    \end{array}\]

    With probability at least $1 - 6\eta/10 \ge 1-\eta$, the output is at least:
    \[ r_A + 5(\mu(B) + \eps)
        \ge t_{\mu_A} + 5((\mu(B) - \eps) + \eps)
        \ge (t_\mu - 5\mu(B)) + 5\mu(B)
        = t_\mu \]
    Where the last $\ge$-transition is correct by Lemma \reflemma{good-partition--tmu-by-tmuA-muB}.
\end{proof}

\begin{lemma}{find-advice-large-mu2--expval-mu-B}
    Consider the call to $\proc*{find-advice-large-mu2}(\eta; \mu, \eps)$, for $0 < \eta \le 1/3$ and $0 < \eps < 1$. The expected value of $\mu(\neg A)$ is bounded by $9t_\mu + \eps$.
\end{lemma}
\begin{proof}
    By Lemma \reflemma{friendly-reduction-A-must-have} and \reflemma{friendly-reduction-B-must-have}, with probability at least $1-2\eps\eta/10 \ge 1-\eps$, $A \subseteq \{ i : \mu(i) > \frac{11}{20}\|\mu\|_2^2 \}$ and $\neg A \subseteq \{ i : \mu(i) < \frac{2}{3}\|\mu\|_2^2 \}$. In this case, by Lemma \reflemma{good-partition--mu-B-small}, $\mu(\neg A) \le 9t_\mu$.

    By the law of total expectation,
    \[  \E[\mu(B)]
        \le \Pr[\text{good}] \times (\max \mu(B) \text{ if good}) + \Pr[\neg \text{good}] \cdot (\max \mu(B))
        \le 1 \cdot 9t_\mu + \eps \cdot 1
        = 9t_\mu + \eps \]
\end{proof}

\begin{lemma}{find-advice-large-mu2--expval-N}
    Consider the call to $\proc*{find-advice-large-mu2}(\eta; \mu, \eps)$, for $0 < \eta \le 1/3$ and $0 < \eps < 1$. The expected value of $N$ is bounded by $O(1/\|\mu\|_2)$ and $\E[N \log N] = O(\log(1/\|\mu\|_2)/\|\mu\|_2)$.
\end{lemma}
\begin{proof}
    Recall that $N$ is the size of the set $A = \{ \hat{p}_i : \hat{p}_i > (3/5)\ell \}$. Since $\sum_{i=1}^N \hat{p}_i = 1$ by definition and all of them are non-negative, $N = \abs{A} \le (5/3)/\ell$. Therefore, $\E[N] = O(1/\ell) = O(1/\|\mu\|_2^2)$ and $\E[N \log N] = O(\log (1/\|\mu\|_2) / \|\mu\|_2^2)$ (both by Lemma \reflemma{estimate-L2-BC}).
\end{proof}

\begin{lemma}{find-advice-large-mu2--expval-tmuA}
    Consider the call to $\proc*{find-advice-large-mu2}(\eta; \mu, \eps)$, for $0 < \eta \le 1/3$ and $0 < \eps < 1$. The expected value of $t_{\mu_A}$ is $O(t_\mu + \eps/\|\mu\|_2)$.
\end{lemma}
\begin{proof}
    By Lemma \reflemma{friendly-reduction-A-must-have} and \reflemma{friendly-reduction-B-must-have}, with probability at least $1-2\eps\eta/10 \ge 1-\eps$, $A \subseteq \{ i : \mu(i) > \frac{11}{20}\|\mu\|_2^2 \}$ and $\neg A \subseteq \{ i : \mu(i) < \frac{2}{3}\|\mu\|_2^2 \}$.
    
    If this happens, then by Lemma \reflemma{good-partition--mu-B-small}, $\mu(\neg A) \le 9t_\mu$, and by Lemma \reflemma{good-partition--A-is-friendly}, $\mu_A$ is friendly. In this case, $t_{\mu_A} \le t_\mu + 5\mu(\neg A) \le 46t_\mu$, unless $t_\mu \ge 1/900$, in which case $t_{\mu_A} \le \eps/\|\mu_A\|_2 \le \eps \sqrt{N}$.

    If this does not happen, then $t_{\mu_A} \le 1/\|\mu_A\|_2 \le \sqrt{N}$.

    Note that, by Jensen's inequality and Lemma \reflemma{find-advice-large-mu2--expval-N}, $\E[\sqrt{N}] \le \sqrt{\E[N]} = O(1/\|\mu\|_2)$. By the law of total expectation,
    \begin{eqnarray*}
        \E[t_{\mu_A}]
        &\le& 1 \cdot (46t_\mu + \eps \cdot \E[\sqrt{N}]) + \eps \cdot \E[\sqrt{N}] \\
        &\le& 1 \cdot (46t_\mu + \eps \cdot O(1/\|\mu\|_2)) + \eps \cdot O(1/\|\mu\|_2)
        = O(t_\mu + \eps/\|\mu\|_2)
    \end{eqnarray*}
\end{proof}

\repprovelemma{find-advice-large-mu2}
\begin{proof}
    For correctness, see Lemma \reflemma{find-advice-large-mu2--large-t} ($t_\mu \ge 1/900$) and Lemma \reflemma{find-advice-large-mu2--small-t} ($t_\mu < 1/900$).

    In the following we analyze both the expected sample complexity and the expected output.

    The cost of \proc{estimate-t-directly} with a fixed accuracy parameter (two calls) is $O(\log(1/\eta)/\|\mu\|_2^{4/3})$ (Lemma \reflemma{estimate-t-directly}).

    The cost of \proc{estimate-L2-BC} (for estimating $\ell$) with a fixed accuracy parameter is $O(\log (1/\eta) / \|\mu\|_2)$. Additionally, $\ell$ preserves all negative moments (Lemma \reflemma{estimate-L2-BC}).

    Determining the value of $q$ is sample-free. The expected value of $q$ is $\E[q] = O(\E[\ln(1/\eta \eps \ell) / \ell]) = O(\ln (1/\eta \eps \|\mu\|_2) / \|\mu\|_2^2)$, where the last transition is correct since $\ell$ preserves all negative moments.

    Constructing $A$ and (virtually) $B$ is sample-free. By Lemma \reflemma{find-advice-large-mu2--expval-mu-B}, $\E[\mu(B)] = \E[\mu(\neg A)] \le 9t_\mu + \eps = O(t_\mu + \eps)$.

    The cost of \proc{estimate-indicator-additive} (for estimating $r_B$) is (Lemma \reflemma{estimate-indicator-additive}):
    \[  O\left(\log \frac{1}{\eta} \cdot \left(\frac{1}{\eps} + \frac{\E[\mu(B)]}{\eps^2}\right)\right)
        = O\left(\log \frac{1}{\eta} \cdot \left(\frac{1}{\eps} + \frac{t_\mu + \eps}{\eps^2}\right)\right)
        = O\left(\log \frac{1}{\eta} \cdot \left(\frac{1}{\eps} + \frac{t_\mu}{\eps^2}\right)\right) \]

    Determining the value of $N$ is sample-free. The expected value of $N$ is $O(1/\|\mu\|_2^2)$, and additionally $\E[N \log N] = O(\log (1/\|\mu\|_2) / \|\mu\|_2^2)$ (Lemma \reflemma{find-advice-large-mu2--expval-N}).

    By Lemma \reflemma{find-advice-large-mu2--expval-tmuA}, the expected value of $t_{\mu_A}$ is $O(t_\mu + \eps/\|\mu\|_2)$.

    The cost of \proc{estimate-t-friendly} applied to $\mu_A$, according to the specified rejection-sampling mechanism, is an $O(\log (1/\eta))$ additive penalty, plus the baseline cost (Lemma \reflemma{estimate-t-friendly}):
    \begin{eqnarray*}
        O\!\left(\E[N \log (N/\eta\eps)] + \frac{\log(1/\eta) \E[\sqrt{N}]}{\eps} + \frac{\log(1/\eta)}{\eps \|\mu\|_2}\right)
        &\!\!\!=\!\!\!& O\!\left(\frac{\log (1/\eta\eps\|\mu\|_2)}{\|\mu\|_2^2} + \frac{\log(1/\eta)}{\eps \|\mu\|_2} + \frac{\log(1/\eta)}{\eps \|\mu\|_2}\right) \\
        &\!\!\!=\!\!\!& O\!\left(\frac{\log (1/\eta\eps\|\mu\|_2)}{\|\mu\|_2^2} + \frac{\log (1/\eta)}{\eps \|\mu\|_2}\right)
    \end{eqnarray*}

    With probability at least $1-\eps$, The partition $A \cup B$ is good. In this case, by Lemma \reflemma{estimate-t-friendly}, $\E[r_A | A] = O(t_{\mu_A} + \eps) = O(t_\mu + \mu(B) + \eps) = O(t_\mu + \eps)$. If the partition is not good, then by Lemma \reflemma{estimate-t-friendly}, $\E[r_A | A] \le \sqrt{N}$.
    
    Combined,
    \[  \E[r_A]
        \le 1 \cdot O(t_\mu + \eps) + \eps \cdot O(\E[\sqrt{N}])
        =  O(t_\mu + \eps/\|\mu\|_2) \]

    Therefore, the expected output is:
    \[ \E[r_A + 5(r_B + \eps)]
        = O(t_\mu + \eps/\|\mu\|_2) + O(t_\mu + \eps) + O(\eps)
        = O(t_\mu + \eps/\|\mu\|_2) \]
\end{proof}

\subsection{Deferred proofs of technical lemmas}
The next section begins at Page \pageref{sec:lbnd-tools}.

\repprovelemma{technical:t-is-linear-by-erasing-small-element}
\begin{proof}
    Recall that:
    \begin{eqnarray*}
        \|\mu\|_2^2 &=& (1 - \mu(i))^2 \|\tau\|_2^2 + (\mu(i))^2 \\
        \|\mu\|_3^3 &=& (1 - \mu(i))^3 \|\tau\|_2^2 + (\mu(i))^3
    \end{eqnarray*}

    For the upper bound:
    \begin{eqnarray*}
        \frac{\|\tau\|_3^3}{\|\tau\|_2^4}
        &=& \frac{(\|\mu\|_3^3 - (\mu(i))^3) / (1-\mu(i))^3}{((\|\mu\|_2^2 - (\mu(i))^2) / (1 - \mu(i))^2)^2} \\
        &=& (1 - \mu(i)) \frac{\|\mu\|_3^3 - (\mu(i))^3}{(\|\mu\|_2^2 - (\mu(i))^2)^2} \\
        &\le& (1 - \mu(i)) \frac{\|\mu\|_3^3}{(\|\mu\|_2^2 - (\mu(i))^2)^2} \\
        &=& (1 - \mu(i)) \frac{\|\mu\|_3^3}{\|\mu\|_2^4 (1 - (\mu(i))^2 / \|\mu\|_2^2)^2} \\
        \text{[Since $\mu(i) < \|\mu\|_2^2$]} &\le& (1 - \mu(i)) \frac{\|\mu\|_3^3}{\|\mu\|_2^4 (1 - \mu(i))^2} \\
        &=& \frac{1}{1 - \mu(i)} \cdot \frac{\|\mu\|_3^3}{\|\mu\|_2^4} \\
        \text{[Lemma \reflemma{sub-m22-bounded-by-1/2}, $\mu(i) \le 1/2$]} &\le& (1 + 2\mu(i)) \cdot \frac{\|\mu\|_3^3}{\|\mu\|_2^4} \\
        \text{[Since $\|\mu\|_3^3 \le 2\|\mu\|_2^4$]} &\le& \frac{\|\mu\|_3^3}{\|\mu\|_2^4} + 4\mu(i)
    \end{eqnarray*}

    For the lower bound:
    \begin{eqnarray*}
        \frac{\|\tau\|_3^3}{\|\tau\|_2^4}
        &=& \frac{(\|\mu\|_3^3 - (\mu(i))^3) / (1-\mu(i))^3}{((\|\mu\|_2^2 - (\mu(i))^2) / (1 - \mu(i))^2)^2} \\
        &=& (1 - \mu(i)) \frac{\|\mu\|_3^3 - (\mu(i))^3}{(\|\mu\|_2^2 - (\mu(i))^2)^2} \\
        &=& (1 - \mu(i)) \frac{\|\mu\|_3^3 (1 - (\mu(i))^3 / \|\mu\|_3^3)}{\|\mu\|_2^4 (1 - (\mu(i))^2 / \|\mu\|_2^2)^2} \\
        &\ge& (1 - \mu(i)) \frac{\|\mu\|_3^3 (1 - (\mu(i))^3 / \|\mu\|_3^3)}{\|\mu\|_2^4} \\
        \text{[Since $\mu(i) < \|\mu\|_2^2$]} &\ge& (1 - \mu(i)) \frac{\|\mu\|_3^3}{\|\mu\|_2^4} \cdot \left(1 - (\|\mu\|_2^2)^2 \mu(i) / \|\mu\|_3^3\right) \\
        \text{[Since $\|\mu\|_2^4/\|\mu\|_3^3 \le 1$]} &\ge& (1 - \mu(i))^2 \frac{\|\mu\|_3^3}{\|\mu\|_2^4}\\
        &\ge& (1 - 2\mu(i)) \frac{\|\mu\|_3^3}{\|\mu\|_2^4} \\
        \text{[Since $\|\mu\|_3^3 \le 2\|\mu\|_2^4$]} &\ge& \frac{\|\mu\|_3^3}{\|\mu\|_2^4} - 4\mu(i)
    \end{eqnarray*}
\end{proof}

\section{Tools for the lower bound}
\label{sec:lbnd-tools}

We Recall a few well-known statements:

\begin{lemma}[Pinsker's inequality]{pinsker}
    For two distributions $\mu$ and $\tau$, $\DKL{\tau}{\mu} \ge 2 (\dtv(\tau,\mu))^2$.
\end{lemma}

\begin{lemma}{dkl-bounded-by-chi-sqr}
    For every two distributions $\mu$ and $\tau$, $\DKL{\tau}{\mu} \le \frac{1}{\ln 2}\chi^2(\tau,\mu)$.
\end{lemma}

\begin{lemma}{mgf-poisson}
    For every $\lambda > 0$ and $\alpha > 0$, if $X$ distributes like $\Poi(\lambda)$, then $\E[\alpha^X] = e^{\lambda(\alpha - 1)}$.
\end{lemma}

\begin{lemma}{chi-sqr-plus-one}
    For two discrete distributions over $\Omega$, $1 + \chi^2(\tau, \mu) = \E_{i \sim \mu}\left[(\tau(i) / \mu(i))^2\right]$.
\end{lemma}
\begin{proof}
    \begin{eqnarray*}
        \chi^2(\tau, \mu)
        = \E_{i \sim \mu}\left[\frac{(\tau(i) - \mu(i))^2}{(\mu(i))^2}\right]
        &=& \sum_{i \in \Omega} \frac{1}{\mu(i)} \left((\tau(i))^2 - 2\tau(i)\mu(i) + (\mu(i))^2\right) \\
        &=& \sum_{i \in \Omega} \frac{(\tau(i))^2}{\mu(i)} - 2 \sum_{i \in \Omega} \tau(i) + \sum_{i \in \Omega} \mu(i) \\
        &=& \E_{i \sim \mu} \left[\frac{(\tau(i))^2}{(\mu(i))^2}\right] - 2 + 1
        = \E_{i \sim \mu} \left[\frac{(\tau(i))^2}{(\mu(i))^2}\right] - 1
    \end{eqnarray*}
\end{proof}

The following ad-hoc lemma allows avoiding the use of large-deviation inequalities when bounding the collision norm of a randomly constructed distribution.

\begin{lemma}{disjoint-lambdas-for-eps-mu2-construction}
    For every $\lambda > 0$, let $f_\lambda(x) = \lambda x + \lambda^2$ and $K_\lambda = { x : \abs{f_\lambda(x)} \le 5 }$. In this setting, the ranges $K_2$, $K_6$, $K_8$ and $K_{10}$ are pairwise disjoint.
\end{lemma}
\begin{proof}
    Explicitly:
    \[  K_2 = [-9/2, 1/2], \quad
        K_6 = [-41/6, -31/6], \quad
        K_8 = [-69/8, -59/8], \quad
        K_{10} = [-21/2,-19/2]   \]
\end{proof}

In the following we show a technical bound about Poisson distribution, which arises (with multiple variants) in distribution-testing lower-bound analysis.

\begin{lemma}{poisson-dual-eps-dkl}
    Let $\lambda_1 \ge \lambda_2 > 0$. For $0 < \eps < 1$ and $0 \le \delta \le \eps \lambda_2$, let $\lambda_1^{\pm} = \lambda_1 \pm \delta$ and $\lambda_2^{\pm} = \lambda_2 \pm \delta$. In this setting,
    \[\DKL{\frac{1}{2}\left(\Poi(\lambda_1^+) \times \Poi(\lambda_2^-)\right) +
    \frac{1}{2} \left(\Poi(\lambda_1^-) \times \Poi(\lambda_2^+)\right) }{\Poi(\lambda_1) \times \Poi(\lambda_2)}
    \le 2\eps^4 \lambda_2^2 \]
\end{lemma}
\begin{proof}
    We use the following short-hand notations:
    \begin{eqnarray*}
        \mu^0 &:& \Poi(\lambda_1) \times \Poi(\lambda_2), \\
        \mu^{\pm} &:& \Poi(\lambda_1 + \delta) \times \Poi(\lambda_2 - \delta), \\
        \mu^{\mp} &:& \Poi(\lambda_1 - \delta) \times \Poi(\lambda_2 + \delta), \\
        \mu^* &:& \frac{1}{2}\mu^{\pm} + \frac{1}{2}\mu^{\mp}
    \end{eqnarray*}
    Also, for two integers $k_1 \ge 0$, $k_2 \ge 0$ and a label $\sigma \in \{0, \pm, \mp, *\}$, we use $p^\sigma(k_1, k_2) = \mu^\sigma(k_1,k_2)$.
    
    For every $k_1, k_2 \ge 0$:
    \begin{eqnarray*}
        p_0(k_1,k_2) &=& \frac{\lambda_1^{k_1} \lambda_2^{k_2} e^{-(\lambda_1+\lambda_2)}}{k_1! k_2!} \\
        p^{\pm}(k_1,k_2) &=& \frac{(\lambda_1 + \delta)^{k_1} (\lambda_2 - \delta)^{k_2} e^{-(\lambda_1+\lambda_2)}}{k_1! k_2!}
        = p_0(k_1,k_2) \cdot \left(1 + \frac{\delta}{\lambda_1}\right)^{k_1} \left(1 - \frac{\delta}{\lambda_2}\right)^{k_2} \\
        p^{\mp}(k_1,k_2) &=& \frac{(\lambda_1 - \delta)^{k_1} (\lambda_2 + \delta)^{k_2} e^{-(\lambda_1+\lambda_2)}}{k_1! k_2!}
        =  p_0(k_1,k_2) \cdot \left(1 - \frac{\delta}{\lambda_1}\right)^{k_1} \left(1 + \frac{\delta}{\lambda_2}\right)^{k_2} \\
        p^*(k_1,k_2) &=& p_0(k_1,k_2) \cdot \frac{1}{2}\left(\left(1 + \frac{\delta}{\lambda_1}\right)^{k_1} \left(1 - \frac{\delta}{\lambda_2}\right)^{k_2} + \left(1 - \frac{\delta}{\lambda_1}\right)^{k_1} \left(1 + \frac{\delta}{\lambda_2}\right)^{k_2}\right)
    \end{eqnarray*}

    For convenience, let:
    \[  \psi^+(k_1,k_2) = \left(1 + \frac{\delta}{\lambda_1}\right)^{k_1} \left(1 - \frac{\delta}{\lambda_2}\right)^{k_2}, \qquad
        \psi^-(k_1,k_2) = \left(1 - \frac{\delta}{\lambda_1}\right)^{k_1} \left(1 + \frac{\delta}{\lambda_2}\right)^{k_2}
    \]
    So that $p^*(k_1,k_2) = p_0(k_1,k_2) \cdot \frac{1}{2}(\psi^+(k_1,k_2) + \psi^-(k_1,k_2))$.

    For $(X,Y)$ distributing like $\Poi(\lambda_1) \times \Poi(\lambda_2)$:
    \begin{eqnarray*}
        \chi^2(\mu^*, \mu^0)
        &=& \sum_{k_1=0}^\infty \sum_{k_2=0}^\infty \frac{(p^*(k_1, k_2) - p^0(k_1, k_2))^2}{p^0(k_1, k_2)} \\
        &=& \sum_{k_1=0}^\infty \sum_{k_2=0}^\infty p_0(k_1,k_2) \cdot \left(\frac{1}{2}\psi^+(k_1,k_2) + \frac{1}{2}\psi^-(k_1,k_2) - 1\right)^2 \\
        &=& \E\left[\left(\frac{1}{2}\psi^+(X,Y) + \frac{1}{2}\psi^-(X,Y) - 1\right)^2\right] \\
        &=& \frac{1}{4} \E\left[\left(\psi^+(X,Y) + \psi^-(X,Y)\right)^2\right] - \E\left[\psi^+(X,Y) + \psi^-(X,Y)\right] + 1 \\
        &=& \frac{1}{4} \E\left[\left(\psi^+(X,Y)\right)^2\right] + \frac{1}{2} \E\left[\psi^+(X,Y) \psi^-(X,Y)\right] + \frac{1}{4} \E\left[\left(\psi^-(X,Y)\right)^2\right] + \cdots\\&& - \E\left[\psi^+(X,Y)\right] - \E\left[\psi^-(X,Y)\right] + 1
    \end{eqnarray*}

    We use the moment-generating function (Lemma \reflemma{mgf-poisson}) to resolve each component individually.

    For $s \in \{+1, -1\}$:
    \begin{eqnarray*}
        \E\left[\psi^s(X,Y)\right]
        &=& \E\left[\left(1 + \frac{s \delta}{\lambda_1}\right)^X \left(1 - \frac{s \delta}{\lambda_2}\right)^Y\right] \\
        &=& \E\left[\left(1 + \frac{s \delta}{\lambda_1}\right)^X\right]\E\left[\left(1 - \frac{s \delta}{\lambda_2}\right)^Y\right]
        = e^{\lambda_1 \cdot s\delta/\lambda_1} e^{\lambda_2 \cdot (-s\delta/\lambda_2)}
        = 1
    \end{eqnarray*}
    
    For $s \in \{+1, -1\}$:
    \begin{eqnarray*}
        \E\left[\left(\psi^s(X,Y)\right)^2\right]
        &=& \E\left[\left(1 + \frac{s \delta}{\lambda_1}\right)^{2X} \left(1 - \frac{s \delta}{\lambda_2}\right)^{2Y}\right] \\
        &=& \E\left[\left(1 + \frac{s \delta}{\lambda_1}\right)^{2X}\right]\E\left[\left(1 - \frac{s \delta}{\lambda_2}\right)^{2Y}\right] \\
        &=& e^{\lambda_1 (2s\delta/\lambda_1 + \delta^2/\lambda_1^2)} e^{\lambda_2 (-2s\delta/\lambda_2 + \delta^2/\lambda_2^2)}
        = e^{\delta^2(1 /\lambda_1 + 1 / \lambda_2)}
    \end{eqnarray*}

    For the ``composite'' component:
    \begin{eqnarray*}
        \E\left[\psi^+(X,Y) \cdot \psi^-(X,Y)\right]
        &=& \E\left[\left(1 + \frac{\delta}{\lambda_1}\right)^{X}\left(1 - \frac{\delta}{\lambda_2}\right)^{Y} \cdot \left(1 - \frac{\delta}{\lambda_1}\right)^{X}\left(1 + \frac{\delta}{\lambda_2}\right)^{Y}\right] \\
        &=& \E\left[\left(1 - \frac{\delta^2}{\lambda_1^2}\right)^{X}\right] \E\left[\left(1 - \frac{\delta^2}{\lambda_2^2}\right)^{Y}\right] \\
        &=& e^{-\delta^2/\lambda_1} e^{-\delta^2/\lambda_2}
        = e^{-\delta^2(1/\lambda_1 + 1/\lambda_2)}
    \end{eqnarray*}

    Combining all components:
    \begin{eqnarray*}
        \chi^2(\mu^*, \mu^0)
        &=& \frac{1}{4} \cdot e^{\delta^2 (1/\lambda_1 + 1/\lambda_2)} + \frac{1}{2} e^{-\delta^2 (1/\lambda_1 + 1/\lambda_2)} + \frac{1}{4} \cdot e^{\delta^2 (1/\lambda_1 + 1/\lambda_2)} - 1 - 1 + 1 \\
        &=& \frac{1}{2} e^{\delta^2 (1/\lambda_1 + 1/\lambda_2)} + \frac{1}{2} e^{-\delta^2 (1/\lambda_1 + 1/\lambda_2)} - 1 \\
        &=& \cosh (\delta^2 (1/\lambda_1 + 1/\lambda_2)) - 1
        \le \frac{1}{2} \cdot (\delta^2 (1/\lambda_1 + 1/\lambda_2))^2
        \le 2 \delta^4 / \lambda_2^2
    \end{eqnarray*}
    Where the last transition follows from $\lambda_1 \ge \lambda_2 > 0$.

    By Lemma \reflemma{dkl-bounded-by-chi-sqr},
    \[  \DKL{\mu^*}{\mu^0}
        \le \frac{1}{\ln 2}\chi^2(\mu^*,\mu^0)
        \le \frac{2}{\ln 2}\delta^4 / \lambda_2^2
        \le 3(\eps \lambda_2)^4 / \lambda_2^2
        =   3\eps^4 \lambda_2^2
    \]
\end{proof}

\section{An $\Omega(1 / \eps \|\mu\|_2)$ lower-bound}
\label{sec:lbnd-eps-mu2}

We recall the lemmas stated in the technical-overview and prove them.

\repprovelemma{eps2-lower-bound-extreme-mu}
\begin{proof}
    Let $a \in \Omega$ be an element for which $\mu(i) \ge \frac{1}{8}\|\mu\|_2$ and let $\bot \notin \Omega$ be a new element. We define $\nu$ as the distribution over $\Omega \cup \{\bot\}$ for which $\nu(a) = \mu(a) - 10 \eps \|\mu\|_2$ and $\nu(i) = \mu(i)$ for every $i \in \Omega \setminus \{a\}$ and $\nu(\bot) = 10 \eps \|\mu\|_2$.

    The total-variation distance between $\mu$ and $\nu$ is $10 \eps \|\mu\|_2$, and therefore, there is a trivial lower bound of $\Omega(1 / \eps \|\mu\|_2)$ samples to distinguish between them. Moreover, if $\eps \le 1/200$, then:
    \begin{eqnarray*}
        \|\nu\|_2^2
        &=& \|\mu\|_2^2 + ((\mu(a) - 10 \eps \|\mu\|_2)^2 - (\mu(a))^2) + ((10 \eps \|\mu\|_2)^2 - 0) \\
        &=& \|\mu\|_2^2 - 20 \eps \|\mu\|_2 \mu(a) + \eps^2 \|\mu\|_2^2 + 100 \eps^2 \|\mu\|_2^2 \\
        &\le& \|\mu\|_2^2 - \frac{5}{2} \eps \|\mu\|_2^2 + 101\eps^2 \|\mu\|_2^2 \\
        &=& \|\mu\|_2^2 - \eps \|\mu\|_2^2(5/2 - 101\eps) \\
        \text{[Since $\|\mu\|_2 \le 1$ and $\eps \le 1/500$]} &<& (1 - (9/4)\eps)\|\mu\|_2^2 \\
        &<& \frac{1-\eps}{1+\eps}\|\mu\|_2^2
    \end{eqnarray*}
    In other words, any algorithm that computes the collision norm within $(1 \pm \eps)$-factor fails to distinguish between $\nu$ and $\mu$ (since $(1 + \eps)\|\nu\|_2^2 < (1 - \eps)\|\mu\|_2^2$) unless it draws $\Omega(1 / \eps \|\mu\|_2)$ samples.
\end{proof}

\repprovelemma{base-deviation-construction-for-eps-mu2}
\begin{proof}
    Let $\lambda \in \{2,6,8,10\}$ be chosen according to $\eps$ and $\mu$, as described later in the proof, and let $\hat\eps = \lambda \sqrt{\eps}$. Additionally, for every $j \ge 1$, let $\hat\eps_j = \hat\eps \mu(2j)$.

    Draw $s_1,s_2,\ldots \in \{+1, -1\}$ uniformly and independently, and let $\nu$ be defined as $\nu(i) = \mu(i) + (-1)^i s_j \hat\eps_j$, where $j = \ceil{i/2}$. Clearly, the entries of $\nu$ sum to $1$ and they are non-negative since $\hat\eps \le 1$.

    The collision norm of $\nu$ is:
    \begin{eqnarray*}
        \|\nu\|_2^2
        = \sum_{i=1}^\infty (\nu(i))^2
        &=& \sum_{j=1}^\infty \left((\mu(2j-1) - s_j \hat\eps_j)^2 + (\mu(2j) + s_j \hat\eps_j)^2\right) \\
        &=& \sum_{j=1}^\infty \left((\mu(2j-1))^2 + (\mu(2j))^2 + 2s_j \hat\eps_j(\mu(2j) - \mu(2j-1)) + 2\hat\eps_j^2\right) \\
        &=& \|\mu\|_2^2 + 2 \hat\eps \sum_{j=1}^\infty s_j \mu(2j) (\mu(2j) - \mu(2j-1)) + 2 \hat\eps^2 \sum_{j=1}^\infty (\mu(2j))^2 \\
        &=& \|\mu\|_2^2 + 2\lambda\sqrt{\eps} \sum_{j=1}^\infty s_j \mu(2j) (\mu(2j) - \mu(2j-1)) + 2\lambda^2 \eps \sum_{j=1}^\infty (\mu(2j))^2 \\
        &=& \|\mu\|_2^2 + 2\eps\left(\lambda \cdot \frac{1}{\sqrt{\eps}} \sum_{j=1}^\infty s_j \mu(2j) (\mu(2j) - \mu(2j-1)) + \lambda^2 \sum_{j=1}^\infty (\mu(2j))^2\right)
    \end{eqnarray*}

    Since $\frac{1}{\sqrt{2}}\mu(2j-1) \le \mu(2j) \le \mu(2j-1)$ for every $j \ge 1$, there exists some $1/3 \le \beta_\mu \le 1$ for which:
    \begin{eqnarray*}
        \|\nu\|_2^2
        &=& \|\mu\|_2^2 + 2\eps\left(\lambda \cdot \frac{1}{\sqrt{\eps}} \sum_{j=1}^\infty s_j \mu(2j) (\mu(2j) - \mu(2j-1)) + \lambda^2 \beta_\mu \|\mu\|_2^2\right) \\
        &=& \|\mu\|_2^2 + 2\eps \|\mu\|_2^2 \left(\lambda \cdot \frac{1}{\sqrt{\eps} \|\mu\|_2^2} \sum_{j=1}^\infty s_j \mu(2j) (\mu(2j) - \mu(2j-1)) + \lambda^2 \beta_\mu\right)
    \end{eqnarray*}

    Therefore, there exists some $1/3 \le \alpha_\nu \le 1$ for which:
    \[  \|\nu\|_2^2
        = \|\mu\|_2^2 + 2\alpha_\nu \eps\|\mu\|_2^2 \left(\lambda \cdot \frac{1}{\sqrt{\eps} \|\mu\|_2^2}\sum_{j=1}^\infty s_j \mu(2j) (\mu(2j) - \mu(2j-1)) + \lambda^2\right) \]

    Let $X = \frac{1}{\sqrt{\eps} \|\mu\|_2^2}\sum_{j=1}^\infty s_j \mu(2j) (\mu(2j) - \mu(2j-1))$. By Lemma \reflemma{disjoint-lambdas-for-eps-mu2-construction}, the four events for ``$\abs{\lambda X + \lambda^2} \le 5$'', when considering $\lambda$ in $\{2,6,8,10\}$, are disjoint. Therefore, there exists a choice of $\lambda$ in this set for which $\abs{\lambda X + \lambda^2} \ge 5$ with probability at least $3/4$. For this choice of $\lambda$:
    \[  \Pr\left[\|\nu\|_2^2 \notin \left(1 \pm \frac{5}{2}\eps\right)\|\mu\|_2^2\right]
        \ge \Pr\left[\|\nu\|_2^2 \notin (1 \pm 10\alpha_\nu \eps)\|\mu\|_2^2\right]
        \ge \frac{3}{4} \]

    For indistinguishability, observe that $\nu(i) \in (1 \pm \hat{\eps}_j)\mu(i)$ (for $j = \ceil{i/2}$) for every $i \ge 1$. Therefore, by independence of the $s_j$s and by Lemma \reflemma{poisson-dual-eps-dkl}, for an algorithm drawing $\Poi(q)$ samples from its input distribution,
    \begin{eqnarray*}
        \DKL{\nu^{\Poi(q)}}{\mu^{\Poi(q)}}
        &\le& 3 \sum_{j=1}^\infty \left(\frac{\hat\eps_j}{\mu(2j)}\right)^4 (\mu(2j))^2 q^2 \\
        &=& 3 \sum_{j=1}^\infty \left(\frac{\hat\eps \mu(2j)}{\mu(2j)}\right)^4 (\mu(2j))^2 q^2 \\
        &=& 3\hat\eps^4 q^2 \sum_{j=1}^\infty (\mu(2j))^2
        \le 3 \cdot \lambda^4 \eps^2 q^2 \cdot 2\|\mu\|_2^2
        \le 6 \cdot 10^4 \eps^2 \|\mu\|_2^2 q^2
    \end{eqnarray*}
    
    If $q \le \frac{1}{10^4\eps \|\mu\|_2}$, then:
    \[  \DKL{\nu^{\Poi(q)}}{\mu^{\Poi(q)}}
        \le 6\cdot 10^4 \cdot \eps^2 \|\mu\|_2^2 \cdot \frac{1}{10^8 \eps^2 \|\mu\|_2^2}
        \le \frac{1}{1600} \]
    And therefore, by Pinsker's inequality, the total-variation distance of the algorithm's answer when given $\mu$ or $\nu$ is bounded by $\sqrt{2/1600} < 1/12$.
\end{proof}

\repprovelemma{lbnd-eps-mu2}
\begin{proof}
    Without loss of generality, let the domain of $\mu$ be $\Omega = \{1,2,\ldots\}$ (with zero-probability elements allowed), and assume that $\mu(1) \ge \mu(2) \ge \cdots$. Also, we assume that $16\eps \le 1/8000$ (we can use $\eps' = \min\{\eps, 1/128000\}$).

    If $\mu(1) \ge \frac{1}{8}\|\mu\|_2^2$, then the lower bound is covered by Lemma \reflemma{eps2-lower-bound-extreme-mu}. We proceed assuming that $\mu(1) < \frac{1}{8}\|\mu\|_2$.

    Let $B_1$ be the set of odd indexes for which $\mu(i) > \sqrt{2}\mu(i+1)$. Observe that $B_1$ is exponentially decreasing (if $i_1 < i_2$ and both belong to $B_1$ then $\mu(i_2) \le \frac{1}{\sqrt{2}}\mu(i_1)$), and therefore, $\sum_{i\in B_1} (\mu(i))^2 \le \sum_{r=0}^\infty 2^{-i} \max_i (\mu(i))^2 = 2 \cdot \left(\frac{1}{8}\|\mu\|_2\right)^2 \le \frac{1}{32} \|\mu\|_2^2$. Also, $\mu(B) \le (\sum_{r=0}^\infty 2^{-i/2}) \max_{i\in B} \mu(i) \le (2 + \sqrt{2}) \cdot \frac{1}{8}\|\mu\|_2 \le \frac{3}{7}\|\mu\|_2$.

    Let $B_2 = \{ i : i - 1 \in B_1 \}$ be the set of even indexes paired with the indexes of $B_1$. By definition of $B_2$:
    \begin{eqnarray*}
        \mu(B_2) &\le& \sum_{i\in B_1} \frac{1}{\sqrt{2}}\mu(i) = \frac{1}{\sqrt{2}}\mu(B_1) \le \frac{(2+\sqrt{2})/8}{\sqrt{2}}\|\mu\|_2 \le \frac{4}{13}\|\mu\|_2 \\
        \sum_{i\in B_2} (\mu(i))^2 &\le& \sum_{i\in B_1} \left(\frac{1}{\sqrt{2}}\mu(i)\right)^2 \le \frac{1}{64}\|\mu\|_2^2
    \end{eqnarray*}

    Let $A = \Omega \setminus (B_1 \cup B_2)$, so that $\|\mu_A\|_2^2 \ge \sum_{i\in A} (\mu(i))^2 \ge \frac{15}{16}\|\mu\|_2^2$. By the bound on $\mu(B_1)$ and $\mu(B_2)$, $\mu(A) \ge 1 - \frac{67}{91}\|\mu\|_2 \ge \frac{1}{4}$.

    We can apply Lemma \reflemma{base-deviation-construction-for-eps-mu2} on the restricted distribution $\mu_A$ using $\eps' = 16\eps$, to obtain a distribution $\mathcal D_A$ over distributions over $A_1$. We define the distribution $\mathcal D$ over distributions over $\Omega$ by drawing $\nu_A$ from $\mathcal D_A$ and then let $\nu = \mu(A) \times \nu_A + \mu(B_1 \cup B_2) \times \mu_{B_1 \cup B_2}$. It holds that:
    \begin{itemize}
        \item When drawing $\nu$ from $\mathcal D$, with probability at least $3/4$, $\abs{\|\nu\|_2^2 - \|\mu\|_2^2} \ge (\mu(A))^2 \cdot \frac{9}{4} \cdot 16\eps \cdot \|\mu_A\|_2^2 \ge 2.1 \eps \|\mu\|_2^2 > \max\left\{ 1 - \frac{1-\eps}{1+\eps}, \frac{1+\eps}{1-\eps} - 1 \right\} \cdot \|\mu\|_2^2$.
        \item Any algorithm that distinguishes between $\mu$ and an input distribution $\nu$ drawn from $\mathcal D$, and therefore, must distinguish between $\mu_A$ and $\nu_A$, with total-variation distance greater than $1/12$, must draw $\Omega(1/(16\eps) \|\mu_A\|_2) = \Omega(1/\eps \|\mu\|_2)$ samples.
    \end{itemize}
    Therefore, $\mu$ is indistinguishable from distributions for which $(1 + \eps)\|\nu\|_2^2 < (1 - \eps)\|\mu\|_2^2$ or $(1 - \eps)\|\nu\|_2^2 > (1 + \eps)\|\mu\|_2^2$ when using fewer than $\Omega(1 / \eps \|\mu\|_2)$ samples.
\end{proof}

\newpage

\phantomsection

\section{An $\Omega(t_\mu / \eps^2)$ lower-bound}
We recall our notations: $\eps > 0$ is the threshold parameter, $\mu$ is a discrete distribution over a (possibly infinite) domain $\Omega$ and $t_\mu = \|\mu\|_3^3 / \|\mu\|_2^4 - 1$.

In this section we provide a few additional notations:
\begin{itemize}
    \item $a > 0$ is a construction parameter.
    \item $\eps_a(i) = \frac{a \eps}{t_\mu}\left(\frac{\mu(i)}{\|\mu\|_2^2} - 1\right)$.
    \item $\delta_a(i) = \mu(i) \cdot \eps_a(i)$.
\end{itemize}

We recall the construction described in the technical overview. For $\eps > 0$, $s \in \{+1,-1\}$ and $a > 0$, we define the distribution $\nu_{s,a}$ over the same domain $\Omega$ as $\nu_{s,a}(i) = \mu(i) + s \delta_a(i) = \mu(1 + s \eps_a(i))$.

This construction is valid if $t_\mu \ge a\eps/\|\mu\|_2$, as described in the following lemmas.

\begin{lemma}{lbnd-t-over-eps2-construction-non-negative}
    Let $\eps > 0$ and $a > 0$. Let $\mu$ be a discrete distribution over a domain $\Omega$ for which $t_\mu \ge a \eps/\|\mu\|_2$. For every $s \in \{+1, -1\}$ and $i \in \Omega$, $\nu_{s,a}(i) \ge 0$.
\end{lemma}
\begin{proof}
    It suffices to show that $\abs{\eps_a(i)} \le 1$ (so that $\nu_{s,a}(i) \ge (1 - \eps_a(i))\mu(i) \ge 0$) for every $i \in \Omega$.
    
    For the lower bound:
    \[  \eps_a(i)
        = \frac{a \eps}{t_\mu}\left(\frac{\mu(i)}{\|\mu\|_2^2} - 1\right)
        \ge -\frac{a \eps}{a\eps / \|\mu\|_2}
        = -\|\mu\|_2 \cdot \eps
        \ge -1
    \]
    
    For the upper bound, note that $\mu(i) \le \|\mu\|_2$ for every $i \in \Omega$, since $\|\mu\|_2^2 = \sum_{i \in \Omega} (\mu(i))^2$.
    \[  \eps_a(i)
        = \frac{a \eps}{t_\mu}\left(\frac{\mu(i)}{\|\mu\|_2^2} - 1\right)
        \le \frac{a \eps}{a \eps / \|\mu\|_2} \cdot \frac{\|\mu\|_2}{\|\mu\|_2^2}
        = \eps
        \le 1
    \]
\end{proof}

\begin{lemma}{lbnd-t-over-eps-construction-sums-to-1}
    Let $\eps > 0$ and $a > 0$. Let $\mu$ be a discrete distribution over a domain $\Omega$. For every $s \in \{+1, -1\}$, $\sum_{i \in \Omega} \nu_{s,a}(i) = 1$.
\end{lemma}
\begin{proof}
    Let $c_{s,a} = s \cdot a \cdot \eps$. Directly by definition:
    \begin{eqnarray*}
        \sum_{i \in \Omega} \nu_{s,a}(i)
        &=& \sum_{i \in \Omega} \mu(i) \left(1 + c_{s,a}\left(\frac{\mu(i)}{\|\mu\|_2^2} - 1\right)\right) \\
        &=& 1 + c_{s,a} \sum_{i \in \Omega} \left(\frac{(\mu(i))^2}{\|\mu\|_2^2} - \mu(i)\right)
        = 1 + c_{s,a} \left(\frac{1}{\|\mu\|_2^2} \sum_{i \in \Omega} (\mu(i))^2 - 1\right)
        = 1 + c_{s,a} \cdot 0
        = 1
    \end{eqnarray*}
\end{proof}

We recall Lemma \reflemma{three-of-four-nu-have-far-mu22} and prove it.

\repprovelemma{three-of-four-nu-have-far-mu22}
\begin{proof}
    Observe that:
    \begin{eqnarray*}
        \|\nu_{s,a}\|_2^2
        &=& \sum_{i \in \Omega} (\mu(i))^2 \cdot \left(1 + s \frac{a\eps}{t_\mu} \left(\frac{\mu(i)}{\|\mu\|_2^2} - 1\right)^2\right) \\
        &=& \sum_{i \in \Omega} (\mu(i))^2 + \frac{2sa\eps}{t_\mu} \sum_{i \in \Omega} (\mu(i))^2 \left(\frac{\mu(i)}{\|\mu\|_2^2} - 1\right) + \frac{a^2 \eps^2}{t_\mu^2} \sum_{i \in \Omega} (\mu(i))^2 \left(\frac{\mu(i)}{\|\mu\|_2^2} - 1\right)^2 \\
        &=& \|\mu\|_2^2 + \frac{2sa\eps}{t_\mu} \left(\frac{\|\mu\|_3^3}{\|\mu\|_2^2} - \|\mu\|_2^2\right) + \frac{a^2 \eps^2}{t_\mu^2} \sum_{i \in \Omega} (\mu(i))^2 \left(\frac{\mu(i)}{\|\mu\|_2^2} - 1\right)^2 \\
        &=& \|\mu\|_2^2 + 2sa\eps\|\mu\|_2^2 + a^2 \eps\|\mu\|_2^2 \cdot \underbrace{\frac{\eps}{t_\mu^2 \|\mu\|_2^2} \sum_{i \in \Omega} (\mu(i))^2 \left(\frac{\mu(i)}{\|\mu\|_2^2} - 1\right)^2}
    \end{eqnarray*}

    We let $K = \frac{\eps}{t_\mu^2 \|\mu\|_2^2} \sum_{i \in \Omega} (\mu(i))^2 \left(\frac{\mu(i)}{\|\mu\|_2^2} - 1\right)^2$, and rephrase the result as $\|\mu\|_2^2(1 + \eps (2sa + a^2 K))$.
    
    Note that $K$, which depends on $\eps$ and $\mu$, is non-negative and independent of $a$. For $s=+1$ and $a \in \{3,8\}$, $2sa + a^2 K \ge 2a \ge 3$, and therefore, $\|\nu_{+1,a}\|_2^2 \ge (1 + 3\eps)\|\mu\|_2^2$.
    
    For $s=-1$:
    \begin{itemize}
        \item If $K_{\mu,\eps} \le 1/3$, then for $a=3$, $-2a + a^2 K_{\mu,\eps} = -6 + 9 K_{\mu,\eps} \le -3$.
        \item If $K_{\mu,\eps} \ge 1/3$, then for $a=8$, $-2a + a^2 K_{\mu,\eps} = -8 + 64K_{\mu,\eps} \ge 3$.
    \end{itemize}
    That is, at least one $\nu \in \{\nu_{-1,3}, \nu_{-1,8}\}$ has $\abs{\|\nu\|_2^2 - \|\mu\|_2^2} \ge 3\eps\|\mu\|_2^2$.
\end{proof}

Before we show that it is hard to distinguish between $\mu$ and the constructed distributions, we state a few technical lemmas, whose proofs are deferred to the last subsection.

\begin{lemma}{lbnd-over-eps2-expval-eps-a}
    $\E_{i \sim \mu}[\eps_a(i)] = 0$.
\end{lemma}

\begin{lemma}{lbnd-over-eps2-expval-eps-a-squared}
    $\E_{i \sim \mu}[(\eps_a(i))^2] = a^2 \eps^2 / t_\mu$.
\end{lemma}

\begin{lemma}{lbnd-t-over-eps2-expval-same-sign}
    For $s \in \{+1, -1\}$, $\E_{i \sim \mu}[(1 + s \eps_a(i))^2] \le e^{a^2 \eps^2 / t_\mu}$.
\end{lemma}

\begin{lemma}{lbnd-t-over-eps2-expval-pm}
    $\E_{i \sim \mu}[1 - (\eps_a(i))^2] \le e^{-a^2 \eps^2/t_\mu}$.
\end{lemma}

At this point we restate Lemma \reflemma{lbnd-t-over-eps2-hardness-pair}, about distinguishing between $\mu$ and two constructed distributions, $\nu_{\pm 1, a}$, and prove it.

\repprovelemma{lbnd-t-over-eps2-hardness-pair}
\begin{proof}
    We use Lemma \reflemma{chi-sqr-plus-one} to bound the $\chi^2$ divergence. Our random variable is $\vec{x} = (x_1,\ldots,x_q)$.
    
    \begin{eqnarray*}
        1 + \chi^2\left(\frac{1}{2}\nu_{+1,a}^q + \frac{1}{2}\nu_{-1,a}^q, \mu^q\right)
        &=& \E_{\vec{x} \sim \mu^q} \left[\left(\frac{\frac{1}{2}\nu_{+1,a}^q(\vec{x}) + \frac{1}{2}\nu_{-1,a}^q(\vec{x})}{\mu^q(\vec{x})}\right)^2\right] \\
        &=& \frac{1}{4} \E_{\vec{x} \sim \mu^q} \left[\left(\frac{\nu_{+1,a}^q(\vec{x}) + \nu_{-1,a}^q(\vec{x})}{\mu^q(\vec{x})}\right)^2\right]    
    \end{eqnarray*}

    For the inner expression, we expand the square to obtain:
    \[ \frac{\prod_{\ell=1}^q (\nu_{+1,a} (x_\ell))^2 + 2 \prod_{\ell=1}^q \nu_{+1,a} (x_\ell) \nu_{-1,a}(x_\ell) + \prod_{\ell=1}^q (\nu_{-1,a} (x_\ell))^2}{\prod_{\ell=1}^q (\mu(x_\ell))^2} \]
    
    We use the definition of $\nu_{s,a}(i) = \mu(i) \left(1 + s \eps_{a}(i)\right)$ to obtain:
    \[ \frac{\prod_{\ell=1}^q (\mu(x_\ell))^2 (1 + \eps_a(x_\ell))^2 + 2 \prod_{\ell=1}^q (\mu(x_\ell))^2 (1 + \eps_a(x_\ell)) (1 - \eps_a(x_\ell)) + \prod_{\ell=1}^q (\mu(x_\ell))^2 (1 - \eps_a(x_\ell))^2}{\prod_{\ell=1}^q (\mu(x_\ell))^2} \]

    The $\mu$-factors cancel:
    \[  1 + \chi^2\left(\frac{1}{2}\nu_{+1,a}^q + \frac{1}{2}\nu_{-1,a}^q, \mu^q\right)
        = \frac{1}{4} \E_{\vec{x} \sim \mu^q} \left[\prod_{\ell=1}^q (1 + \eps_a(x_\ell))^2 + 2 \prod_{\ell=1}^q (1 - (\eps_a(x_\ell))^2) + \prod_{\ell=1}^q (1 - \eps_a(x_\ell))^2\right]
    \]

    By linearity of expectation, we sum the expected value of each product separately. The $x_\ell$s are drawn from $\mu^q$ and therefore independent and identically distributed. Therefore, $1 + \chi^2\left(\frac{1}{2}\nu_{+1,a}^q + \frac{1}{2}\nu_{-1,a}^q, \mu^q\right)$ equals to:
    \[  \frac{1}{4} \left(\E_{i \sim \mu} \left[(1 + \eps_a(i))^2\right]\right)^q + \frac{1}{2} \left(\E_{i \sim \mu} \left[(1 - (\eps_a(i))^2)\right]\right)^q + \frac{1}{4} \left(\E_{i \sim \mu}\left[(1 - \eps_a(i))^2\right]\right)^q
    \]

    We use Lemma \reflemma{lbnd-t-over-eps2-expval-same-sign} and Lemma \reflemma{lbnd-t-over-eps2-expval-pm} to obtain that:
    \begin{eqnarray*}
        1 + \chi^2\left(\frac{1}{2}\nu_{+1,a}^q + \frac{1}{2}\nu_{-1,a}^q, \mu^q\right)
        &\le& \frac{1}{2} e^{(a^2 \eps^2 / t_\mu) q} + \frac{1}{2} e^{- (a^2 \eps^2 / t_\mu) q} \\
        &=& \cosh ((a^2 \eps^2 / t_\mu)q)
        \le \cosh (8^2 / 500)
        \le 1 + 1/121
    \end{eqnarray*}
    
    We cancel the $1$ term to obtain that $\chi^2\left(\frac{1}{2}\nu_{+1,a}^q + \frac{1}{2}\nu_{-1,a}^q, \mu^q\right) \le 1/121$.
    
    To conclude, we Pinsker's inequality (Lemma \reflemma{pinsker}) and the KL-$\chi^2$ bound (Lemma \reflemma{dkl-bounded-by-chi-sqr}):
    \[  \dtv
        \le \sqrt{\dkl / 2}
        \le \sqrt{\chi^2 / (2 \ln 2)}
        \le \sqrt{1 / (121 \cdot 2 \ln 2)}
        \le \frac{1}{12}
    \]
\end{proof}

We conclude by restating the lower bound (Lemma \reflemma{lbnd-t-eps2}) and proving it.

\repprovelemma{lbnd-t-eps2}
\begin{proof}
    If $t_\mu < 8\eps/\|\mu\|_2$, then we use Lemma \reflemma{lbnd-eps-mu2} to derive the lower bound $\Omega(1/\eps\|\mu\|_2) = \Omega(t_\mu/\eps^2)$. In the following, we assume that $t_\mu \ge 8\eps/\|\mu\|_2$ to show a lower bound of $t_\mu / 500\eps^2 = \Omega(t_\mu / \eps^2)$ samples.
    
    Let $\nu$ be a distribution uniformly chosen from $\{\nu_{+1,3}, \nu_{-1,3}, \nu_{+1,8}, \nu_{-1,8}\}$. With probability at least $3/4$, $\abs{\|\nu\|_2^2 - \|\mu\|_2^2} \ge 3\eps\|\mu\|_2^2$ (Lemma \reflemma{three-of-four-nu-have-far-mu22}). In other words, either $(1 + \eps)\|\nu\|_2^2 < (1 - \eps)\|\mu\|_2^2$ or $(1 - \eps)\|\nu\|_2^2 > (1 + \eps)\|\mu\|_2^2$.

    The total-variation distance between $q$ samples drawn from $\mu$ and $q$ samples drawn from the chosen $\nu$ is:
    \[  \dtv\left(\mu^q, \frac{1}{4}\nu_{+1,3}^q + \frac{1}{4}\nu_{-1,3}^q + \frac{1}{4}\nu_{+1,8}^q + \frac{1}{4}\nu_{-1,8}^q\right)
    \]

    By the triangle inequality, this is bounded by:
    \[ \frac{1}{2} \dtv\left(\mu^q, \frac{1}{2}\nu_{+1,3}^q + \frac{1}{2}\nu_{-1,3}^q\right) + \frac{1}{2} \dtv\left(\mu^q, \frac{1}{2}\nu_{+1,8}^q + \frac{1}{2}\nu_{-1,8}^q\right)
    \]

    If $q \le t_\mu/500\eps^2$, then By Lemma \reflemma{lbnd-t-over-eps2-hardness-pair}, this is bounded by $\frac{1}{2} \cdot \frac{1}{12} + \frac{1}{2} \cdot \frac{1}{12} = \frac{1}{12}$. Therefore, we cannot distinguish with probability $3/4 - 1/12 = 2/3$ between the given $\mu$ and distributions for which the ranges $(1 \pm \eps)\|\nu\|_2^2$ and $(1 \pm \eps)\|\mu\|_2^2$ are disjoint.
\end{proof}

\subsection{Deferred proofs of technical lemmas}
The next section begins at Page \pageref{sec:one-after-t-eps2-lbnd}.

\repprovelemma{lbnd-over-eps2-expval-eps-a}
\begin{proof}
    \[  \E_{i \sim \mu}[\eps_a(i)]
        = \frac{a\eps}{t_\mu} \sum_{i \in \Omega} \mu(i) \left(\frac{\mu(i)}{\|\mu\|_2^2} - 1\right)
        = \frac{a\eps}{t_\mu} \frac{1}{\|\mu\|_2^2} \sum_{i \in \Omega} (\mu(i))^2 - \sum_{i \in \Omega} \mu(i)
        = \frac{a\eps}{t_\mu} \cdot (1 - 1)
        = 0
    \]
\end{proof}

\repprovelemma{lbnd-over-eps2-expval-eps-a-squared}
\begin{proof}
    \begin{eqnarray*}
        \E_{i \sim \mu}[(\eps_a(i))^2]
        &=& \frac{a^2 \eps^2}{t_\mu^2} \sum_{i \in \Omega} \mu(i) \left(\frac{\mu(i)}{\|\mu\|_2^2} - 1\right)^2 \\
        &=& \frac{a^2 \eps^2}{t_\mu^2} \sum_{i \in \Omega} \frac{(\mu(i))^3}{\|\mu\|_2^4} - 2\sum_{i \in \Omega} \frac{(\mu(i))^2}{\|\mu\|_2^2} + \sum_{i \in \Omega} \mu(i)
        = \frac{a^2 \eps^2}{t_\mu^2} \frac{\|\mu\|_3^3}{\|\mu\|_2^4} - 2 + 1
        = \frac{a^2 \eps^2}{t_\mu}
    \end{eqnarray*}
\end{proof}

\repprovelemma{lbnd-t-over-eps2-expval-same-sign}
\begin{proof}
    By Lemma \reflemma{lbnd-over-eps2-expval-eps-a} and \reflemma{lbnd-over-eps2-expval-eps-a-squared},
    \[  \E_{i \sim \mu}[(1 + s \eps_a(i))^2]
        = 1 + 2s \E_{i \sim \mu}[\eps_a(i)] + \E_{i \sim \mu}[(\eps_a(i))^2]
        = 1 + 2s \cdot 0 + \frac{a^2 \eps^2}{t_\mu}
        \le e^{a^2 \eps^2 / t_\mu}
    \]
\end{proof}

\repprovelemma{lbnd-t-over-eps2-expval-pm}
\begin{proof}
    By Lemma \reflemma{lbnd-over-eps2-expval-eps-a-squared},
    \begin{eqnarray*}
        \E_{i \sim \mu}[1 - (\eps_a(i))^2]
        = 1 - \E_{i \sim \mu}[(\eps_a(i))^2]
        = 1 - \frac{a^2 \eps^2}{t_\mu}
        \le e^{-a^2 \eps^2/t_\mu}
    \end{eqnarray*}
\end{proof}

\label{sec:one-after-t-eps2-lbnd}

\section{Acknowledgements}
I would like to thank my PhD advisors, Eldar Fischer (Technion) and Amit Levi (University of Haifa), for giving me both the tools and the opportunity for doing this advise-free, independent research project.

\newpage
\appendix
\section{Standard notations and large-deviation bounds}
\label{apx:std}

Formal notations regarding distributions:
\begin{itemize}
    \item The notations $\mu(i)$ and $\Pr_\mu[i]$ are considered identical.
    \item For an element $i\notin \Omega$, we consider $\mu(i) = 0$.
    \item For a set $A$, the notations $\mu(A)$, $\Pr_\mu[A]$ and $\sum_{i\in A} \mu(i)$ are considered identical.
    \item For a set $A$ for which $\mu(A) > 0$, the \emph{conditional distribution} $\mu_A$ is defined as $\mu_A : A \to [0,1]$ using $\mu_A(i) = \mu(i) / \mu(A)$.
    \item The notations $\mu_A(i)$, $\Pr_{\mu_A}[i]$ and $\Pr_\mu[i | A]$ are considered identical.
    \item For a set $B$, the notations $\mu_A(B)$, $\Pr_{\mu_A}[B]$ and $\Pr_\mu[B | A]$ are considered identical.
\end{itemize}

For a \emph{random variable} $X$, represented as a function $X : \Omega \to \mathbb R$:
\begin{itemize}
    \item The \emph{expected value} of $X$ (according to $\mu$) is defined as $\E_\mu[X] = \sum_{i \in \Omega} \mu(i) \cdot X(i)$.
    \item The \emph{variance} of $X$ (according to $\mu$) is defined as $\Var_\mu[X] = \E_\mu[X^2] - (\E_\mu[X])^2$.
\end{itemize}

Large-deviation bounds:
\begin{itemize}
    \item Markov (first-moment): if $\Pr[X \ge 0] = 1$ and $a > 0$, then $\Pr[X \ge a\E[X]] \le 1/a$.
    \item Chebyshev (second-moment): If $a > 0$, then $\Pr[\abs{X - \E[X]} \ge a] \le \Var[X] / a$.
    \item Chernoff (additive): if $a > 0$, then $\Pr[\Bin(n,p) \le np - a], \Pr[\Bin(n,p) \ge np + a] \le e^{-2a^2/n}$.
    \item Chernoff (multiplicative): If $\delta > 0$, then $\Pr[\Bin(n,p) \ge (1+\delta)np] \le e^{-\frac{\delta^2}{2+\delta} np}$.
    \item Chernoff (multiplicative): If $0 < \delta < 1$, then $\Pr[\Bin(n,p) \le (1-\delta)np] \le e^{-\frac{1}{2}\delta^2 np}$.
\end{itemize}
Note that the multiplicative Chernoff bounds are also applicable to Poisson distribution (through $\lambda = np$), and that the coefficient $\frac{\delta^2}{2+\delta}$ is usually relaxed to $\frac{1}{3}\delta$ for $\delta \ge 1$ and to $\frac{1}{3}\delta^2$ for $0 < \delta \le 1$.

\newpage
\addcontentsline{toc}{section}{References}
% \label{sec:one-after-t-eps2-lbnd}

\bibliographystyle{alpha}
\bibliography{main}

\newcommand{\etalchar}[1]{$^{#1}$}
\begin{thebibliography}{DGPP16}

\bibitem[ADK15]{ADK15}
Jayadev Acharya, Constantinos Daskalakis, and Gautam Kamath.
\newblock Optimal testing for properties of distributions.
\newblock {\em Advances in Neural Information Processing Systems}, 28, 2015.

\bibitem[BC17]{BC17}
Tugkan Batu and Cl{\'{e}}ment~L. Canonne.
\newblock Generalized uniformity testing.
\newblock {\em arxiv}, abs/1708.04696, 2017.

\bibitem[BFF{\etalchar{+}}01]{batuFFKRW2001}
Tugkan Batu, Eldar Fischer, Lance Fortnow, Ravi Kumar, Ronitt Rubinfeld, and Patrick White.
\newblock Testing random variables for independence and identity.
\newblock In {\em Proceedings 42nd IEEE Symposium on Foundations of Computer Science}, pages 442--451. IEEE, 2001.

\bibitem[DGPP16]{DGPP16}
Ilias Diakonikolas, Themis Gouleakis, John Peebles, and Eric Price.
\newblock Collision-based testers are optimal for uniformity and closeness.
\newblock {\em arXiv preprint arXiv:1611.03579}, 2016.

\bibitem[DK16]{DK16}
Ilias Diakonikolas and Daniel~M Kane.
\newblock A new approach for testing properties of discrete distributions.
\newblock In {\em 2016 IEEE 57th Annual Symposium on Foundations of Computer Science (FOCS)}, pages 685--694. IEEE, 2016.

\bibitem[DKN15]{DKN15}
Ilias Diakonikolas, Daniel~M Kane, and Vladimir Nikishkin.
\newblock Testing identity of structured distributions.
\newblock In {\em Proceedings of the Annual ACM-SIAM Symposium on Discrete Algorithms}, pages 1841--1854, 2015.

\bibitem[Gol20]{goldreich2020uniform}
Oded Goldreich.
\newblock The uniform distribution is complete with respect to testing identity to a fixed distribution.
\newblock In {\em Computational Complexity and Property Testing: On the Interplay Between Randomness and Computation}, pages 152--172. Springer, 2020.

\bibitem[Pan08]{paninski2008coincidence}
Liam Paninski.
\newblock A coincidence-based test for uniformity given very sparsely sampled discrete data.
\newblock {\em IEEE Transactions on Information Theory}, 54(10):4750--4755, 2008.

\bibitem[Val11]{valiant2011testing}
Paul Valiant.
\newblock Testing symmetric properties of distributions.
\newblock {\em SIAM Journal on Computing}, 40(6):1927--1968, 2011.

\bibitem[VV13]{VV13}
Gregory Valiant and Paul Valiant.
\newblock Instance-by-instance optimal identity testing.
\newblock {\em Electron. Colloquium Comput. Complex.}, {TR13-111}, 2013.

\bibitem[VV17]{vv17automatic}
Gregory Valiant and Paul Valiant.
\newblock An automatic inequality prover and instance optimal identity testing.
\newblock {\em SIAM Journal on Computing}, 46(1):429--455, 2017.

\end{thebibliography}

\end{document}